\documentclass[12pt, hidelinks]{article}
\usepackage[left = 0.8 in, right = 0.8 in, top = 1 in, bottom = 1.25 in]{geometry}              
\usepackage[utf8]{inputenc}
\usepackage{amsmath, amsthm, amssymb,xcolor,multicol,url}
\usepackage[sectionbib, round]{natbib}
\usepackage{graphicx,bbm}
\usepackage{booktabs,epstopdf,color}
\usepackage[space]{grffile}
\usepackage{lineno,comment}
\usepackage[plain,noend]{algorithm2e}
\usepackage{multirow}

\usepackage{hyperref}

\usepackage{amssymb}
\usepackage{subfig}

\usepackage{authblk}
\def\m{\mathcal}
\def\mb{\mathbb}

\def\wt{\widetilde}

\makeatletter
\renewcommand{\algocf@captiontext}[2]{#1\algocf@typo. \AlCapFnt{}#2} 
\def\@algocf@capt@plain{top}
\renewcommand{\algocf@makecaption}[2]{%
  \addtolength{\hsize}{\algomargin}%
  \sbox\@tempboxa{\algocf@captiontext{#1}{#2}}%
  \ifdim\wd\@tempboxa >\hsize
    \hskip .5\algomargin%
    \parbox[t]{\hsize}{\algocf@captiontext{#1}{#2}}
  \else%
    \global\@minipagefalse%
    \hbox to\hsize{\box\@tempboxa}
  \fi%
  \addtolength{\hsize}{-\algomargin}%
}
\makeatother

\def\T{{ \mathrm{\scriptscriptstyle T} }}

\DeclareMathOperator*{\argmin}{arg\,min}
\DeclareMathOperator*{\argmax}{arg\,max}

\def\m{\mathcal}
\def\mb{\mathbb}

\def\wt{\widetilde}

\def\T{{\mathrm{\scriptscriptstyle T} }}

\def\ELPD{{\mathrm{ELPD} }}
\def\EMM{{\mathrm{EMM} }}
\def\SE{{\mathrm{SE} }}
\def\Var{{\mathrm{Var} }}

\newcommand{\be}{\begin{equs}}
\newcommand{\ee}{\end{equs}}

\newcommand{\beginsupplement}{%
        \setcounter{lemma}{0}
        \renewcommand{\thelemma}{S\arabic{lemma}}%
        \setcounter{table}{0}
        \renewcommand{\thetable}{S\arabic{table}}%
        \setcounter{figure}{0}
        \renewcommand{\thefigure}{S\arabic{figure}}%
        \setcounter{section}{0}
        \renewcommand{\thesection}{S\arabic{section}}%
        \setcounter{equation}{0}
        \renewcommand{\theequation}{S.\arabic{section}.\arabic{equation}}%
        \setcounter{page}{1}
}

\begin{document}


\numberwithin{equation}{section}
\newtheorem{theorem}{Theorem}
\newtheorem{assumption}{Assumption}
\newtheorem{remark}{Remark}
\newtheorem{corollary}{Corollary}
\newtheorem{lemma}{Lemma}
\newtheorem{proposition}{Proposition}

\newtheorem{definition}{Definition}

\providecommand{\keywords}[1]
{
  \small	
  \textbf{\textit{Keywords---}} #1
}


\title{Robust probabilistic inference via a constrained transport metric}

\author[1]{Abhisek Chakraborty}
\author[1]{Anirban Bhattacharya}
\author[1]{Debdeep Pati}
\affil[1]{Department of Statistics, Texas A\&M University, 
College Station, TX, USA
}

\maketitle

\begin{abstract}
Flexible Bayesian models are typically constructed using limits of large parametric models with a multitude of parameters that are often uninterpretable.  In this article, we offer a novel alternative by constructing an exponentially tilted empirical likelihood carefully designed to concentrate near a parametric family of distributions of choice with respect to a novel variant of the Wasserstein metric, which is then combined with a prior distribution on model parameters to obtain a robustified posterior. The proposed approach finds applications in a wide variety of robust inference problems, where we intend to perform inference on the parameters associated with the centering distribution in presence of outliers.  Our proposed transport metric enjoys great computational simplicity, exploiting the Sinkhorn regularization for discrete optimal transport problems, and being inherently parallelizable. We demonstrate superior performance of our methodology when compared against state-of-the-art robust Bayesian inference methods. We also demonstrate equivalence of our approach with a nonparametric Bayesian formulation under a suitable asymptotic framework, testifying to its flexibility. The constrained entropy maximization that sits at the heart of our likelihood formulation finds its utility beyond robust Bayesian inference; an illustration is provided in a trustworthy machine learning application. 
\end{abstract}


\section{Introduction}
In most modeling exercises, our objective is limited to approximating a few key features of the true data-generating mechanism to ensure interpretable inference. It is often futile, if not misleading, to try to model small-scale and complicated underlying contaminating effects. Thus, the interplay between model adequacy and robustness  are fundamental areas of interest in model-based inference. Robust inferential methods \citep{huber2011robust}
possess an established and influential literature 
in statistics that has permeated many modern areas of research including differential privacy \citep{Dwork09differentialprivacy, doi:10.1080/01621459.2019.1700130, https://doi.org/10.48550/arxiv.2111.06578}, algorithmic fairness \citep{NEURIPS2020_37d097ca, https://doi.org/10.48550/arxiv.2105.11570}, noise-robust training of deep neural nets \citep{han2018coteaching, 9156647}, sequential decision making \citep{NIPS2010_19f3cd30, doi:10.1080/02331934.2019.1655738}, transfer learning \citep{Shafahi2020Adversarially}, quantification learning \citep{doi:10.1080/01621459.2021.1909599}, to name a few.   Bayesian procedures, however, being almost exclusively model-based,  inevitably fall prey to model mis-specification and/or perturbation of the data-generating mechanism -- an issue that exacerbates as sample size increases \citep{doi:10.1080/01621459.2018.1469995}. Credible intervals obtained from such parametric Bayesian models under model mis-specification may not have the desired asymptotic coverage \citep{kleijn2012bernstein}.
Non-parametric Bayes methods are routinely used to guard against such mis-specification, either by enlarging the parameter space to impart flexibility or by taking their limit to construct infinite-dimensional prior distributions \citep{muller2004nonparametric, 10.1214/009053606000000029, 10.2307/24310519}.  



Despite the success of non-parametric Bayes methods over the last few  decades; see \citet{muller2015bayesian} for a comprehensive review; the presence of a large number of non-identifiable parameters can be contentious, 
particularly when the interest is solely on simpler population features. For example, in many scientific applications,  standard parametric models are often preferred for convenience and ease of interpretation. This long-standing issue regarding the presence of a large number of uninterpretable parameters in non-parametric Bayes procedures has led to a proliferation of pseudo-likelihood-based approaches \citep{RePEc:eee:econom:v:115:y:2003:i:2:p:293-346, 2008,hooker2012bayesian, sandwitch,JMLR:v18:16-655, safebayes, Bernton_2019, doi:10.1080/01621459.2018.1469995} targeted towards specific parameters of interest. 
However, these approaches typically lack generative model interpretations making the calibration of the associated dispersion or temperature parameters  challenging \citep{10.1093/biomet/asx010, safebayes}. 

 An empirical likelihood (EL;  \citet{owen2001empirical}) offers an attractive solution which approximates the underlying distribution with a discrete distribution supported at the observed data points and obtains the induced maximum likelihood of the parameter of interest defined through estimating equations, by effectively profiling out the nuisance parameters. Exponentially titled Empirical Likelihood (ETEL) is a variant of this idea that minimizes the Kullback--Leibler  divergence of this discrete distribution with the empirical distribution of the observed data subject to satisfying the estimating equation.  One can import such a likelihood in a Bayesian framework to infer on parameters   \citep{10.2307/30042042, 10.1093/biomet/92.1.31,Chib2018,chib2021bayesian}. Interestingly, posterior credible intervals obtained from such  Bayesian procedure do have the correct frequentest coverage \citep{chib2021bayesian}, thereby effectively eliminating the longstanding criticism associated with parametric Bayesian inference under model mis-specification. 

Our goal here is to develop a flexible Bayesian  semi-parametric procedure that centers around a guessed parametric family $F_\theta$, without having to explicitly model aspects of the underlying data-generating mechanism we are not interested in. The task is similar to developing a robust Bayesian procedure \citep{RePEc:eee:econom:v:115:y:2003:i:2:p:293-346, 2008,hooker2012bayesian, doi:10.1080/01621459.2018.1469995, JMLR:v18:16-655} that allows departures from a parametric model to accommodate outlying observations. One may, alternatively, consider a non-parametric Bayes procedure \citep{Ferguson, teh2010dirichlet, 10.1214/aos/1176342871, 10.1214/aos/1176325623, 10.1214/aos/1024691240} where the parametric guess $F_\theta$ (with density $f_\theta$) assumes the role of the base measure, with the precision parameter controlling the extent of concentration around $F_\theta$. 
However, unlike these approaches, we desire our approach to be devoid of nuisance parameters, and that the inference is solely targeted to the parameter of interest while retaining the interpretation of a generative probability model. In a sense these are similar to the goals of EL (or ETEL) where one can simply consider the estimating equation $\mbox{E}[\partial \log f_\theta(X)/ \partial \theta]=0$ to infer about the parameter $\theta$. Observe that such a restriction enforces specific constraints on the moments of the distribution (e.g first moment if $F_\theta$ is Gaussian). Such a moment based constraint is agnostic to the tails of $F_\theta$ which could be significantly affected by outlying observations unless one allows the entire distribution to be constrained to lie in a neighborhood of $F_\theta$.  To that end, we propose a novel adaptation of ETEL by centering the discrete distribution $P$ of the data around $F_\theta$ using a suitable distance metric $\mbox{D}$ that encapsulates a more holistic discrepancy between the two distributions. More specifically,  we restrict $P$ within the neighborhood  $\mbox{D}[P , F_{\theta}]  < \varepsilon$, for some radius  $\varepsilon> 0$. In an inferential task, this framework provides a good balance between modeling flexibility by adaptively tuning $\varepsilon$ and interpretability, since we have the provision to invoke a non-parametric likelihood that concentrates around an  interpretable parametric guess, where the nuisance parameters are profiled out within the ETEL framework.

Naturally, a key ingredient in our proposal is the choice of the  metric $\mbox{D}$ that yields a non-trivial distance between $F_\theta$ and the empirical distribution on the observed data, and at the same time enjoys computational simplicity and straightforward multivariate extension. An equally important task is to have the provision to allow the user to select from relatively wider class of distributions $F_\theta$. Although having a fully flexible $F_\theta$ defeats the purpose of constructing a procedure devoid of nuisance parameters, we choose to work with elliptical mixture model (EMM)s which offers the user a sufficiently large class to choose from.  Because  $F_\theta$ is potentially absolutely continuous with respect to the Lebesgue measure and the empirical distribution on the observed data is discrete, it rules out many standard distances e.g Kullback--Leibler, Hellinger, total variation, $\chi^2$ etc. The $p$-Wasserstein metric \citep{Villani2003TopicsIO, 2019} provides a feasible choice, although it does not allow for a computationally efficient multivariate extension and the search for the optimal coupling becomes challenging.

To this end, we propose a novel adaptation of the $2$-Wasserstein metric by a {\em restriction} and an 
\emph{augmentation} scheme. In the {\em restriction} scheme, motivated by \citet{delon:hal-02178204}, we assume $F_\theta$ to be an EMM and adapt $\mbox{D}$ by further restricting the coupling measures to  the class of EMMs, which considerably reduce the computational cost and yet encompasses a rich class of
coupling measures. However, this renders the metric to depend only on the variance-covariance matrix of $F_\theta$ which ignores finer comparison in the tails. We address this in the {\em augmentation} scheme, where we augment the coupling measure with a product of univariate couplings. This tantamounts to adding a sum of univariate Wasserstein metrics to our adaptation, which effectively captures tail features. 
Further, the restriction scheme can   exploit a entropic regularization of discrete optimal transport \citep{le2019treesliced, cuturi2013sinkhorn}  that remains expressive, and computationally tenable even in multivariate cases.  Finally, we also developed a data-driven framework for tuning $\varepsilon$
exploiting an interplay of estimated expected log point-wise predictive density (ELPD) and its standard error \citep{JMLR:v17:14-540}. The resulting metric is termed ANDREW with the complete procedure named as $\mbox{D}$-BETEL. 

Having proposed a seemingly flexible method centering around a parametric family, it is of crucial importance to study to what degree the resulting posterior distribution of $\theta$ obtained from $\mbox{D}$-BETEL deviates from the same obtained from a typical non-parametric Bayes procedure centered around $F_\theta$.  In Section \ref{ssec:npBayes}, we address this by showing that the posterior of $\theta$ obtained from $\mbox{D}$-BETEL is similar to the posterior of $\theta$ obtained from a mixture-of-finite mixtures model \citep{doi:10.1080/01621459.2016.1255636} centered around $F_\theta$ under a suitable asymptotic framework. This is quite satisfying as it indicates that in performing inference on $\theta$, we are able to bypass the caveats associated with a fully non-parametric Bayes model, while still retaining its advantages in offering 
flexibility regarding the data generation process.  

The distributionally constrained entropic optimization, which forms the backbone of our likelihood formulation, is likely to be useful more broadly. To demonstrate its applicability beyond robust inference, we exploit the idea of re-weighting a distribution to increase proximity to another distribution in the context of demographic parity in algorithmic fairness problems in Section \ref{ssec:algorithmic_fairness}. We envisage more applications of this nature in the trustworthy AI paradigm, and wish to explore them in more details elsewhere.  

\section{The D-BETEL}\label{D_BETEL}
Let $\{F_\theta: \theta \in \Theta \subseteq \mb R^d\}$ be a parametric family of distributions. In the sequel, we develop a flexible Bayesian semi-parametric procedure that centers around this parametric family while allowing for flexible departures from it. 
Our approach draws inspiration from the Bayesian exponentially tilted empirical likelihood (Bayesian ETEL or BETEL; \citet{10.1093/biomet/92.1.31}) for moment conditional models that are specified by a collection of moment conditions $\mb{E}_{P}\big[g(X,\theta)\big] = 0$, where the expectation  $\mb{E}_{P}$ is taken with respect to the unknown generating distribution $P$,\  $g: \mb{R}^d\times\Theta \rightarrow \mb{R}^r$ is a vector of known functions, $\theta \in \Theta$ is the parameter of interest, and $0$ refers to a vector of $r$ zeros.  
Operating under a non-parametric Bayesian framework, \citet{10.1093/biomet/92.1.31} proposed a flexible prior on $P$ with an entropy-maximizing flavor.  
Under a specific asymptotic regime that allowed analytic marginalization of nuisance parameters describing the generative model, the corresponding {\em marginal} posterior distribution of $\theta$ given a random sample $x = (x_1,\ldots, x_n)^\T$ from $P$ was shown to approach a limiting distribution, called the BETEL posterior, given by
\begin{equation}\label{eqn:betelpost}
\pi_{\rm MCM}(\theta\mid x_1,\ldots,x_n) \ \propto\ \pi(\theta)\ L_{\rm MCM}(\theta)    
\end{equation} 
where the `likelihood' $L_{\rm MCM}(\theta)$ is called the exponentially tilted empirical likelihood,
\begin{equation}\label{eqn:betel}
  L_{\rm MCM}(\theta) = \bigg\{\prod_{i=1}^n w_{i} :  \argmax_{w}\prod_{i=1}^n w_{i}^{-w_{i}},\ w_i>0,\ \sum_{i=1}^n w_i = 1, \ \sum_{i=1}^n w_i g(x_i, \theta) = 0  \bigg\},
\end{equation} 
and $\pi(\cdot)$ denotes a prior distribution on $\theta$. Here and elsewhere, we use MCM as an acronym for {\em moment condition model}. 

The maximization problem in \eqref{eqn:betel} admits a non-trivial closed-form solution 
when the convex hull of $\cup_{i=1}^n g(x_i, \theta)$ contains the origin, leading to $L_{\rm MCM}(\theta) = \prod_{i=1}^n w_i^\star(\theta)$, with  
\begin{align*}
    w_{i}^{\star}(\theta) =  \frac{\exp[\lambda(\theta)^{\T}g(x_i, \theta)]}{\sum_{j=1}^n \exp[\lambda(\theta)^{\T}g(x_j, \theta)]}, \quad \lambda(\theta) = \argmin_{\eta} n^{-1} \sum_{i=1}^n \exp[\eta^{\T}g(x_i, \theta)],
\end{align*}
When the convex hull condition is not satisfied, $\pi_{\rm MCM}(\theta\mid x_1,\ldots,x_n)$ is set to zero. 
The maximization problem defining $\lambda(\theta)$ is convex, which leads to efficient computation of the ETEL likelihood $L_{\rm MCM}$, and the corresponding BETEL posterior $\pi_{\rm MCM}$ can be sampled using standard MCMC procedures. 
\citet{Chib2018} significantly contributed towards the theoretical underpinning of BETEL for moment conditional models; proving Bernstein–von Mises (BvM) theorem \& model selection consistency results under model mis-specification; and also numerically displayed its utility in wide-ranging econometric and statistical applications. 




The feature of BETEL most relevant to our purpose is that while the BETEL is motivated from a non-parametric Bayesian angle, it {\em operationally} avoids a complete probabilistic specification of the data-generating mechanism; the user only needs to specify a prior distribution on the parameter of interest $\theta$. In a similar spirit, our goal is to avoid a full non-parametric modeling of the data-generating distribution and only place a prior distribution on the (typically low-dimensional) parameter $\theta$ describing the centering model. A direct application of the ETEL to our setup is challenging as moment conditions describing parameters of general parametric models; especially those beyond exponential families; can be quite cumbersome or even unavailable in an analytically tractable form. Instead, our approach is to design a modified likelihood by constraining a weighted empirical distribution of the observed data $\nu_{w,x} :\,= \sum_{i=1}^n w_i\delta_{x_i}$ to be close to the parametric model $F_\theta$ with respect to a statistical metric. Specifically, we propose a likelihood function
\begin{equation}\label{eqn:wbetel}
  L_{\rm DCM}(\theta) :\,= \bigg\{\prod_{i=1}^n w_{i} :  \argmax_{w}\prod_{i=1}^n w_{i}^{-w_{i}},\ w_i>0,\ \sum_{i=1}^n w_i = 1, \  \mbox{D}[F_{\theta}, \nu_{w, x} ]\ \leq  \varepsilon \bigg\}  
\end{equation} 
where $\mbox{D}[\cdot,\cdot ]$ is an appropriate statistical distance, $\varepsilon > 0$ is a concentration parameter which controls fidelity to the centering model, and DCM is an acronym for {\em distributionally constrained model}. With this DCM likelihood, and a prior distribution on the parameter $\theta$, the corresponding posterior distribution is \begin{equation}\label{eqn:wbetelpost}
\pi(\theta\mid x_1,\ldots,x_n) \ \propto\ \pi(\theta)\ L_{\rm DCM}(\theta)    
\end{equation}
We refer to our formulation in  \eqref{eqn:wbetel} -- \eqref{eqn:wbetelpost} as the Bayesian ETEL subject to distributional constraint ($\mbox{D}$-BETEL). We show in Section \ref{ssec:npBayes} that the $\mbox{D}$-BETEL posterior arises organically from a non-parametric Bayes model by marginalization of the nuisance parameters specifying a mixing measure which has a mixture of finite mixtures (MFM; \citet{doi:10.1080/01621459.2016.1255636}) interpretation.   Further, in sub-section \ref{ssec:gen_reg}, the proposed methodology is extended to the regression setup, and detailed empirical study is performed to showcase its efficacy over standard Bayesian methodology and moment conditional models based on maximum likelihood equations.

The idea of centering the distribution of the observed data around a pre-specified parametric model is not new.  In fact, the Dirichlet process prior \citep{Ferguson, teh2010dirichlet} in Bayesian non-parametric is exactly designed to achieve this; other related approaches include \citet{10.1214/aos/1176342871, 10.1214/aos/1176325623, 10.1214/aos/1024691240}. Notably, the  traditional non-parametric priors are often accompanied by a large number of un-interpretable nuisance parameters that result in a computational overhead. On the contrary, $\mbox{D}$-BETEL directly obtain the marginal posterior parameter of interest, and  since the remaining nuisance parameters in the model are marginalised out, it enjoys improved interpretability. Moreover, while there is substantial literature on tuning the concentration parameter of the Dirichlet process mixture model \citep{EscobarWest, IZ, McAuliffe}, it still remains to be a  notoriously difficult task. $\mbox{D}$-BETEL involves a hyper-parameter $\varepsilon$, that too controls concentration around a parametric distribution. Since the nuisance parameters are effectively marginalized out in $\mbox{D}$-BETEL, we simply adopt a predictive approach  to devise a data-driven and principled tuning scheme. Moreover, we demonstrate in  Section \ref{ssec:npBayes} that a constrained empirical likelihood is asymptotically equivalent to a mixture model, centered at the pre-specified parametric density. Thus we are able to retain the advantages of non-parametric Bayes models while being devoid of nuisance parameters.

A key ingredient in our proposal is a  metric $\mbox{D}$ that yields a non-trivial distance between $F_\theta$ and weighted empirical distributions, is computationally convenient, and admits a seamless multivariate extension. Since $F_\theta$ is potentially absolutely continuous with respect to the Lebesgue measure, many standard distances e.g Hellinger, total variation, $\chi^2$ etc. are ruled out. The $p$-Wasserstein metric \citep{Villani2003TopicsIO} provides a feasible choice, however, it too is intractable in many multivariate cases we care about. To that end, assuming $F_\theta$ to be a elliptical mixture model (EMM), we further restrict the coupling measures to a carefully chosen sub-family, which considerably reduce the computational cost and yet encompasses a rich class of
coupling measures. Further, exploiting a entropic regularization of discrete optimal transport \citep{le2019treesliced, cuturi2013sinkhorn},
we propose a carefully tailored Wasserstein metric in Section \ref{ssec:andrew}  that remains expressive, and computationally tenable even in multivariate cases.

Before we move on, we undertake a closer look at the constrained entropy maximization problem at the core of our likelihood formulation in \eqref{eqn:wbetel}. Since
$$
\log \prod_{i=1}^n w_i^{-w_i}= \log n - \sum_{i=1}^n w_i \log (w_i/(1/n))
$$, it is apparent that solving the above  maximization problem is equivalent to finding the probability vector $(w_1,\ldots,w_n)$ that minimizes the Kullback--Leibler  divergence between the probabilities $w_1,\ldots,w_n$ assigned to each sample and the empirical probabilities $1/n,\ldots,1/n$, subject to the distance constraint $\mbox{D}[F_{\theta}, \nu_{w, x} ] <  \varepsilon$. Unlike the case for $L_{\rm MCM}$ \eqref{eqn:betel}, the optimization problem 
for $L_{\rm DCM}$ \eqref{eqn:wbetel} does not allow a closed form solution. Fortunately, we can access augmented Lagrangian methods \citep{auglag1, auglag2} and conic solvers \citep{TCFOCS} via the R interface \citep{RCore}  of  constrained non-linear optimization solvers (e.g. $\mbox{NLopt}$; \citet{nlopt} and $\mbox{CVX}$; 
\citet{gb08}). In particular, for fixed $\theta\in\Theta$ and $\varepsilon>0$, we can express the non-linear programming problem in \eqref{eqn:wbetel} in standard form as:
\begin{align*}
&\min_{w \in \m S^{n-1}}\ \sum_{i=1}^n w_i\log w_i, \quad
\text{subject to}\quad \mbox{D}[F_{\theta}, \nu_{w, x} ] \leq \varepsilon ,
\end{align*}
with domain $\mathcal{S}^{n-1} :\,= \{v \in \mb R^n: \ v_i>0,\ i =1,\ldots,n;\ \sum_{i=1}^n v_i = 1\}$ with a non-empty interior.  
The associated Lagrangian function $\mathcal{L}: \mb{R}^n\times\mb{R}\to\mb{R}$ is defined as
\begin{align}\label{eqn:dbetel:alt}
  \mathcal{L}(w, \lambda^*) =    \sum_{i=1}^n w_i\log w_i +  \lambda^{\star}\ \mbox{D}^2[F_{\theta} ,\nu_{w, x}], 
\end{align}
where $\lambda^\star$ is the Lagrange multiplier, and the Lagrange dual function $v:\mb{R}\to\mb{R}$ takes the form 
\begin{align}\label{wbetel_dual}
v(\lambda^\star) = \inf_{w\in\mathcal{S}} \mathcal{L}(w, \lambda^\star). 
\end{align}
This dual formulation enables us to access off-the-self augmented Lagrangian based optimization algorithms \citep{auglag1, auglag2} to compute $L_{\rm DCM}$.
Moreover, from an application stand-point, the formulation in \eqref{wbetel_dual} may render itself useful in a wide range of problems, where instead of 
the weighted empirical distribution $\nu_{w,x}$ and a parametric guess $F_\theta$, we consider a pair of distributions with at least one of them discrete, and the goal is to ensure that a re-weighted version of the discrete one is close to the other distribution, subject to certain conditions on the weights.
We present one such application in Section \ref{ssec:algorithmic_fairness} 
in the context of ensuring demographic parity in machine learning algorithms.
Before we put a more specific structure to $\mbox{D}$-BETEL, we shall discuss a key feature of our proposal, that we briefly alluded to earlier, in concrete terms. In particular, we demonstrate that one may view our proposed methodology as a non-parametric Bayes approach based on centering mixture models around a specific parametric family by establishing an intriguing asymptotic equivalence relationship between our framework and a hierarchical setup similar to the mixture of finite mixture (MFM) models \citep{doi:10.1080/01621459.2016.1255636}. Moreover, this enables us to formally identify $\mbox{D}$-BETEL as a generative model - a feature illusive to many existing pseudo-likelihood-based robust Bayesian methods.


\subsection{Non-parametric Bayes interpretation of D-BETEL}\label{ssec:npBayes}
In the following, we offer a concrete probabilistic justification to $\mbox{D}$-BETEL by building a Bayesian hierarchical generative model centered around $F_\theta$ so that the marginal posterior of $\theta$ converges in distribution to the $\mbox{D}$-BETEL posterior under a limiting environment motivated by \citet{10.1093/biomet/92.1.31}. This result is established in Theorem \ref{th:npbayes_equivalence}. 

In the following, we first describe a generative model for the data points $x_1, \ldots, x_n$ which closely mimics commonly used Bayesian nonparametric methods such as the mixture of finite mixture of Gaussians. The description  proceeds via a probability model for the independent $d$-variate observations $x_i$ conditional on its own set of parameters $\eta_i \in \mb R^d$, i.e., $x_i \mid \eta_i \overset{ind.} \sim f(\cdot \mid \eta_i)$ for $i = 1, \ldots, n$. To impart flexibility, the random effects $\eta_i$ are independently drawn from a common {\em mixing measure}
$P^{(N)}$ defined on $(\mathbb{R}^d, \mathcal{B}(\mathbb{R}^d))$, where $N$ is a positive integer involved in the description of $P^{(N)}$. This renders the marginal density of $x_i \mid P^{(N)}$ to be $\int f(x_i \mid \eta_i) P^{(N)}(d \eta_i)$, independently for $i = 1, \ldots, n$. The mixing distribution $P^{(N)}$ is parameterized 
through its associated nuisance parameters 
$\xi^\star = \big(k, b, \{\mu_h\}_{h=1}^k\big)$, where $k \in \mb N$, $b = (b_1, \ldots, b_k)$ where each $b_h$ is a positive integer subject to the constraint $\sum_{h=1}^k b_h = N$, and $\mu_h = (\mu_{h,1}, \ldots, \mu_{h, d})^{\T}\in \mb R^d$ for $h = 1, \ldots, k$. We induce a prior distribution on $P^{(N)}$ through a prior distribution on $\xi^\star$. To do so, 
we construct  a joint prior on $(\xi^{\star}, \theta)$ hierarchically by first specifying the marginal prior on the parameter of interest $\theta$, and then the conditional prior of $\xi^{\star} \mid \theta$ in terms of a $\theta$ dependent slice on the support of an unconditional distribution $\pi_{\infty, N}(\cdot)$ for $\xi^\star$. In essence,  $\xi^*$ act as a {\em bridge} between the data and the parameter of interest $\theta$ in the hierarchical formulation. 
This is where our modeling departs from a typical non-parametric Bayes model where $P^{(N)}$ is the object of inference and $\theta$ is viewed as a derived quantity from $P^{(N)}$. Instead, in our framework, $\theta$ retains its own identity and $P^{(N)}$ is viewed as an infinite-dimensional nuisance parameter. In other words, $(P^{(N)}, \theta)$ describes a semi-parametric object for inference, where $P^{(N)}$ is a flexible  probability measure, and $\theta$ is the parameter of interest.

We specify the details for each of these pieces from top down in the sequel. 
First, the distribution $f$ of the data given random effects is chosen to be an appropriate uniform distribution. Specifically, given $\tau > 0$, let
\begin{equation}\label{eqn:npbayes3}
\begin{aligned}
 x_i &\mid \eta_i, P^{(N)}\ \stackrel{\text{ind.}}\sim\ \prod_{l=1}^d \mbox{Uniform}(\eta_{i, l} -  \tau^{-1},\ \eta_{i, l} +  \tau^{-1}), \ i=1,\ldots, n,  \\
\eta_i &\mid P^{(N)} \sim  P^{(N)}.
\end{aligned}
\end{equation}
The uniform kernel is chosen for analytic tractability in ensuing calculations. We expect the main results to hold for more general kernels, albeit with additional technical challenges.
Next, for any set $A \in \mathcal{B}(\mathbb{R}^d)$, define $P^{(N)}(A)$ as  $P^{(N)}(A) = \sum_{h=1}^k \pi_h \delta_{\mu_h}(A)$, where conditional on $k$, the mixture weights $(\pi_1, \ldots, \pi_k)$ are constructed via normalised  counts $(b_1/N,\ldots, b_k/N)$. 
We specify distributions on the pieces to define an {\em unconditional distribution} $\pi_{\infty, N}(\cdot)$ for $\xi^*$, 
\begin{equation}\label{eqn:npbayes1}
\begin{aligned}
&(b_1,\ldots, b_k) \mid k \ \sim \ \mbox{Multinomial}(N; 1/k,\ldots, 1/k)\\
&\mu_h \mid k\ \stackrel{\text{i.i.d.}} \sim \ \mbox{H}^{(N)} , \ h=1,\ldots, k;\quad  k \sim p(k)\equiv \mbox{Geometric}(p);
\end{aligned}
\end{equation}
where  
$\mbox{H}^{(N)}$ is a suitably chosen $d$-dimensional ``base'' distribution; refer to \eqref{eqn:assumptions} for details. Given a draw of $k$, the $k$ atoms $\{\mu_h\}_{h=1}^k$ are drawn 
independently  from $\mbox{H}^{(N)}$, and the count vector $(b_1,\ldots, b_k)$ that yields the mixture weights is drawn from $\mbox{Multinomial}(N; 1/k,\ldots, 1/k)$; instead of  direct draws of the mixture weights from a Dirichlet distribution, commonly used in the finite dimensional version of the Dirichlet process \citep{ishwaran2002dirichlet,ishwaran2002exact}, or in the mixture of finite mixtures setup \citep{doi:10.1080/01621459.2016.1255636}.
The distributional specification $\pi_{\infty, N}$ for $\xi^\star$ 
in \eqref{eqn:npbayes1} induces a mixture of finite mixtures (MFM; \citet{doi:10.1080/01621459.2016.1255636}) for $P^{(N)}$ given by $P^{(N)} = \sum_{k=1}^\infty p(k) \big[\sum_{h=1}^k (b_h/N) \delta_{\mu_h}\big]$.  

Finally,  we construct a joint prior on $(\xi^{\star}, \theta)$ by first specifying a prior distribution $\pi(\cdot)$ on $\theta$, and then the conditional distribution of  $\xi^* \mid \theta$ by restricting the distribution $\pi_{\infty, N}(\cdot)$ to the slice 
$$
A_{\varepsilon,N}(\theta) :\, = \{\xi^{\star} \,:\, \mbox{D}(P^{(N)}, F_{\theta}) < \varepsilon \}
$$
defined on the support of $\xi^{\star}$, where the metric $\mbox{D}$ and the scalar $\varepsilon > 0$ are as in \eqref{eqn:betel}. Thus, 
\begin{align}\label{eqn:npbayes2}
 \pi_{\varepsilon, N}(\xi^{\star} \mid \theta) \ \propto \ \pi_{\infty, N}(\xi^{\star}) \ 1_{A_{\varepsilon, N}(\theta)}(\xi^{\star}),
\end{align}
that is, given a specific value of $\theta$, only draws from the unconditional prior $\pi_{\infty, N}$ are retained for which $P^{(N)}$  and $F_\theta$ are $\varepsilon$-close under the metric $\mbox{D}$. Figure \ref{fig:cartoon_npbayes} presents a schematic of the hierarchical model in equations \eqref{eqn:npbayes3}--\eqref{eqn:npbayes2}.

Combining the joint prior  $\pi_{\varepsilon,N}(\xi^\star, \theta)$ 
with the 
generative model in \eqref{eqn:npbayes3}, one obtains the joint posterior distribution $\pi_{\varepsilon, N}(\theta, \xi^\star\mid x_{1:n})$ of $(\theta, \xi^\star)$. A fully Bayesian analysis of the posterior of $\pi_{\varepsilon, N}(\theta, \xi^\star\mid x_{1:n})$ entails traversing the gigantic parameter space of $(\xi^\star, \theta)$ to simultaneously learn  $(P^{(N)}, \theta)$. Instead, motivated by \citet{10.1093/biomet/92.1.31},
we marginalise $\pi_{\varepsilon, N}(\theta, \xi^\star\mid x_{1:n})$ with respect to nuisance parameters $\xi^\star$ to obtain  the marginal posterior $\pi_{\varepsilon, N}(\theta\mid x_{1:n})$, that enables us to access targeted inference on the parameter of interest $\theta$.  In the remainder of this section, we shall operate in an asymptotic regime motivated by \citet{10.1093/biomet/92.1.31}, where we let the hyperparameters $\tau \equiv \tau(N), p \equiv p(N)$ and the base-measure $H^{(N)}$ to evolve with $N$. Under this environment, we show below that the marginal posterior $\pi_{\varepsilon, N}(\theta\mid x_{1:n})$ converges to \eqref{eqn:wbetelpost} as $N \to \infty$. \\

\begin{figure}
\centering
    \includegraphics[width=17cm, height = 4cm]{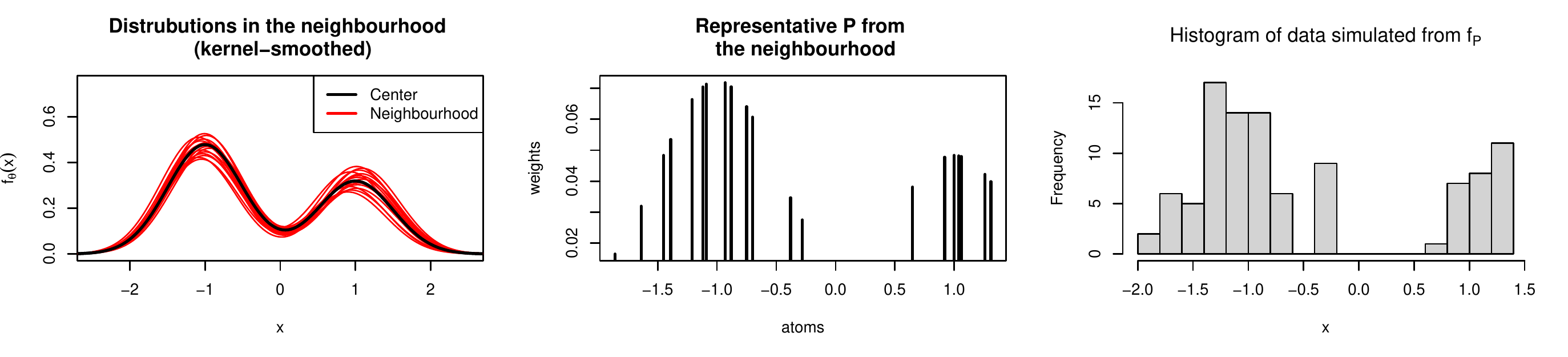} 
    \caption{\emph{\textbf{Left panel.} the pdf $f_{\theta}(x)$ corresponding to the centering distribution $F_{\theta}\equiv 0.6\times\mbox{N}(-1, 0.5^2) + 0.4\times\mbox{N}(1, 0.5^2)$ in black and representative discrete distributions in a $\mbox{D}$-neighborhood (with D chosen as  $\mbox{W}_{\rm AR}$ introduced in Section \ref{ssec:andrew}) of $F_{\theta}$ in red after kernel smoothing;  \textbf{Middle panel.} one particular  $P= \sum_{h=1}^k \pi_h \delta_{\mu_h}$ with $k=20$ in the $\mbox{W}_{\rm AR}$-neighbourhood of $F_{\theta}$ with $W_{\rm AR}^2(P, F_\theta) = 2.5$; \textbf{Right panel.} histogram of a random sample of size $100$ drawn from $f_P(x) = \int f(x \mid \eta) P(d \eta)$ with $\tau = 10^{2}$ in equation \eqref{eqn:npbayes3}}.\label{fig:cartoon_npbayes}}%
\end{figure}

\noindent \textbf{Assumptions:}
We consider the following constructions on $p=p(N), \tau = \tau(N)$, and the base measure $\mbox{H}^{(N)}$.
First, we introduce sequences $(K_N, L_N) = (N^{\alpha}, N^{\beta}) $ with $\alpha\in \mb{R}^{+}, \beta\in \mb{Z}^{+}, \alpha > \beta$ 
. Now, let $\mathcal{X}^N$ be the mid-points of the cells of an uniform $(K_N + 1)(K_N + 1)\ldots(K_N + 1)$ grid on the hyper-cube $\big[-L_N/2, L_N/2 \big]^{d}$, with the total number of grid-points $M_N = (K_N+1)^{d}$ and the grid spacing $2\rho_N = L_N/K_N $. We then set 
\begin{align}\label{eqn:assumptions}
&\mbox{H}^{(N)}\equiv\ \mbox{Uniform}(\mathcal{X}^N), \quad
\tau(N) = \frac{1}{\rho_N}, \quad  p(N)=1- \frac{1}{AM_N^{N+1}};
\end{align}
for some $A>0$. 

The assumptions above place specific structures on $p=p(N),\ \tau = \tau(N)$ and the base measure $\mbox{H}^{(N)}$ that make aspects of the model-prior increasingly diffuse while also yielding tractable analytic calculation of the marginal posterior of $\theta\mid x_{1:n}$. In particular, under the above construction, $p = p(N)$ converges to $1$, rendering increased penalty on the number of components $k$, and $\tau = \tau(N)$ diverges to $\infty$, allowing the support of the uniform kernel to shrink, as $N \to \infty$.
Moreover, $\mbox{H}^{(N)}$ is a flat  prior on a compact set which expands with $N$. These assumptions share key structural commonalities with the assumptions for the main result in \citet{10.1093/biomet/92.1.31} 
In particular, \citet{10.1093/biomet/92.1.31} also considers sequences $(K_N, L_N)$ such that $K_N, L_N\to\infty$ and $L_N/K_N\to 0$ as $N\to\infty$. We make further specific choices of $K_N, L_N$ in order to simplify the expressions of various posterior quantities in the hierarchical specification (\eqref{eqn:npbayes3}--\eqref{eqn:npbayes1}), and provide rigorous asymptotic results. This ensures improved clarity in the proofs of the next two theorems while maintaining generality of the arguments.

\begin{theorem}\label{th2:lemma2}
Fix sample size $n$. Under the hierarchical specification in \eqref{eqn:npbayes3}--\eqref{eqn:npbayes2} and assumptions in  \eqref{eqn:assumptions}, as $N\to\infty$, the marginal posterior probability $P(k=n\mid x_{1:n})\to 1 $ almost surely.
\end{theorem}
Theorem \ref{th2:lemma2} is an important building block for our main result in Theorem \ref{th:npbayes_equivalence} below. 
As noted earlier, Assumption \eqref{eqn:assumptions} imply that under the hierarchical specification in \eqref{eqn:npbayes1}, 
(a) the number of atoms $\mu_h$s drawn from the base measure $\mbox{H}^{(N)}$ is controlled by the prior on $k$ that strongly encourages small support, i.e smaller number of unique atoms, and  (b) the atoms $(\mu_1,\ldots, \mu_k)$ take values in a uniform grid $\mathcal{X}^N\subset\mb{R}^d$ that expands to $\mb{R}^d$  while getting increasingly dense as $N\to\infty$. Together with \eqref{eqn:npbayes3}--\eqref{eqn:npbayes1}, (a) ensures that the $P(k=n\mid x_{1:n}) = 1$  almost surely and $\eta_{1:n}$ 
collapses on the observed sample as $N\to \infty$, and  (b) ensures that the hierarchical model describes an extremely flexible generative model. The flexible yet discrete nature of $\mathcal{X}^N\subset\mb{R}^d$ critically simplifies the arguments in the proof of Theorem \ref{th:npbayes_equivalence} below.

\begin{theorem}\label{th:npbayes_equivalence}
Fix the concentration parameter $\varepsilon>0$ and sample size $n$.  Suppose $\mbox{H}^{(N)},\ p$ and $\tau$ satisfy the assumptions stated above as $N \to \infty$. Then the marginal posterior $\pi_{\varepsilon, N}(\theta\mid x_{1:n})$ defined after \eqref{eqn:npbayes3}--\eqref{eqn:npbayes2} converges point-wise in $\theta$ to the $\mbox{D}$-BETEL posterior in equation \eqref{eqn:wbetelpost},
\begin{align*}
|\pi_{\varepsilon, N}(\theta\mid x_{1:n}) - \pi(\theta\mid x_{1:n})|\to 0 \quad \text{as}\quad N\to \infty.
\end{align*}
An application of Scheffe's theorem \citep{Resnick} yields, 
\begin{align*}
    \| \pi(\cdot \mid x_{1:n}) - \pi(\cdot \mid x_{1:n}) \|_{\rm TV} :\,= \frac{1}{2} \int_{\theta}|\pi_{\varepsilon, N}(\theta\mid x_{1:n}) - \pi(\theta\mid x_{1:n})| d\theta \to 0 \quad \text{as}\quad N\to \infty.
\end{align*}
\end{theorem}

The proof of Theorems \ref{th2:lemma2} and \ref{th:npbayes_equivalence} and the required results are deferred to Section \ref{ssec:th1th2_proof} of the supplementary document. Instead, we discuss the key takeaway messages that the theorems expose. Perhaps most importantly, Theorem \ref{th:npbayes_equivalence} enables us to formally recognize $\mbox{D}$-BETEL as a limiting non-parametric Bayes posterior. Moreover, the formulation in equations \eqref{eqn:npbayes3}--\eqref{eqn:npbayes2} allows decoupling of the parameter of interest $\theta$ and the potentially infinite-dimensional nuisance parameter $\xi^\star$, tying them together through the slicing. The decoupling enables us explicitly quantify prior information on $\theta$ through prior $\pi(\theta)$ in our implicitly specified models, and the slicing potentially provides an efficient scheme to navigate the enormous $(\xi^{\star},\theta)$ space. We hope that this framework of providing the parameter of interest $\theta$ with its own identity in implicitly specified models and making the nuisance parameter $\xi^{\star}$ dependent on $\theta$ opens up possibilities of developing non-parametric Bayes models tailored to perform inference on specific parameters going forward.

It is also important to point out that the asymptotic regime allows the effect of the prior $P^{(N)}$ to get washed away for any finite $n$ so that the resulting posterior distribution of $P^{(N)}$ converges to the empirical distribution subject to the distributional constraint. As a consequence, the induced posterior distribution on $\theta$ converges to the posterior obtained from the empirical likelihood.  Although in principle, the mixture of finite mixture prior for $P^{(N)}$ can be replaced any other distribution whose effect is allowed to weaken under the asymptotic regime, the MFM construct is a natural choice that allows the $P^{(N)}$ to asymptotically be degenerate at the observed data points.


Another key revelation from the presentation above is the fact that the hyper-parameter $\varepsilon$ bears clear similarity to the concentration parameter in a Dirichlet process \citep{Ferguson, teh2010dirichlet}, as it determines how tightly $\nu_{w,x}$ sits around $F_{\theta}$ with respect to $\mbox{D}$. Since the $\mbox{D}$-BETEL formulation enjoys concrete probabilistic interpretation, we are able to provide a principled guideline for hyper-parameter $\varepsilon$. To that end we recall that, given $x_1,\ldots,x_n$, the Bayesian leave-one-out estimate of out-of-sample predictive fit \citep{JMLR:v17:14-540} is 
\begin{align*}
 \ELPD_{\varepsilon} =   \sum_{i=1}^n \ \log \pi(x_i\mid x_{-i})\quad \text{where}\quad  \pi(x_i\mid x_{-i}) = \int \pi(x_i\mid\theta)\ \pi(\theta\mid x_{-i})\ d\theta   
\end{align*}
is the leave-one-out predictive density given the data without the $i$-th data point, and the corresponding standard error is 
\begin{align*}
\SE[\ELPD_{\varepsilon}] = \sqrt{n}\ \sqrt{\Var[\ \log \pi(x_1\mid x_{-1}),\ldots, \log \pi(x_n\mid x_{-n})\ ]}.    
\end{align*}
 When $\varepsilon$ is too large, the distance-based restriction does not kick in and estimated $\SE[\ELPD_{\varepsilon}]$ is close to $0$. So, we consider a decreasing sequence of $\varepsilon$ values, say $\varepsilon_1,\ldots,\varepsilon_h$, such that $\varepsilon_i > \varepsilon_j \ \forall \ 1 \leq i<j\leq h$. A general strategy to select the sequence is to first consider a grid over powers of $2$ and then use a finer grid in the interval where $\ELPD_{\varepsilon}$ undergoes steep change. Suppose $\varepsilon_{h_0}$ is the largest value of $\varepsilon$ for which the distance-based restriction is active. Then, our estimate of the model parameter $\theta$ is 
\begin{align}\label{eqn:ma}
  \hat{\theta}_{\rm MA} = \sum_{i=h_{0}}^h\kappa_i\hat{\theta}_i ,\quad \text{with}\quad \kappa_i = \frac{\exp(-\ELPD_{\varepsilon_{i}})}{\sum_{j=h_{0}}^h \exp(-\ELPD_{\varepsilon_{j}})}
\end{align}
where $\hat{\theta}_i$ and $\exp(-\ELPD_{\varepsilon_{i}})$ are the parameter estimate and estimated $\ELPD$  at $\varepsilon=\varepsilon_i$ respectively. 
From the definition of ELPD, we can interpret it as a measure of the extent of unequal weighting of the observations. In the presence of contamination, our approach of selecting the hyper-parameter promotes  unequal weighting of the observations to ensure  -- under weighting of outlying observations, and over-weighting observations around the ``center". This inbuilt mechanism of ensuring immunity against outliers while maintaining  a valid generative model interpretation is what sets our method apart from lot of the existing pseudo-likelihood based approaches. Finally, it is also important to point out that, although $\hat{\theta}_{\rm MA}$ in equation \eqref{eqn:ma} is calculated via an weighted average, in practice $\hat{\theta}_{\rm MA}$ and the associated HPD set typically degenerate to those corresponding to a handful of values of $\varepsilon$. Thus this procedure inherits the generative model interpretation of $\mbox{D}$-BETEL.  

The rest of this section is devoted to completing the specifications in our proposed $\mbox{D}$-BETEL formulation in \eqref{eqn:wbetel}-\eqref{eqn:wbetelpost}. To be precise, we discuss what would be a judicious choice of the centering parametric family $F_{\theta}$, together with the distance metric $\mbox{D}$ between $F_{\theta}$ and $\nu_{w,x}$, that remains computationally feasible across the ensuing applications.

\subsection{ANDREW: an AugmeNteD \& REstricted  Wasserstein metric}\label{ssec:andrew}
To offer flexibility in the centering mechanism, we seek that the family of centering distributions describe a large and expressive class of models. To that end, we suggest employing Elliptical Mixture Models (EMM) as a general choice of the centering distribution $F_{\theta}$. Given a vector $m \in \mb R^d$, a positive (semi-)definite scale matrix $\Sigma \in \mb R^{d \times d}$, and a generator function $h: (0, \infty) \to (0, \infty)$, the elliptical distribution $\mbox{ED}_{h}(m,\Sigma)$ is defined to be the distribution with characteristic function 
\begin{align}\label{eqn:ch_EMM}
    t\rightarrow\exp(i t^\T \mu)\ h(t^\T\Sigma t), \quad t \in \mathbb{R}^{d}.
\end{align}
Recall that a multivariate Gaussian distribution $\mbox{N}_d(m,\Sigma)$ has characteristic function 
$\exp(i t^\T\mu)\ h(t^\T\Sigma t)$ for $t \in \mb{R}^d$ where $h(z) = \exp(-z/2)$ for $z>0$. Elliptical distributions \citep{muirhead2005aspects} 
allow a wider class of positive functions $h$. 
A mixture of such elliptical distributions, $\sum_{k=1}^{K_{0}} s_{0k}\ \mbox{ED}_h(m_k,\Sigma_k)$, provides a flexible tool for statistical modeling \citep{Cambanis1981OnTT, 10.2307/4616956} and probabilistic embedding of complex objects \citep{muzellec2019generalizing, le2019treesliced}. Consequently, elliptical mixture models (EMM) serve as an attractive candidate for our parametric centering family. With this choice, we now discuss the construction of a novel statistical distance $\mbox{D}[\cdot ,\cdot ]$ that is expressive and easy to calculate between the centering elliptical mixture model and the weighted empirical distribution of the observed data.

As discussed earlier, we require a statistical distance $\mbox{D}$ that (i) returns a non-trivial distance between a discrete and a continuous distribution. For the ensuing applications we also require that $\mbox{D}$: (ii) allows a straightforward multivariate extension; (iii) is computationally feasible; and (iv) effectively captures the tail behavior of the distributions. The requirement (i) itself rules out the applicability of many popular statistical distances/divergences like the Kullback--Leibler  divergence, Hellinger distance, total variation distance, $\chi^2$ distance, etc. The Cramer--von Mises metric on $\mb R$ satisfies (i), (iv), but its multivariate extension is not immediate. The $p$-Wasserstein metric \citep{Villani2003TopicsIO} satisfies (i), (ii), and (iv), and is an attractive candidate. However, it remains computationally challenging in multivariate examples we care about. This motivates the need for a specialized adaptation of the Wasserstein metric. To that end, we recall some relevant facts about the $p$-Wasserstein metric first.    
\begin{definition}
For $p\geq 1$, the Wasserstein space $\mb{P}_p(\mb{R}^d)$ is defined as the set of probability measures $\mu$ with finite moment of order $p$, i.e $\{\mu \,: \, \int_{\mb{R}^d} \left\lVert x\right\rVert^p d\mu(x) < \infty\}$, where $\left\lVert \ \cdot \ \right\rVert$ is the euclidean norm on $\mb{R}^d$.
\end{definition}
\begin{definition}
For $p_0, p_1 \in \mb{P}_p(\mb{R}^d)$, let $\pi(p_0, p_1) \subset \mb{P}_p(\mb{R}^d\times\mb{R}^d) $ denote the subset of joint probability measures (or {\em couplings}) $\nu$ on $\mb{R}^d\times\mb{R}^d$ with marginal distributions $p_0$ and $p_1$, respectively. Then, the $p$-Wasserstein distance $W_p$ between $p_0$ and $p_1$ is defined as 
\begin{align}\label{eq:W_def}
 W_{p}^{p}(p_0, p_1)
= \inf_{\nu\in \pi(p_0,p_1)} \int_{\mb{R}^d\times\mb{R}^d}\
\left\lVert y_0 - y_1 \right\rVert^p\ d\nu(y_0,y_1).  
\end{align}
\end{definition}

\begin{figure}[!htb]
\includegraphics[width=1\textwidth]{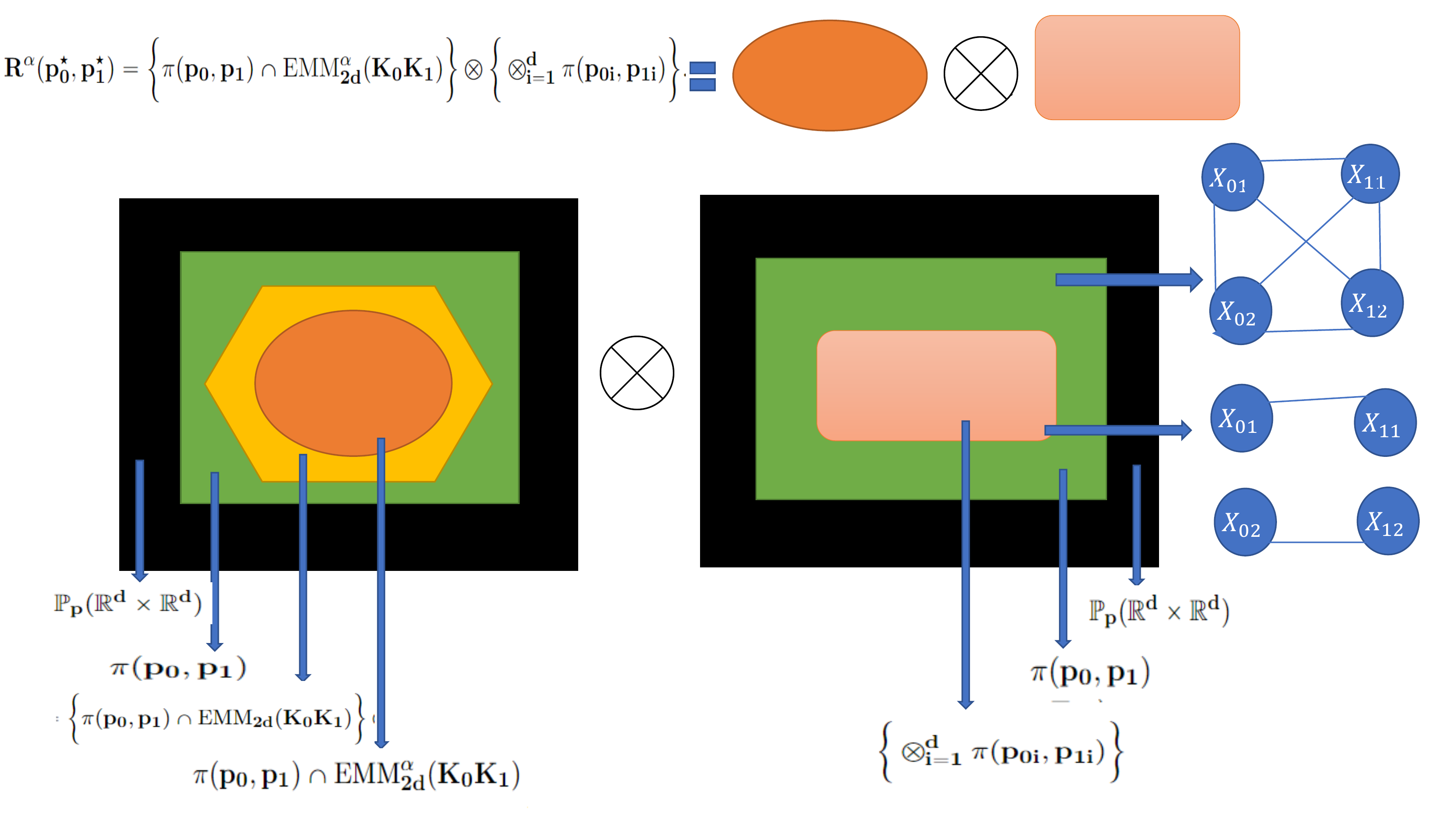}
\caption{\emph{The augmentation and restriction scheme to construct $R^{\alpha}(p^{\star}_0,p^{\star}_1)$; left: construction of $\big\{\pi(p_0,p_1)\cap \EMM^{\alpha}_{2d}(K_0 K_1)\big\}$ ; right: construction of $\big\{\otimes_{i=1}^d \pi(p_{0i},p_{1i})\big\}$ with $d=2$.} 
}\label{fig:andrew_cartoon}
\end{figure}

In the one-dimensional case,  $W_p$ has a closed-form expression as the $L_p$ distance between the corresponding quantile functions \citep{noauthororeditor}.
However, such closed-form expressions for $d\geq 2$ are unavailable except for a few special cases, and numerical approximations, while available, are computationally expensive. \citet{cuturi2013sinkhorn} demonstrated that regularizing the minimization (or transport) problem \eqref{eq:W_def} with an entropic penalty provides an efficient numerical approach to compute the Wasserstein metric in discrete cases. There also has been a recent line of work \citep{delon:hal-02178204, BionNadal2019OnAW} modifying the 2-Wasserstein metric via restricting the class of coupling measures to a carefully chosen sub-family, which considerably reduces the computational cost and yet encompasses a rich class of coupling measures. In particular, \citet{delon:hal-02178204} proposed a modified Wasserstein metric with $L_2$ cost between two Gaussian mixture models in $\mb{R}^d$ via restricting the class of coupling measures to all possible Gaussian mixtures in $\mb{R}^{2d}$, and derived a computationally convenient discrete formulation for this metric. Moreover, the authors, although without a formal proof,  indicate that the proposed distance could be extended to other types of mixtures if they satisfy a marginal consistency property and an identifiability property. A rigorous proof is presented in Lemma \ref{th1:lemma4}. 
Although Gaussian mixtures already describe a rich class of models, if we directly extend their proposal to a more versatile class of models, the proposed metric will be incapable of capturing key differences  of interest, i.e., tail properties, between probability distributions ; please refer to the discussions following the statement of Theorem \ref{th3} for details. 
In what follows, focusing on the class of identifiable elliptical mixture models in equation \eqref{def:indentifiable_emm},  we describe a new and carefully crafted strategy based on augmentation succeeded by restriction in the space of coupling measures that yield a transport metric ANDREW that not only automatically inherits computational tractability as in  \citet{delon:hal-02178204}, but also remains expressive and is capable of accessing improved computational algorithms based on an entropic regularization of the discrete optimal transport. 
In essence, our novel strategy presents a general recipe for devising increasingly expressive transport metrics and describing a corresponding modified class of couplings, of which ANDREW introduced next is a specific example.

With the requirements (i)-(iv) in mind, we 
place the discussion in the previous paragraph in concrete terms and  
propose a modified 2-Wasserstein metric (ANDREW), denoted $W_{\rm AR}$, between elliptical mixture distributions. 
To that end, we {\em augment} the class of coupling measures $\pi(p_0,p_1)\subset \mb{P}_p(\mb{R}^d\times\mb{R}^d)$ into a class of coupling measures  $\pi(p^{\star}_0,p^{\star}_1)\subset\mb{P}_p(\mb{R}^{2d}\times\mb{R}^{2d})$, and then {\em restrict} $\pi(p^{\star}_0,p^{\star}_1)$ to a carefully chosen sub-class of couplings. We describe the details below. 

\begin{definition}
Let $p_0, p_1 \in \mb{P}_{p}(\mb{R}^d)$ with $p_j = \sum_{k=1}^{K_j} s_{jk} \mbox{ED}_h(m_{jk},\Sigma_{jk}), (j = 0, 1)$. Next, we shall consider an augmentation followed by a restriction scheme as follows:\\
(a) Augmentation:
Define probability distribution $p_{0}^\star \in\mb{P}_2(\mb{R}^{2d})$ as 
\begin{align*}
p_{0}^\star :\,= p_0 \otimes \widetilde{p}_0, \ \text{with}\ \widetilde{p}_0 :\,= p_{01} \otimes \ldots \otimes p_{0d}    
\end{align*}
and $p_{0i}$ the $i$th marginal of $p_0$. Clearly, if $X_0 \sim p_0$, and $\widetilde{X}_0$ independent of $X_0$ is distributed as $\widetilde{p}_0$, then $X_0^{\star} = (X_0, \widetilde{X}_0)^{\T}\sim p_{0}^\star$. Similarly, define $p_1^\star$.
By construction we have
\begin{align*}
 \pi(p^{\star}_0,p^{\star}_1) = \pi(p_0,p_1)\otimes \pi(\widetilde{p}_0,\widetilde{p}_1)  = \pi(p_0,p_1)\otimes \big\{\otimes_{i=1}^d \pi(p_{0i},p_{1i})\big\}.   
\end{align*}
(b) Restriction: 
Suppose
\begin{equation*}
 \EMM_{2d}(K_0,K_1) = \bigg\{ \sum_{k,l} \pi_{kl}\ {\rm ED}_{h}(m_{kl},\Sigma_{kl}): \pi_{kl}\geq 0,\ \sum_{k,l} \pi_{kl} = 1\bigg\}   
\end{equation*}
denote the collection of all ($K_0\times K_1$)-component mixture of identifiable elliptical distributions. Define a subset $\EMM^{\alpha}_{2d}(K_0, K_1)$ of $\EMM_{2d}(K_0,K_1)$ by imposing the entropic restriction 
\begin{equation*}
D_{\rm KL}\big[\Pi\ || \ s_0 s_{1}^{\T}\big]\leq \alpha\ \text{where}\ \Pi = ((\pi_{kl})) \in \mb R^{K_0 \times K_1}
\end{equation*}
is the joint probability matrix of the mixture weights, and $s_0, s_1$ are the respective marginals, i.e., $\Pi 1_{K_1} = s_0, \Pi^{\T}1_{K_0} = s_1$.
Finally, define a collection of couplings $R^{\alpha}(p^{\star}_0,p^{\star}_1) \subset \pi(p^{\star}_0,p^{\star}_1)$ as 
\begin{align*}
R^{\alpha}(p^{\star}_0,p^{\star}_1) = \big\{\pi(p_0,p_1)\cap \EMM^{\alpha}_{2d}(K_0, K_1)\big\}\otimes \big\{\otimes_{i=1}^d \pi(p_{0i},p_{1i})\big\}.   
\end{align*}
Refer to Figure \ref{fig:andrew_cartoon} for a schematic representation of the proposed augmentation and restriction strategy.
With these notations in place, we define ANDREW as
\begin{equation}\label{eqn:andrewdef}
 W_{\rm AR}^{2}(p_0, \ p_1) = \inf_{\nu\in R^{\alpha}(p^{\star}_0,p^{\star}_1)} \mb{E}_{\nu}||X_{0}^{\star} - X_{1}^{\star}||^2. \end{equation}
\end{definition}


To cater to our original goal of centering the $\mbox{D}$-BETEL around an EMM, we now present a simplified form of $W_{\rm AR}$ for the case when one of $p_0, p_1$ is discrete.
\begin{theorem}\label{th3}
Suppose $p_{0}\equiv \sum_{k=1}^{K_0} s_{0k} \mathrm{ED}_h(m_{0k},\Sigma_{0k}),\
p_{1}\equiv {\sum_{k=1}^{K_1} s_{1k} \delta_{m_{1k}}}$, and $M = ((||m_{0k}-m_{1l}||^2 )) \in \mathbb{R}^{K_0 \times K_1}$ be the quadratic cost matrix. Then, there exists $\lambda_{\alpha}$ depending on $\alpha$ such that 
{\smaller \begin{align*} 
W_{\rm AR}^2(p_0, p_1) =
\inf_{\pi\in\pi(s_0,s_1)}\bigg[ \big<\Pi, M\big>  - \frac{1}{\lambda_{\alpha}} H(\Pi)\bigg] + \nu_{h}\sum_{k=1}^{K_0} s_{0k} \mathrm{tr}(\Sigma_{0k})+ \sum_{k=1}^d \int_{0}^1 \big(F_{0k}^{-1}(z) - F_{1k}^{-1}(z)\big)^2 dz
\end{align*}}
where $\big<\Pi, M\big> = \mathrm{tr}(\Pi^{\T}M)$, $H(\Pi) = -\sum_{k,l} \pi_{kl}\log \pi_{kl}$ and $F_{jk}^{-1}(\cdot)$ is the quantile function of $X_{jk}$.
\end{theorem}

We defer the proof and a cascade of required Lemmas to Section \ref{ssec:th3} of the supplementary document and make some remarks about ANDREW here. Importantly, the expression above is completely tractable and computationally feasible. The entropic regularization term in $W_{\rm AR}$ makes the discrete optimal transport problem strictly convex, and consequently, it can access linear convergence via Sinkhorn’s fixed point iterations \citep{cuturi2013sinkhorn}. On the other hand, traditional simplex or interior-point methods for an $n\times n$ unregularized optimal transport scales at least in $O(n^3\log n)$. Next, $W_{\rm AR}$ effectively captures the tail behavior of the distributions. Had we restricted the class of coupling measures $\pi(p_0,p_1)$ to $\EMM_{2d}(K_0,K_1)$ instead, then
\begin{align*}
 \mbox{MW}^2_2(p_0, p_1) :\,=  \inf_{\nu\in \{\pi(p_0,p_1)\cap \EMM_{2d}(K_0 K_1)\}} \mb{E}_{\nu}||X_{0}^{\star} - X_{1}^{\star}||^2  = \inf_{\pi\in\pi(s_0,s_1)} \big<\Pi, M\big> + \nu_h \sum_{k=1}^{K_0} s_{0k} \mbox{tr}(\Sigma_{0k})
 \end{align*}
only depends on first and second-order moments, and fails to capture the tail behavior of the distributions. For example, let $p_{0}\equiv \sum_{k=1}^{K_0} s_{0k} t_{\eta}(m_{0k},\Sigma_{0k}),\ p^{\prime}_{0}\equiv \sum_{k=1}^{K_0} s_{0k} t_{\eta^{\prime}}(m_{0k},\Sigma_{0k}^{\prime}),\ p_{1} \equiv {\sum_{k=1}^{K_1} s_{1k} \delta_{m_{1k}}} $ and set $\eta^{\prime} = \eta/m,\ \Sigma_{0k}^{\prime} = \frac{\eta - 2m}{\eta - 2}\Sigma_{0k}$ for some $ m\in\mb{Z}^{+} $ such that the variances of the multivariate t-distributions $t_{\eta}(m_{0k},\Sigma_{0k})$ and $t_{\eta^{\prime}}(m_{0k},\Sigma_{0k}^{\prime})$ match for $k = 1,\ldots,K_0$. Then $\mbox{MW}^2_2(p_0, p_1) = \mbox{MW}^2_2(p_{0}^{\prime}, p_1)$. Since the expression of $W_{\rm AR}$ additionally involves the marginal quantiles, it is capable of capturing the difference in the tail due to the different d.f. of the $t$. We believe the flexibility and the computational simplicity of our novel Wasserstein metric may render itself useful in many optimal transport-based machine learning applications, beyond what we discuss here; see the discussion section for some specific application domains.

We now have all the necessary ingredients for $\mbox{D}$-BETEL, and we illustrate the proposed methodology in a number of specific applications. All the examples in the following section use $\mbox{D}$-BETEL in \eqref{eqn:wbetel} with our proposed transport metric ANDREW in \eqref{eqn:andrewdef}.


\section{Robust Bayesian inference}\label{Robust_Bayes}
\subsection{Model based clustering}
Motivated by the model based clustering example in \citet{doi:10.1080/01621459.2018.1469995}; see also \citet{cai}; we generate data from a bivariate skew-normal distribution \citep{SkewNormal} with pdf $f(x) = 2\phi(x) \Phi(\alpha^{\T} x),\ x \in \mb{R}^{2}$, with the two-dimensional skewness parameter $\alpha \neq (0, 0)$ to imitate a situation where the underlying true distribution is bivariate normal in the presence of mild contamination. We wish to demonstrate that $\mbox{D}$-BETEL is resistant to presence of mild perturbations in the data generating mechanism and can adequately describe the above set up with a bivariate normal centering, without resorting to more complex centering distributions.  We shall showcase all our tools in action on this simple example, and skip some of these details in later sections. Throughout this example, for the purposes of model comparison via marginal likelihood, we follow the approach in \citet{doi:10.1198/016214501750332848} to approximate the log marginal density $\log m(x\mid  \mb{M})$ of a model $\mb{M}$ via
\begin{equation*}
\log m(x\mid \mb{M}) = \log f(x\mid \mb{M},\theta^{*}) + \log \pi(\theta^{*}\mid \mb{M}) - \log \pi(\theta^{*}\mid x,\mb{M}), \end{equation*}
where $\log f(x\mid \mb{M},\theta^{*})$  and $\log \pi(\theta^{*}\mid \mb{M})$ are respectively  the log-likelihood and log prior of the model $\mb{M}$ at $\theta^{*}$, preferably  a high-density point. 

We generate data from a bivariate skew-normal distribution with varying value of skewness parameter $\alpha \neq (0, 0)$. We choose sample sizes $n \in \{100,\ 200,\ 300,\ 500\}$; and set $\alpha= (2.5, 2.5)^\T,\, (3.0, 3.0)^\T,\ (3.5, 3.5)^\T$ -- giving us $12$ simulation set-ups in total. First, we compare the following two fully parametric models: (i) $\mb{M}_{1}$, which models the data as independent draws from $\mbox{N}_2(\mu, \Sigma)$, and imposes a diffuse $\mbox{N}_2(0, 10^3 \mbox{I}_2)$ prior on $\mu$ and $\mbox{Wishart}_2(\nu_0, V_0)$ prior on $\Sigma^{-1}$, independently. (ii) $\mb{M}_{2}$, which used a mixture normal data model $\omega\, \mbox{N}_2(\mu_1, \Sigma_{1}) + (1- \omega)\, \mbox{N}_2(\mu_2, \Sigma_{2})$; and imposes independent diffuse $\mbox{N}_2(0, 10^3 \mbox{I}_2)$ priors on $\mu_1$, $\mu_2$, an $\mbox{U}(0, 1)$ prior on $\omega$, and independent $\mbox{Wishart}_2(\nu_0, V_0)$ priors on $\Sigma_{1}^{-1}$ and $\Sigma_{2}^{-1}$. To explore the high-density neighborhoods of the posterior distributions,  we use coordinate-wise Metropolis--Hastings updates. For smaller sample sizes, the simpler model $\mb{M}_{1}$ provides higher marginal likelihood compared to $\mb{M}_{2}$. However, as the sample size grows, the more complex model $\mb{M}_{2}$ predictably starts being preferred; refer to Figure \ref{clustering_results} which plots the posterior model probability of $\mb{M}_{1}$ as a function of sample size.
Next, we consider the $\mbox{D}$-BETEL counterparts of $\mb{M}_{1}$ and $\mb{M}_{2}$, which we refer to as $\mb{M}_{1}^{\star}$ and $\mb{M}_{2}^{\star}$ respectively, with $\mb{M}_{1}^{\star}$ using a single normal distribution $\mbox{N}_2(\mu, \Sigma)$ as the centering distribution, and $\mb{M}_{2}^{\star}$ centered around $\omega \mbox{N}_2(\mu_1, \Sigma_{1}) + (1- \omega) \mbox{N}_2(\mu_2, \Sigma_{2})$.
We use same the prior specification \& MH sampling scheme as before.

First, we showcase our data driven approach to tune the hyper-parameter $\varepsilon$ for both $\mb{M}_{1}^{\star}$ and $\mb{M}_{2}^{\star}$. Figures \ref{fig:hyperparameter_set1} and \ref{fig:hyperparameter_set2} present plots for $\ELPD_{\varepsilon}$, $\mbox{SE}(\ELPD_{\varepsilon})$ and $\kappa$, defined in Section \ref{ssec:npBayes}, as functions of $\log\varepsilon$ for two  particular combination of $(n,\alpha)$ values. We considered a grid of $\varepsilon$ values over powers of $2$ and then use a finer grid in the interval where $\ELPD_{\varepsilon}$ undergoes steep change. For sufficiently large value of $\varepsilon$, the distance based constraint practically becomes inactive, and consequently $\ELPD_{\varepsilon}$ plateaus out  and $\mbox{SE}(\ELPD_{\varepsilon})\downarrow 0$.  Finally, we obtain the $\mbox{D}$-BETEL based parameter estimates $\hat{\theta}_{\rm MA}$ as delineated in Section \ref{ssec:npBayes}.   Although $\hat{\theta}_{\rm MA}$ in equation \eqref{eqn:ma} is calculated via an weighted average, in practice $\hat{\theta}_{\rm MA}$  typically degenerates to  estimates corresponding to a handful of values  of $\varepsilon$, as apparent in the plot of $\kappa$ as a function of $\log \varepsilon$ in Figures  \ref{fig:hyperparameter_set1}, \ref{fig:hyperparameter_set2}. We observe similar pattern for the remaining combinations of $(n,\alpha)$ values, and refrain from presenting them here in order to avoid repetitiveness.

Figure \ref{clustering_results} presents the posterior  probability of selecting the simpler model with only one bivariate normal component under the standard posterior, a fractional posterior with varying temperature parameters, and $\mbox{D}$-BETEL.
For the standard posterior, the posterior probability of selecting the simpler model $\mb{M}_{1}$ drop below $0.5$  as sample size increases. 
On the contrary, $\mbox{D}$-BETEL and fractional posterior \citep{doi:10.1080/01621459.2018.1469995} with small temperature parameter is more resistant towards presence of mild skewness in the data generating mechanism, and still prefer the simpler model  across the sample sizes we considered. However, unless the temperature parameter of the fractional posterior is chosen to be appropriately small,  it cannot reliably estimate the number of components in finite mixture models, under mild model mis-specification \citep{cai2}.

\begin{figure}[!htb]
    \centering
    \includegraphics[width=18cm, height = 4.5cm]{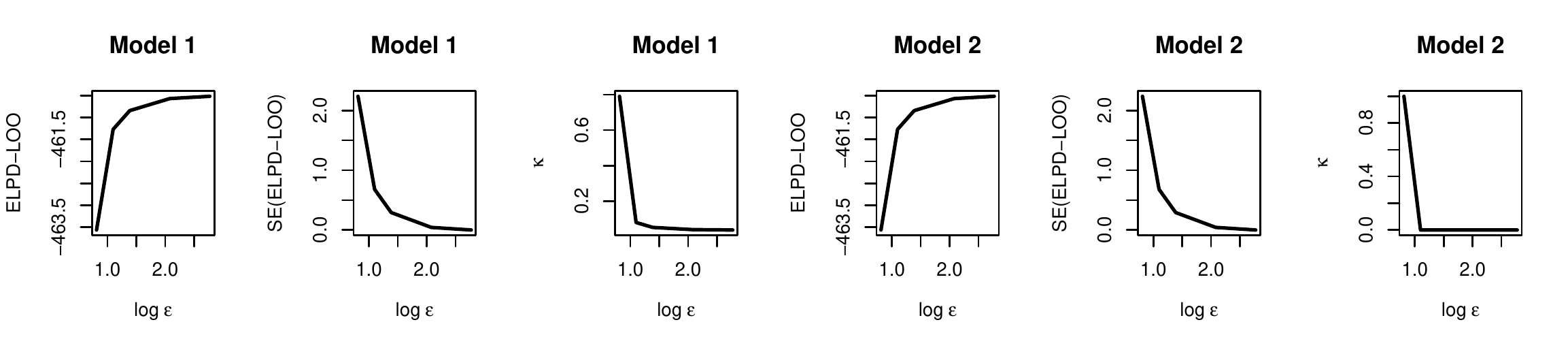} 
    \caption{\emph{\textbf{Hyper-parameter tuning for model based clustering with sample size $\mathbf{n=100}$, skewness parameter $\mathbf{\alpha = (3.5, 3.5)^\T}.$}  $\ELPD_{\varepsilon}$ gradually plateaus out  and $\mbox{SE}(\ELPD_{\varepsilon})\downarrow 0$ as $\log\varepsilon\uparrow$ for both the models. Consequently, weights $\kappa$ corresponding to a handful of $\varepsilon$ values contribute meaningfully to the weighted sum in $\hat{\theta}_{\rm MA}$ and rest are $\approx 0$. }\label{fig:hyperparameter_set1}}
\end{figure}
\begin{figure}[!htb]
    \centering 
    \includegraphics[width=18cm, height = 4.5cm]{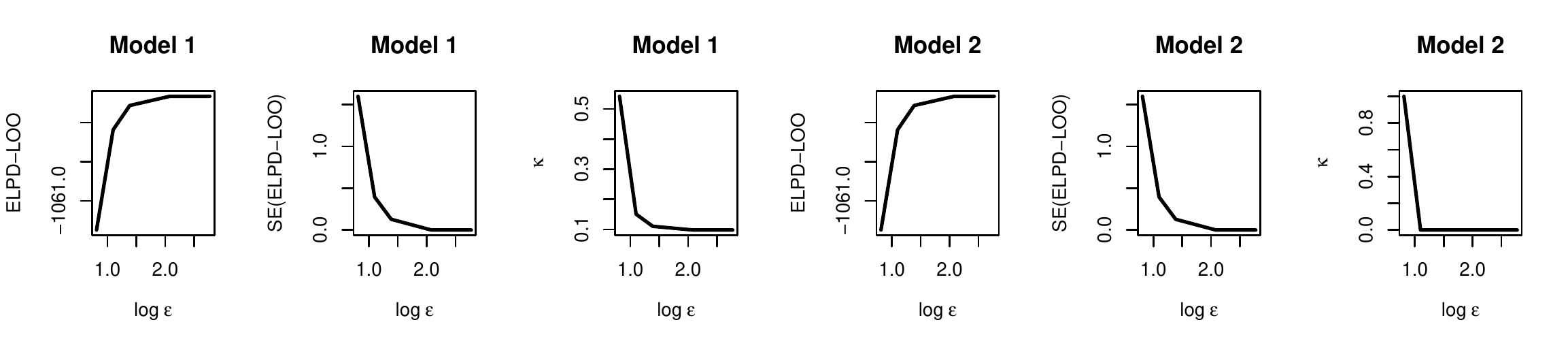} 
    \caption{\emph{\textbf{Hyper-parameter tuning for model based clustering with sample size $\mathbf{n=200}$, skewness parameter $\mathbf{\alpha = (2.5, 2.5)^\T}.$} $\ELPD_{\varepsilon}$ gradually plateaus out  and $\mbox{SE}(\ELPD_{\varepsilon})\downarrow 0$ as $\log\varepsilon\uparrow$ for both the models. Consequently, weights $\kappa$ corresponding to a handful of $\varepsilon$ values contribute meaningfully to the weighted sum in $\hat{\theta}_{\rm MA}$ and rest are $\approx 0$.}
    \label{fig:hyperparameter_set2}}
\end{figure}

\begin{figure}[!htb]
    \centering
    \includegraphics[width=17cm, height = 6cm]{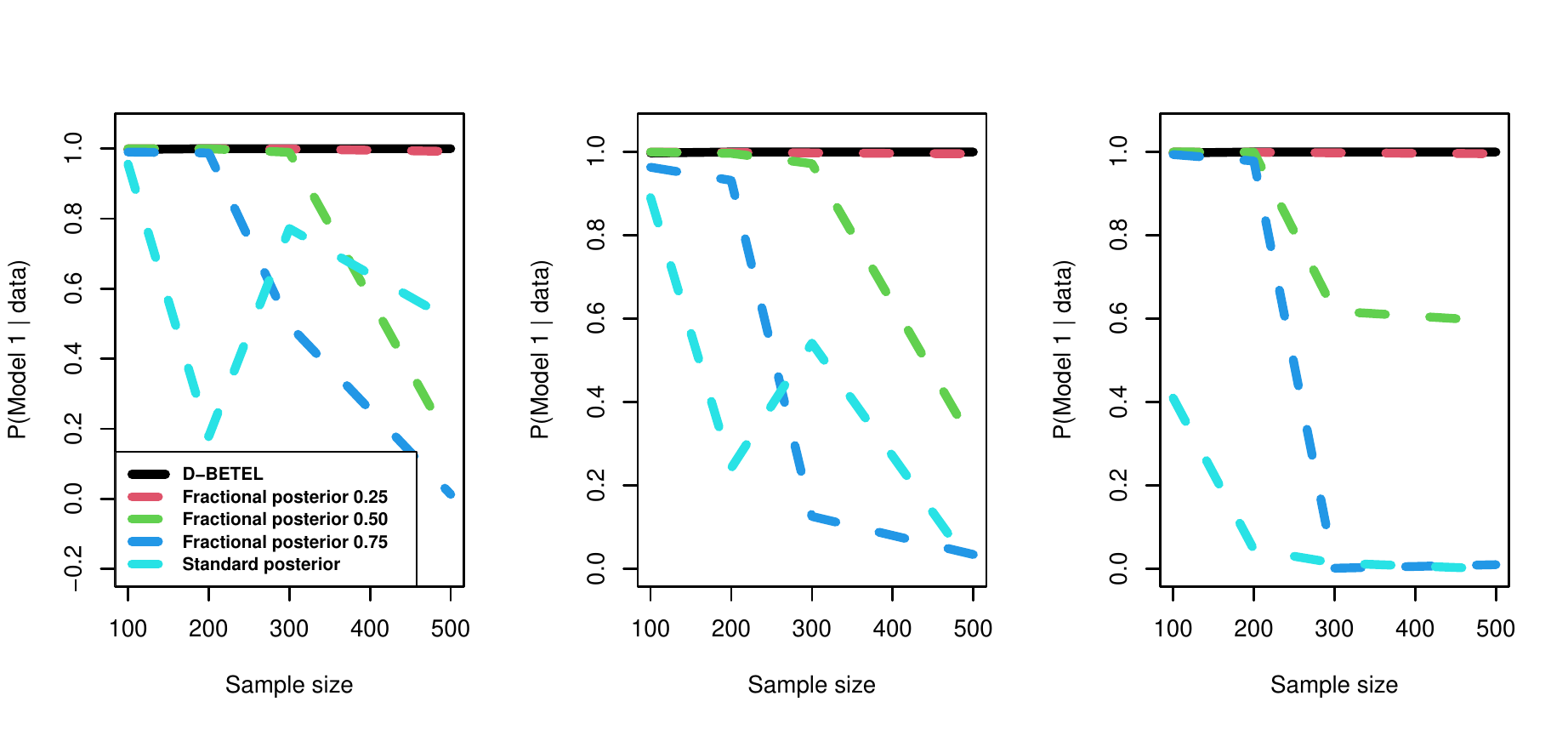} 
    \caption{\emph{\textbf{Model based clustering.} We are comparing the Bayes factor for selecting the simpler model via $\mbox{D}$-BETEL, the standard posterior, and the fractional posterior \citep{doi:10.1080/01621459.2018.1469995} with different temperatures, across varying values of skewness parameter $\alpha$ of the generating skew normal distribution and sample sizes. The left panel is for $\alpha= (2.5, 2.5)^\T$, the middle panel is for $\alpha= (3, 3)^\T$, and the right panel is for $\alpha= (3.5, 3.5)^\T$.  Unlike the standard posterior,  $\mbox{D}$-BETEL and fractional posterior with low temperature still prefer the simpler model across the sample sizes.} 
    \label{clustering_results}}
\end{figure}

\subsection{Generalised linear regression}\label{ssec:gen_reg}
Suppose we observe data  $\{(y_i,x_i) \in \mb{R}\times\mb{R}^{d}\}_{i=1}^n$ on a response variable $y$ and  covariates $x$ for $n$ individuals.  In generalised linear regression set up, we model the response by an exponential family distribution:
\begin{equation*}
 f(y_i\mid\theta_i, \phi) = \exp\bigg[\ \frac{y_i \theta_i - b(\theta_i)}{a(\phi)} + c(y_i, \phi)\bigg]   
\end{equation*}
where $a(\cdot),\ b(\cdot),\ c(\cdot)$ are known functions such that $m_i = b^\prime(\theta_i) ,\ \sigma_{i}^2  =\phi b^{\prime\prime}(\theta_i)$ are respectively the mean and variance of the  distribution, and there exists
a one-to-one continuously differentiable link function  $g(\cdot)$ such that $g^{-1}(x_{i}^\T\beta) = b^\prime(\theta_i) $.   The log-likelihood of the parameter of interest $\beta$ is 
\begin{equation*}
  l(\beta\mid x,y)   = \sum_{i=1}^n l_i(\beta\mid x_i, y_i) = \sum_{i=1}^n \bigg[ \frac{y_i \theta_i - b(\theta_i)}{a(\phi)}+ \log c(y_i,\phi)\bigg]  
\end{equation*}
where $\theta_i$ is a function of $m_i$. The corresponding Fisher's score function  $S = (S_{0},\ S_1,..,\ S_d)^\T$:
\begin{equation*}
S_j = \frac{\partial l}{\partial \beta_j}
= \sum_{i = 1}^n \bigg[ \frac{(y_i - m_i)}{a(\phi)} \frac{1}{V_i} \frac{\partial m_i}{\partial \beta_j}  \bigg] = 0 \end{equation*}
with $V_i =  \frac{\partial m_i}{\partial \theta_i} = b^{\prime\prime}(\theta_i)$. For simplicity of exposition, we express $S = \sum_{i = 1}^n \eta_{i}$ where  $\eta_{ij} = \frac{\partial l_i}{\partial \beta_j},\eta_i = (\eta_{i0},\ \eta_{i1}\ ,...,\ \eta_{id})^{\T},\ i = 1,\ldots,n; \ j  =0,\ldots,d$. The score statistic $S$ is asymptotically normal with mean 0 \citep{Haynes2013} and $n_i$ captures the deviation from 0 for the $i$-th observation. With that intuition, to conduct robust Bayesian inference on $\beta$,  we posit $\mbox{D}$-BETEL  on $\{\eta_i\}_{i=1}^n$ with a finite mixture of $(d+1)$-variate normal densities i.e $\sum_{j=1}^K \pi_j \ \mbox{N}(\mu_j, \Sigma_j)$ such that   $\sum_{j=1}^K \pi_j \mu_j = 0$ as our choice for centering parametric guess. 

We generate data from a Poisson random effects model,
\begin{equation*}
 \log(m_i) = \beta_0 + \beta_1 x_i + h_i ; \quad  y_i \sim \mbox{Poisson}(m_i),\ i = 1,\ldots,n,   
\end{equation*}
where $\beta_0 = 5, \ \beta_1 = 1$, \ $x_i\sim \mbox{N}(5, 1)$ and $h_i \sim (1-p) \ \mathbf{1}\{0\} + p \ \mbox{N}(1, 0.1^2)$, in order  to  mimic a situation where a small proportion of outliers are present in the data-set. We place flat $\mbox{N}(0, 100^2)$ priors on $\log\sigma^2_{1}(\beta)$, $\log\sigma^2_{2}(\beta)$, $\beta_0$ and $\beta_1$ and a $\mbox{U}(-1, 1)$ prior on $\rho(\beta)$, independently. We devise a Metropolis--Hastings algorithm to update $g(\beta) = (\log\sigma^2_{1}(\beta),\ \log\sigma^2_{2}(\beta),\ \rho(\beta), \beta_0,\ \beta_1)^\T$ at $(t+1)$th iteration using the 1-step proposal scheme: 
\begin{equation}
 g^{(t+1)}(\beta) \sim \mbox{N}_5\big(g^{(t)}(\beta),\  k\ \nabla g(\hat{\beta}_{m})\ \mbox{I}^{-1}(\hat{\beta}_{m})\ \nabla^{T} g(\hat{\beta}_{m})\big)   
\end{equation}
where $\mbox{I}(\hat{\beta}_m))$ is the Fisher's information matrix evaluated at the maximum likelihood estimator $\hat{\beta}_{m}$ of $\beta$, and $k$ is a tuning parameter. 

In Table \ref{table:glm_1}, we expand on the performance of $\mbox{D}$-BETEL for varying  extent of perturbations in the data generating mechanism with sample size $n=100$, relative to popular practical approaches. In particular, we compare  $\mbox{D}$-BETEL against a standard posterior as well as Bayesian ETEL \citep{Chib2018} with the estimating equations set to  $\mbox{E}[\partial \log l(\beta\mid X, Y)/ \partial \beta]=0$ to infer about the parameter $\beta$. The latter approach can be regarded as a variant of $\mbox{D}$-BETEL, with a stricter moment-type constraint.
From Table \ref{table:glm_1}, we report the $L_1$ error of posterior means, length of the HPD sets and associated coverage probabilities (within braces) for $\mbox{D}$-BETEL and competing approaches. It is evident that $\mbox{D}$-BETEL is more resistant towards presence of outliers when compared with the standard Bayesian and MCM based approaches, across all the sample sizes and proportion of contamination in the data sets that we considered. Also, $\mbox{D}$-BETEL provides slightly wider credible sets compared to the standard posterior based approach, while maintaining high coverage probability. Additional simulation results for $n=250, 500$ are presented in Section \ref{sup:glm} of the supplement.

\begin{table}[!htb]
  \caption{\emph{\textbf{Generalised linear regression (Poisson regression).} Here the \textbf{sample size $n$ is $100$}. We compare standard posterior yielded from the fully parametric model, moment conditional model (MCM) based on the maximum likelihood equations,
  and $\mbox{D}$-BETEL based parameter estimates over 50 replicated simulations with proportion of outlier $p=0.10, 0.12, 0.15$. 
$\mbox{D}$-BETEL is more resistant towards presence of outliers all values of $p$ considered, however it provides slightly wider $95\%$ credible sets while maintaining the high coverage probability. Additional simulation results for $n=250, 500$ is presented in the supplement \ref{sup:glm}.}}\label{table:glm_1}
  \centering
  \begin{tabular}{lllllllllll}
    \toprule
    \cmidrule(r){1-2}
     &     &      \multicolumn{2}{c}{$\mbox{D}$-BETEL} & \multicolumn{2}{c}{Standard posterior} & \multicolumn{2}{c}{MCM} \\
    \midrule
    p & $\theta$     & $||\theta - \hat{\theta}||_1$     & HPD  & $||\theta - \hat{\theta}||_1$     & HPD   & $||\theta - \hat{\theta}||_1$     & HPD \\
    \midrule
    0.10 & $\beta_0$ & 0.02 & 0.18 (1.00)  & 0.35   & 0.12 (0.20)& 0.41 & 0.22 (0.10) \\
        & $\beta_1$ & 0.01 & 0.02 (1.00)  & 0.07 &0.02 (0.20) &0.06 & 0.04 (0.22)\\
    \midrule
    0.12 & $\beta_0$ &0.02 &0.22 (1.00) &0.31& 0.12 (0.35) & 0.47 & 0.35 (0.14) \\
        & $\beta_1$ & 0.01 &0.04 (1.00) &0.08& 0.02 (0.59) & 0.06 & 0.06 (0.32) \\
    \midrule
    0.15 & $\beta_0$ & 0.06& 0.22 (0.94)  &0.47  & 0.11 (0.00) & 0.54 &0.23 (0.06)\\
        & $\beta_1$ & 0.01 &0.04 (0.94)  &0.06& 0.02 (0.00) & 0.09&0.05 (0.18)\\
    
    \bottomrule
  \end{tabular}
\end{table}

\section{Algorithmic fairness: demographic parity}\label{ssec:algorithmic_fairness}
In this section, we present applications of the dual formulation of D-BETEL presented in equations \eqref{eqn:dbetel:alt}-\eqref{wbetel_dual} in the context of ensuring demographic parity in machine learning algorithms. Machine learning algorithms are increasingly used in critical decisions affecting human lives including but not limited to credit, employment, education, and criminal justice, and hence fairness has emerged as a primary pillar of modern machine learning research in recent years. Discrimination refers to unfavorable treatment of entities due to their membership to certain demographic groups that are determined by the attributes protected by law, called protected attributes. The goal of demographic parity or statistical parity \citep{90450a4b5b49471b8111fc88355f2e7f, gajane2018formalizing} in machine learning is to design algorithms that yield fair inferences devoid of discrimination due to membership to certain demographic groups determined by a protected attribute. 

First, we introduce the mathematical formalization of the notions of demographic parity. To that end,  we assume that $X$ denotes the feature vector used for predictions, $A$ is the protected attribute with two levels $\{S,T\}$, and  $Y$ is the true response. Parity constraints are phrased in terms of the distribution over $(X, A, Y)$. 
Two definitions are in order.
\begin{definition}[Demographic parity, \citep{90450a4b5b49471b8111fc88355f2e7f}]\label{def:DP}
A predictor $h$  satisfies demographic parity under the distribution over  $(X, A, Y)$ if $h(X)$ is independent of the protected attribute $A$, i.e ,
\begin{align*}
    \mb{P}[h(X) \geq z \mid A = S] = \mb{P}[h(X) \geq z \mid A = T] = \mb{P}[h(X) \geq z ]\ \text{for all} \ z.
\end{align*}
\end{definition}

\begin{definition}[Demographic parity in expectation, \citep{90450a4b5b49471b8111fc88355f2e7f}
]\label{def:DP_exp}
A predictor $h$  satisfies demographic parity under the distribution over  $(X, A, Y)$ if $h(X)$ is independent of the protected attribute $A$, i.e ,
\begin{align*}
    \mb{E}[h(X) \mid A = S] = \mb{E}[h(X)  \mid A = T] = \mb{E}[h(X) ].
\end{align*}
\end{definition}
It is perhaps instructive to examine the notion of demographic parity from the viewpoint of causal inference. In particular, \citet{Nabi2018} shows how statistical parity may be unable to reveal an underlying causal effect that is discriminatory. For simplicity in exposition, we concentrate on binary response variables for a while and note down the differences between the notions of demographic parity in fairness and the notion of average treatment (protected attribute) effect in the context of causal inference. Demographic parity merely seeks that the conditional probabilities $ \mb{P}[h(X) = 1 \mid A = S],\ \mb{P}[h(X)=0 \mid A = T]$ are equal. On the other hand, in the potential outcome framework of causal inference, we denote the potential outcome $h(x)$  for an individual as  $h_{(k)}(x)$, had the treatment (protected attribute) $A$ been assigned the value $k\in\{S,T\}$. Then, to ensure that there is no average treatment effect, we seek $\mb{E}[h_{(S)}(X) ]-\mb{E}[h_{(T)}(X)]=0$, or equivalently  $\mb{P}[h_{(S)}(X) = 1]=\mb{P}[h_{(T)}(X) = 1]$. For observational studies, to carry out inference in this potential outcomes framework, it is customary to assume conditional ignorability -- which specifies a set of variables, called confounders, given which the potential outcome becomes independent of the assigned treatment (protected attribute).
In essence, ignorability implies that we can express the probability of a potential outcome  conditional on $E$, in terms of probability of the observed outcome conditional on both $E$ and  $C$, i.e $\mb{P}[h_{(k)}(X)\mid E] = \mb{P}[h(X)\mid E, A = k], \ k\in\{S, T\}$. Now, to ensure that there is no average treatment effect, we need
\begin{align*}
    \int\mb{P}[h(X)\mid E, A = S]\ \pi(E) dE = \int\mb{P}[h(X)\mid E, A = T]\ \pi(E) dE.
\end{align*}
However, deciding upon possible confounders $E$ in fairness problems is often not straightforward. Innovative solutions have been proposed in the literature. For example,  \citet{https://doi.org/10.1111/rssa.12613} proposed the construction of an $\Tilde{X}$, a reconstructed version of the data matrix that is orthogonal to the vector of protected attributes with minimal information loss i.e $||X-\Tilde{X}||_{\rm F}$ is minimum subject to the constraint $\langle \Tilde{X}, A\rangle=0$, where $\langle \cdot, \cdot\rangle$ stands for inner product and $||\cdot||_{\rm F}$ denotes the Frobenius norm of matrices. It is straightforward to see that non-linear dependencies on the vector of $A$ could still be present in the transformed matrix $\Tilde{X}$, but  $X-\Tilde{X}$ can serve as a confounder in numerous practical purposes, e.g in models linear in covariates. This presents an interesting alleyway for future inquiry and is well beyond the scope of this article.  For the rest of our presentation, we shall build on the basic notion of statistical/demographic parity.

Although the notions in Definitions  \ref{def:DP} and \ref{def:DP_exp}  coincide when we work with binary responses, the latter may be amenable to simple computational algorithms \citep{e21080741} compared to the general definition. However, the notion of demographic parity in expectation is extremely prohibitive since one cannot control the predictor $h$ over its entire domain, and depending on the application of interest we may be solely interested in controlling the tails of the predictor \citep{yang2019fair}. Taking refuge to our semi-parametric inference framework, we offer a flexible as well as a computationally feasible compromise between the notions in Definitions \ref{def:DP} and \ref{def:DP_exp}. To that end, we introduce the notion of demographic parity in the Wasserstein metric next.
\begin{definition}[Demographic parity in Wasserstein metric]
A predictor $h$ achieves demographic parity in Wasserstein metric  with bias $\varepsilon$, under the distribution over  $(X, A, Y)$  if 
\begin{align*}
  W^{2}_{\rm AR}\big[F_{h_{S}}, F_{h_{T}}\big] \leq \varepsilon,  
\end{align*}
where $F_{h_{k}}$ is the cdf of $h$ under sub-population $k$ i.e $h(X)\mid A = k, k \in\{S,T\}$.
\end{definition}

Next, we shall see demographic parity in the Wasserstein metric in action. Suppose we have data $(x_i, y_i, a_i)\in \mb{R}^d\times\mb{R}\times\{S,T\}$ for  $n$ individuals on $p$-dimensional covariate $x$, univariate continuous response $y$, and levels of the protected attribute $a\in\{S,T\}$.
For the sake of simplicity in exposition, we also assume that $a_i = S,\ i=1,\ldots, n_S$  and $a_i = T,\ i=n_S+ 1,\ldots, n$ where $n = n_S + n_T$. Next, we posit a predictive model
\begin{align*}
    y_i = h(x_i, \theta_{(a_i)}) + e_i, \ e_i \overset{i.i.d}\sim \mbox{N}(0, \sigma^2), \ i =1, \ldots, n,
\end{align*}
where $h$ is potentially non-linear, and $(\theta_{(S)}, \theta_{(T)}) $ is the model parameter of interest to be estimated under the demographic parity constraint  $W^{2}_{\rm AR}\big[F_{h_{S}}, F_{h_{T}}\big] \leq \varepsilon$. In particular, we consider the empirical cdf of $h$ under sub-population $S$,  $F_{h_{S}}= 1/n_S\sum_{i = 1}^{n_S} \delta\{h(x_{iS})\}$; and a weighted empirical cdf of $h$ under sub-population $T$, $F_{h_{T}}= \sum_{i = n_S + 1}^{n} w_{i}\ \delta\{h(x_{iT})\}$. Here $\delta(\cdot)$ is the Dirac delta measure. The goal is to infer about $(\theta_{(S)}, \theta_{(T)}, w) $ ensuring that demographic parity constraint i.e $F_{h_{S}},F_{h_{T}}$ are close  with respect $W^{2}_{\rm AR}$, at the same time the extent of re-weighting in $F_{h_{T}}$ is minimal i.e the entropy $-\sum_{i=n_S + 1}^n w_i\log w_i$ is close to the maximal entropy $\log n_T$.  A related idea in \citet{pmlr-v115-jiang20a} deals with Wasserstein-1 constrained fair classification problems, but our approach of additionally re-weighting the observations offers more flexibility with possible ramifications in studying fairness in mis-specified models.  We achieve our inferential goal via an \emph{in-model} approach, based on the dual formulation of D-BETEL in equations \eqref{eqn:dbetel:alt}-\eqref{wbetel_dual}:
{\smaller
\begin{equation}\label{wgf_simultenous}
  \max_{w,\ \theta_{(S)},\ \theta_{(T)},\ \sigma^2} \bigg\{ - \frac{1}{n_S}\sum_{i=1}^{n_S} l_{i}(\theta_{(S)}\mid x_i) - \sum_{i=n_s + 1}^n w_i\ l_{i}(\theta_{(T)}\mid x_i)  - (1-\lambda^{\star}) W_{\rm AR}^{2}\big[F_{h_{S}}, F_{h_{T}}\big] - \lambda^{\star} \sum_{i=n_s + 1}^{n} w_i\log w_i \bigg\}
  \end{equation}
}
where $\sum_{i=n_s + 1}^{n} w_i = 1$ and $l_{i}(\theta_{(a_i)}\mid x_i) = (y_i - h(x_i, \theta_{(a_i)}))^2/2\sigma^2, \ i = 1,\ldots,n$. For a resulting re-weighting vector $w^{\star} = (w^{\star}_{n_S + 1},\ldots,w^{\star}_{n})^{\prime}$, we can obtain fair prediction at a new $x\in T$ via a weighted kernel density estimate at $x$.

As a competitor to the \emph{in-model} scheme described in \eqref{wgf_simultenous}, we introduce a \emph{two-step} procedure: \\
\textbf{Step 1:} We obtain model parameter estimates by
\begin{equation}\label{wgf_twostep_1}
 (\hat{\theta}_{(S)},\ \hat{\theta}_{(T)} ,\ \hat{\sigma}^2) =  \argmax_{\theta_{(S)},\ \theta_{(T)} ,\ \sigma^2} \bigg\{ - \frac{1}{n_S}\sum_{i=1}^{n_S} l_{i}(\theta_{(S)}\mid x_i) - \frac{1}{n_T}\sum_{i=n_s + 1}^n l_{i}(\theta_{(T)}\mid x_i)   \bigg\}
\end{equation}
followed by a post-processing step at $(\hat{\theta}_{(S)},\ \hat{\theta}_{(T)} ,\ \hat{\sigma}^2)$ to obtain $w^{\star}$  \\
\textbf{Step 2:}
\begin{equation}\label{wgf_twostep_2}
 w^{\star} = \argmax_{w} \big\{ -(1-\lambda^{\star})  \  W_{\rm AR}^{2}\big[F_{h_{S}}, F_{h_{T}}\big] - \lambda^{\star}  \sum_{i=n_s +1}^n w_i\log w_i \big\}.  
\end{equation}
Next, we shall assess the relative performance of the \emph{in-model} scheme in \eqref{wgf_simultenous} and  the \emph{two-step} scheme defined in \eqref{wgf_twostep_1}--\eqref{wgf_twostep_2} under two real data examples.

\begin{figure}[!htb]
    \centering
    \subfloat
    {{\includegraphics[width=7cm, height = 3.5cm]{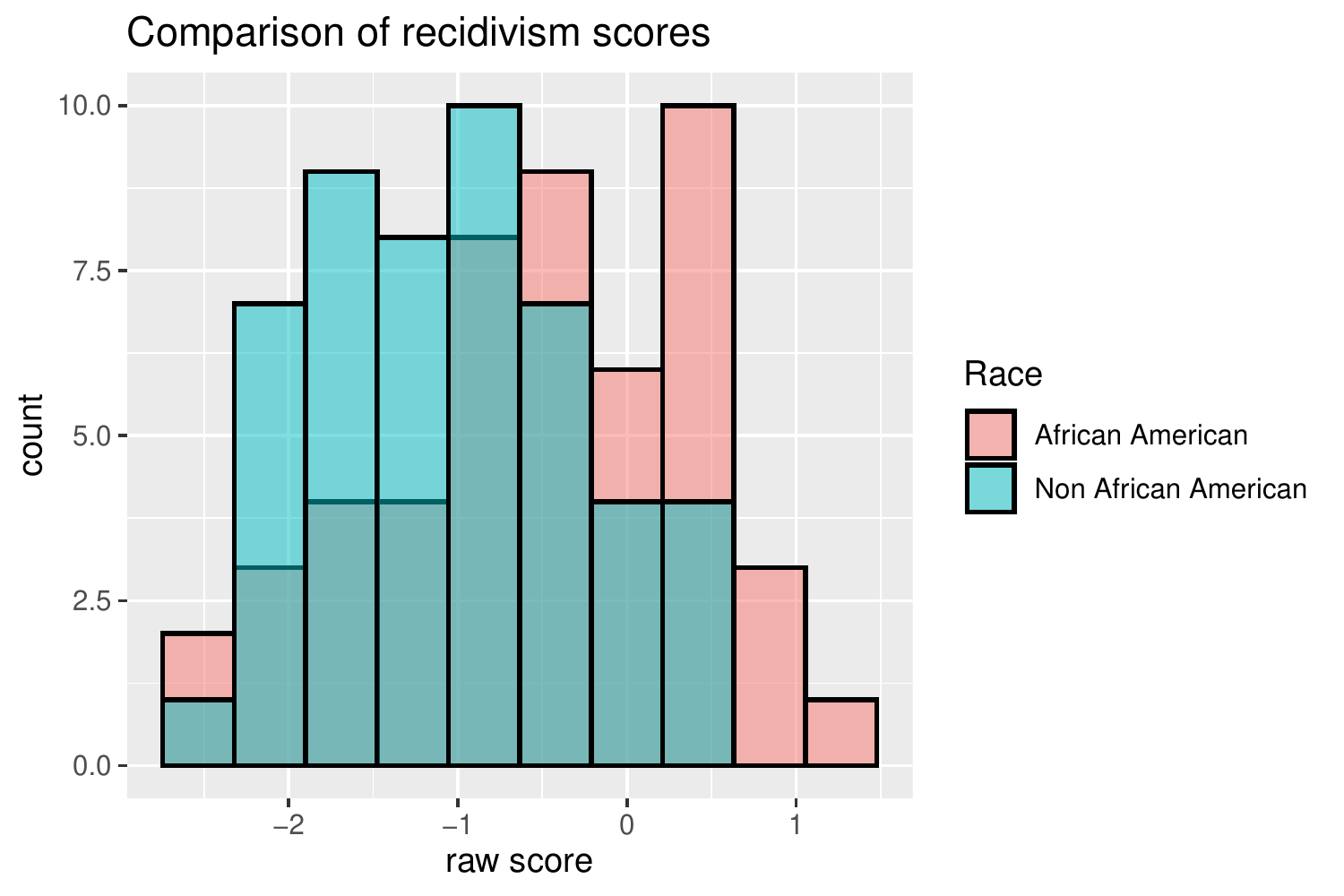}}}%
    \quad
    \subfloat
    {{\includegraphics[width=6cm, height = 3.5cm]{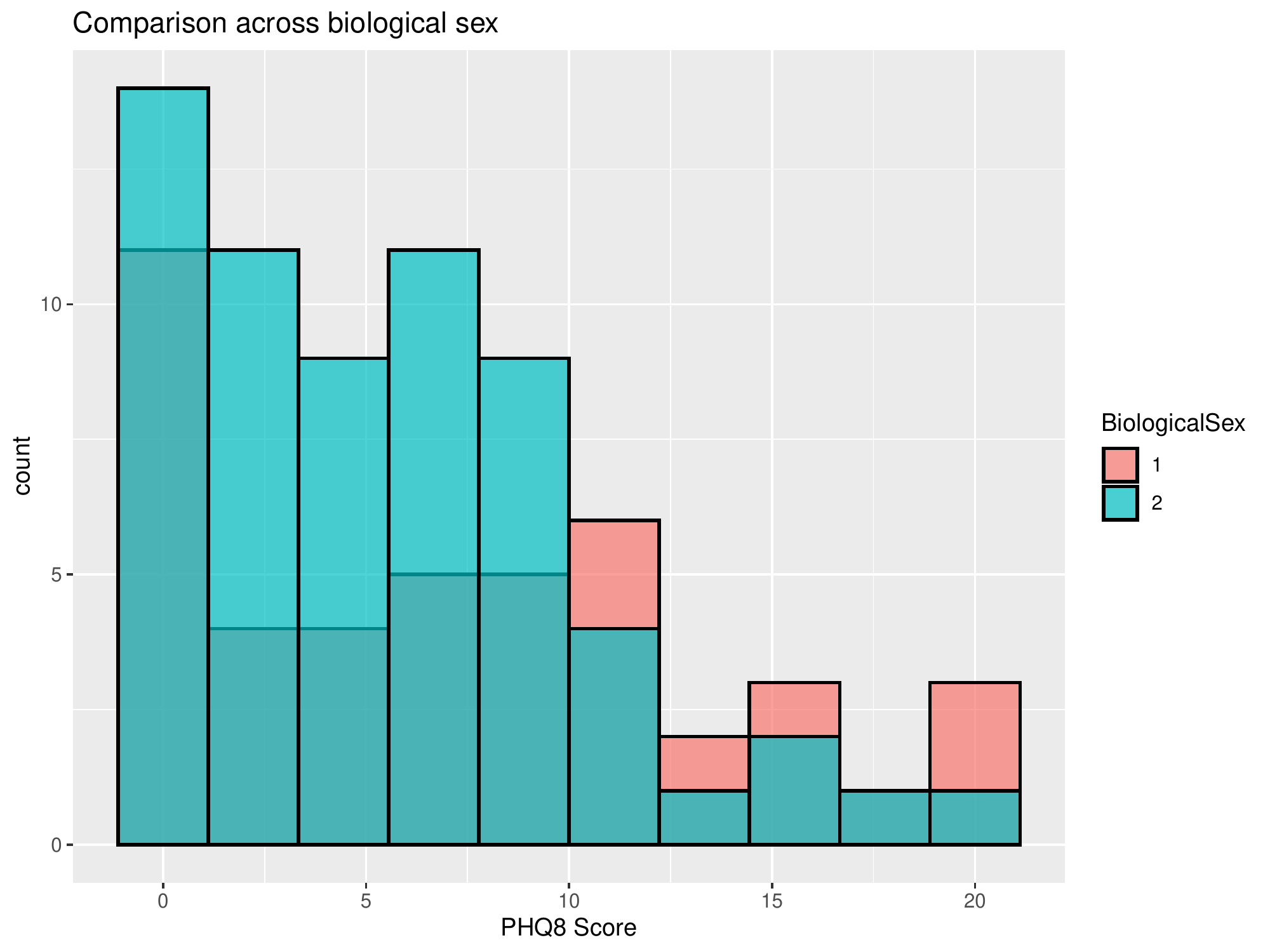} }}%
    \caption{\emph{\textbf{(a) COMPAS Dataset.} The histograms of raw recidivism score for African and non African-Americans show a clear discrepancy.\label{diag:compas_data} \textbf{(b) Distress Analysis Interview Corpus.} The histograms of raw PHQ-8  for the two biological genders show a clear discrepancy.}\label{diag:daic}} 
\end{figure}
\subsection{COMPAS recidivism  data analysis}
In this sub-section, we consider a case study on algorithmic criminal risk assessment. We shall focus here on the popular compas dataset \citep{https://doi.org/10.1111/rssa.12613} that includes detailed information on criminal history for the defendants in Broward County, Florida, freely available from the \href{https://www.propublica.org/datastore/dataset/compas-recidivism-risk-score-data-and-analysis}{\textcolor{purple}{propublica}} website.  For each individual, several features on criminal history are available, such as the number of past felonies, misdemeanors, and juvenile offenses; additional demographic information includes the sex, age, and ethnic group of each defendant.  We focus on predicting two-year recidivism score $y$ (continuous) as a function of the defendant’s demographic
information except for race and criminal history $x$, while race (categorical) serves as a protected attribute. Algorithms for making such
predictions are routinely used in courtrooms to advise judges, and concerns about the fairness of such
tools with respect to the race of the defendants are raised. Therefore, it is of interest to develop novel methods to produce predictions while avoiding disparate treatment on the basis of the protected attribute race.

For simplicity of exposition, we only consider two levels for the protected attribute race, namely, African-American or non-African-American, and consider a sub-sample of the entire data set with $100$ defendants corresponding to each level of the protected attribute. As covariate, for each defendant, we consider demographic information -- sex (binary), age (continuous), marital status (categorical); and criminal status -- legal status (categorical), supervision level (categorical), custody status (categorical). We use linear regression (i.e $h$ is linear in the covariates) as our predictive model of choice; the methodology readily extends to more complicated models. The histograms of raw recidivism score for African-Americans versus non-African-Americans show a clear discrepancy (refer to Figure \ref{diag:compas_data}). We shall assess the relative performance of the \emph{in-model} scheme in \eqref{wgf_simultenous} and \emph{two-step} scheme in \eqref{wgf_twostep_1}--\eqref{wgf_twostep_2} in ensuring demographic parity with respect to the protected attribute race (refer to Figure \ref{diag:compas_plots}). When we fit the predictive model without any fairness constraint, the fitted empirical cumulative distribution functions corresponding to the two sub-populations
are widely different. Our \emph{in-model} scheme, as well as \emph{two-step} scheme significantly reduce the discrepancy owing to their in-built fairness-based regularization. As expected, the \emph{in-model} scheme provides slightly lower bias since it performs the two-step optimization simultaneously.

\begin{figure}[!htb]
    \centering
    {{\includegraphics[width=15cm, height = 4cm]{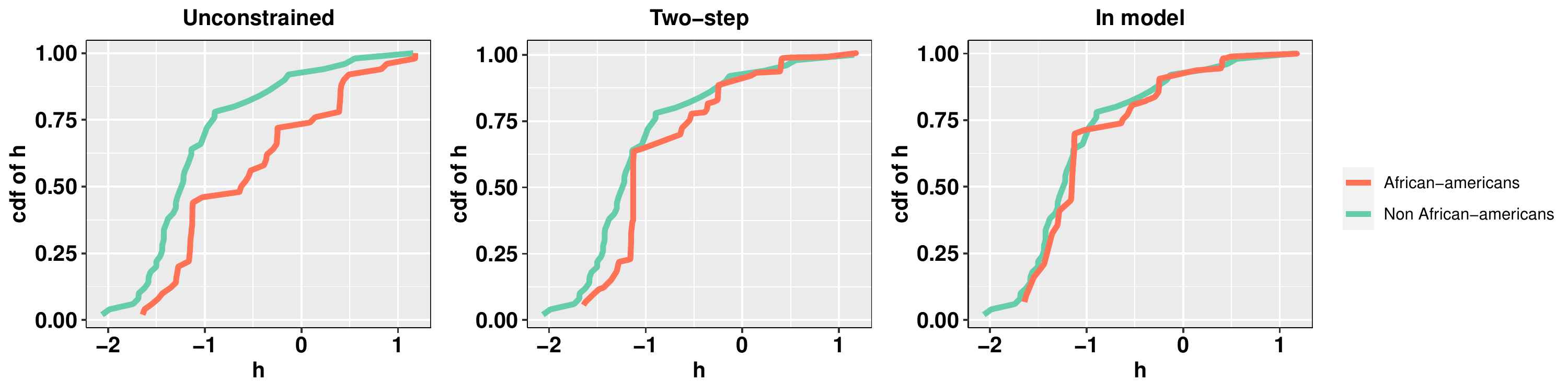} }}%
    
    \caption{\emph{\textbf{COMPAS dataset.} Empirical cdfs of  fitted $h$ for the two groups, with  no fairness constraint $( W_{AR} = 0.72)$, fair post-processing $( W_{AR} = 0.05)$, and fair model fitting with $( W_{AR}= 0.02)$ respectively at $\lambda^{\star} = 0$.}\label{diag:compas_plots}}
\end{figure}

\subsection{Distress Analysis Interview Corpus (DAIC)}
The Distress Analysis Interview Corpus (DAIC) is a multi-modal collection of semi-structured clinical interviews, available upon request from the \href{https://dcapswoz.ict.usc.edu/}{\textcolor{purple}{DAIC-WOZ}} website. Designed to simulate standard protocols for identifying people at risk for post-traumatic stress disorder (PTSD) and major depression, these interviews were collected as part of a larger effort to create a computer agent that interviews people and identifies verbal and nonverbal indicators of mental illness. Participants are drawn from two distinct populations living in the Greater Los Angeles metropolitan area – veterans of the U.S. armed forces and from the general public and are coded for depression, PTSD, and anxiety based on accepted psychiatric questionnaires. The corpus contains audio, video, and depth sensor (Microsoft Kinect) recordings of all the interviews, generated logs of the character’s speech and nonverbal behavior events, questionnaire data, and interview transcriptions. For further details on the data set, readers are advised to refer to \citet{gratch-etal-2014-distress}.

We are particularly interested in the PHQ-8 score that captures the severity of depression. The scores range from $0$ to $27$ with a score from $0-4$ considered none or minimal, $5-9$ mild, $10-14$ moderate, $15-19$ moderately severe, and $20-27$ severe. In this application, we work with this PHQ-8 (continuous response), biological gender (binary protected attribute), and $17$ derived audio features (continuous covariates) corresponding to the $n = 107$ subjects \ref{diag:daic}. The histograms of the PHQ-8 score for two biological genders show a clear discrepancy (refer to Figure \ref{diag:daic}). Therefore, we shall assess the relative performance of the \emph{in-model} scheme in \eqref{wgf_simultenous}) and \emph{two-step} scheme \eqref{wgf_twostep_1}--\eqref{wgf_twostep_2} in ensuring demographic parity with respect to biological gender (refer to Figure \ref{diag:daic_plots}).  As earlier, for the sake of simplicity of exposition, we use linear regression (i.e $h$ is linear in the covariates) as our predictive model of choice.  When we fit the predictive model without any fairness constraint, the fitted empirical cumulative distribution functions corresponding to the two biological genders are widely different. Our \emph{in model} scheme, as well as \emph{two-step} scheme significantly reduce the discrepancy owing to their in-built fairness-based regularization. As noted earlier, the \emph{in model} scheme provides lower bias since it performs the two-step optimization simultaneously.
\begin{figure}[!htb]
    \centering
    {{\includegraphics[width=15cm, height = 4cm]{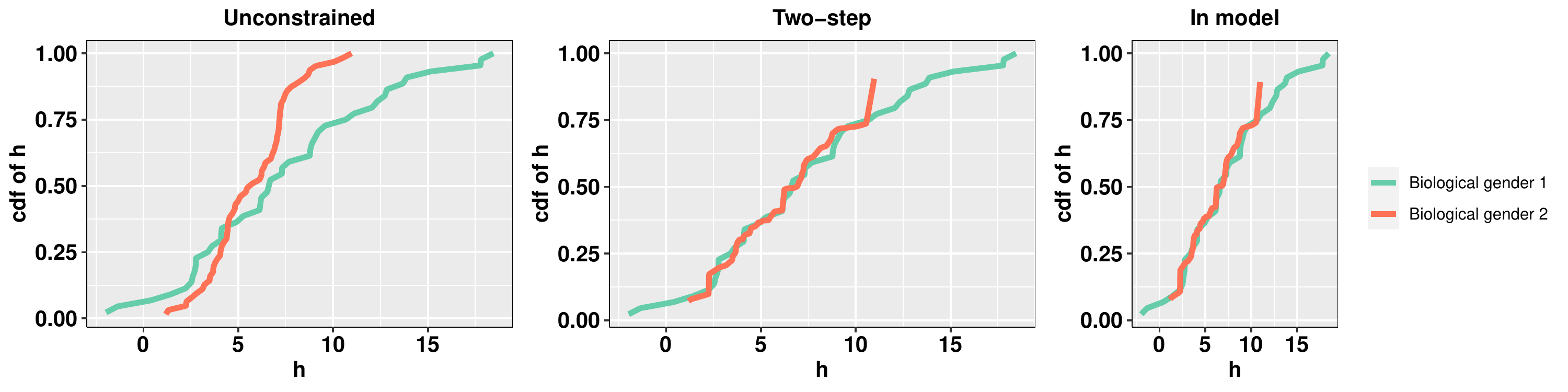} }}%
    
    \caption{\emph{\textbf{Distress Analysis Interview Corpus.} \emph{Empirical cdfs of  fitted $h$ for the two groups, with  no fairness constraint $( W_{AR} = 19.32)$, fair post-processing $( W_{AR} = 2.24)$, and fair model fitting with $( W_{AR}= 0.79$) respectively at $\lambda^{\star} = 0$. }}\label{diag:daic_plots}}
\end{figure}

\section{Discussion}
Generative probabilistic models are immensely popular in applications as they provide a general recipe for statistical inference using the maximum likelihood or Bayesian framework. However, it is also well understood that the resulting inference can crucially depend on the modeling assumptions. 
In this article, we introduced a flexible Bayesian semi-parametric modeling framework $\mbox{D}$-BETEL, and demonstrated it's utility to conduct robust inference under perturbations of the data-generating mechanism. $\mbox{D}$-BETEL is endowed with a fully data-driven hyper parameter tuning scheme, and enjoys a valid generative model interpretation, which is scarce in pseudo-likelihood based robust Bayesian methods. R scripts to reproduce the results presented in the article are available at \href{https://github.com/zovialpapai/D-BETEL}{\textcolor{purple}{zovialpapai/D-BETEL}}.

While semi-parametric in nature, a particularly attractive feature of $\mbox{D}$-BETEL is that the user only needs to specify a plausible family of probability models $F_\theta$ for the data along with a prior distribution for the parameter of interest $\theta$, and does not need to explicitly model departures from the parametric guess as is typical with nonparametric Bayesian techniques; all nuisance parameters are implicitly marginalized out and a marginal posterior for $\theta$ is returned. It remains possible to retrieve a discretized estimate of the generating distribution to allow a more fine-grained analysis of how the data departs from the parameteric guess. The proposed approach is also very general; while we have illustrated its usage for i.i.d. and independent non-i.i.d (i.n.i.d.) setups, extensions to broader classes of dependent data models should be straightforward. Studying theoretical properties of $\mbox{D}$-BETEL, especially second-order properties, is an interesting avenue for future work. 

Our framework can also be extended beyond the traditional statistical modeling setup to mitigate inherent biases in machine learning applications, We have offered an illustration in the context of algorithmic fairness, and we wish to investigate further applications in trustworthy AI, encompassing robustness, fairness \& differential privacy, in the future. We also believe the flexibility and the computational simplicity of our novel Wasserstein metric ANDREW may render itself useful in embedding of complex objects, e.g words, images, as probability distributions \citep{Jebara2004ProbabilityPK, https://doi.org/10.48550/arxiv.1412.6623} which has emerged as a popular application of optimal transport to machine learning problems.

\section*{Supplemental Document}
An online supplement contains proofs of the theorems stated in the main document along with auxiliary lemmas, and additional numerical results. 
\vspace{-0.2in}
\section*{Acknowledgement} Drs. Bhattacharya and Pati acknowledge NSF DMS-1916371 and NSF DMS-2210689 for partially funding the project. 

\bibliography{paper-ref}

\begin{thebibliography}{80}
\providecommand{\natexlab}[1]{#1}
\providecommand{\url}[1]{\texttt{#1}}
\expandafter\ifx\csname urlstyle\endcsname\relax
  \providecommand{\doi}[1]{doi: #1}\else
  \providecommand{\doi}{doi: \begingroup \urlstyle{rm}\Url}\fi

\bibitem[Agarwal et~al.(2019)Agarwal, Dud{\'i}k, and
  Wu]{90450a4b5b49471b8111fc88355f2e7f}
Alekh Agarwal, Miroslav Dud{\'i}k, and {Zhiwei Steven} Wu.
\newblock Fair regression: Quantitative definitions and reduction-based
  algorithms.
\newblock In \emph{36th International Conference on Machine Learning, ICML
  2019}, 36th International Conference on Machine Learning, ICML 2019, pages
  166--183. International Machine Learning Society (IMLS), January 2019.
\newblock 36th International Conference on Machine Learning, ICML 2019 ;
  Conference date: 09-06-2019 Through 15-06-2019.

\bibitem[Aliverti et~al.(2021)Aliverti, Lum, Johndrow, and
  Dunson]{https://doi.org/10.1111/rssa.12613}
Emanuele Aliverti, Kristian Lum, James~E. Johndrow, and David~B. Dunson.
\newblock Removing the influence of group variables in high-dimensional
  predictive modelling.
\newblock \emph{Journal of the Royal Statistical Society: Series A (Statistics
  in Society)}, 184\penalty0 (3):\penalty0 791--811, 2021.
\newblock \doi{https://doi.org/10.1111/rssa.12613}.
\newblock URL
  \url{https://rss.onlinelibrary.wiley.com/doi/abs/10.1111/rssa.12613}.

\bibitem[Antoniak(1974)]{10.1214/aos/1176342871}
Charles~E. Antoniak.
\newblock {Mixtures of Dirichlet Processes with Applications to {B}ayesian
  Nonparametric Problems}.
\newblock \emph{The Annals of Statistics}, 2\penalty0 (6):\penalty0 1152 --
  1174, 1974.
\newblock \doi{10.1214/aos/1176342871}.
\newblock URL \url{https://doi.org/10.1214/aos/1176342871}.

\bibitem[Avella-Medina(2021)]{doi:10.1080/01621459.2019.1700130}
Marco Avella-Medina.
\newblock Privacy-preserving parametric inference: A case for robust
  statistics.
\newblock \emph{Journal of the American Statistical Association}, 116\penalty0
  (534):\penalty0 969--983, 2021.
\newblock \doi{10.1080/01621459.2019.1700130}.
\newblock URL \url{https://doi.org/10.1080/01621459.2019.1700130}.

\bibitem[Azzalini and Valle(1996)]{SkewNormal}
A.~Azzalini and A.~DALLA Valle.
\newblock {The multivariate skew-normal distribution}.
\newblock \emph{Biometrika}, 83\penalty0 (4):\penalty0 715--726, 12 1996.
\newblock ISSN 0006-3444.
\newblock \doi{10.1093/biomet/83.4.715}.
\newblock URL \url{https://doi.org/10.1093/biomet/83.4.715}.

\bibitem[Becker et~al.(2011)Becker, Candès, and Grant]{TCFOCS}
Becker, Candès, and Grant.
\newblock Templates for convex cone problems with applications to sparse signal
  recovery.
\newblock \emph{Math. Prog. Comp}, 2011.

\bibitem[Bernton et~al.(2019)Bernton, Jacob, Gerber, and Robert]{Bernton_2019}
Espen Bernton, Pierre~E. Jacob, Mathieu Gerber, and Christian~P. Robert.
\newblock Approximate {B}ayesian computation with the {W}asserstein distance.
\newblock \emph{Journal of the Royal Statistical Society: Series B (Statistical
  Methodology)}, 81\penalty0 (2):\penalty0 235--269, feb 2019.
\newblock \doi{10.1111/rssb.12312}.
\newblock URL \url{https://doi.org/10.1111%2Frssb.12312}.

\bibitem[Bhatia et~al.(2017)Bhatia, Jain, and
  Lim]{https://doi.org/10.48550/arxiv.1712.01504}
Rajendra Bhatia, Tanvi Jain, and Yongdo Lim.
\newblock On the bures-{W}asserstein distance between positive definite
  matrices, 2017.
\newblock URL \url{https://arxiv.org/abs/1712.01504}.

\bibitem[Bion-Nadal and Talay(2019)]{BionNadal2019OnAW}
Jocelyne Bion-Nadal and Denis Talay.
\newblock On a {W}asserstein-type distance between solutions to stochastic
  differential equations.
\newblock \emph{The Annals of Applied Probability}, 2019.

\bibitem[Birgin and Martínez(2008)]{auglag2}
E.G. Birgin and J.M. Martínez.
\newblock Improving ultimate convergence of an augmented lagrangian method.
\newblock \emph{Optimization Methods and Software}, 23\penalty0 (2):\penalty0
  177--195, 2008.
\newblock \doi{10.1080/10556780701577730}.
\newblock URL \url{https://doi.org/10.1080/10556780701577730}.

\bibitem[Blasi and Walker(2013)]{10.2307/24310519}
Pierpaolo~De Blasi and Stephen~G. Walker.
\newblock Bayesian asymptotics with misspecified models.
\newblock \emph{Statistica Sinica}, 23\penalty0 (1):\penalty0 169--187, 2013.
\newblock ISSN 10170405, 19968507.
\newblock URL \url{http://www.jstor.org/stable/24310519}.

\bibitem[Cai et~al.(2020{\natexlab{a}})Cai, Campbell, and Broderick]{cai}
Diana Cai, Trevor Campbell, and Tamara Broderick.
\newblock Finite mixture models do not reliably learn the number of components,
  2020{\natexlab{a}}.
\newblock URL \url{https://arxiv.org/abs/2007.04470}.

\bibitem[Cai et~al.(2020{\natexlab{b}})Cai, Campbell, and Broderick]{cai2}
Diana Cai, Trevor Campbell, and Tamara Broderick.
\newblock Power posteriors do not reliably learn the number of components in a
  finite mixture, 2020{\natexlab{b}}.
\newblock URL \url{https://openreview.net/pdf?id=BRb4tLp6A3o}.

\bibitem[Cambanis et~al.(1981)Cambanis, Huang, and Simons]{Cambanis1981OnTT}
Stamatis Cambanis, Steel~T. Huang, and Gordon Simons.
\newblock On the theory of elliptically contoured distributions.
\newblock \emph{Journal of Multivariate Analysis}, 11:\penalty0 368--385, 1981.

\bibitem[Chen et~al.(2019)Chen, Yu, and
  Haskell]{doi:10.1080/02331934.2019.1655738}
Zhi Chen, Pengqian Yu, and William~B. Haskell.
\newblock Distributionally robust optimization for sequential decision-making.
\newblock \emph{Optimization}, 68\penalty0 (12):\penalty0 2397--2426, 2019.
\newblock \doi{10.1080/02331934.2019.1655738}.
\newblock URL \url{https://doi.org/10.1080/02331934.2019.1655738}.

\bibitem[Chernozhukov and
  Hong(2003)]{RePEc:eee:econom:v:115:y:2003:i:2:p:293-346}
Victor Chernozhukov and Han Hong.
\newblock {An MCMC approach to classical estimation}.
\newblock \emph{Journal of Econometrics}, 115\penalty0 (2):\penalty0 293--346,
  August 2003.
\newblock URL
  \url{https://ideas.repec.org/a/eee/econom/v115y2003i2p293-346.html}.

\bibitem[Chib and Jeliazkov(2001)]{doi:10.1198/016214501750332848}
Siddhartha Chib and Ivan Jeliazkov.
\newblock Marginal likelihood from the metropolis–hastings output.
\newblock \emph{Journal of the American Statistical Association}, 96\penalty0
  (453):\penalty0 270--281, 2001.
\newblock \doi{10.1198/016214501750332848}.
\newblock URL \url{https://doi.org/10.1198/016214501750332848}.

\bibitem[Chib et~al.(2018)Chib, Shin, and Simoni]{Chib2018}
Siddhartha Chib, Minchul Shin, and Anna Simoni.
\newblock {B}ayesian estimation and comparison of moment condition models.
\newblock \emph{Journal of the American Statistical Association}, 113\penalty0
  (524):\penalty0 1656--1668, 2018.
\newblock \doi{10.1080/01621459.2017.1358172}.
\newblock URL \url{https://doi.org/10.1080/01621459.2017.1358172}.

\bibitem[Chib et~al.(2021)Chib, Shin, and Simoni]{chib2021bayesian}
Siddhartha Chib, Minchul Shin, and Anna Simoni.
\newblock {B}ayesian estimation and comparison of conditional moment models,
  2021.
\newblock URL \url{https://arxiv.org/abs/2110.13531}.

\bibitem[Conn et~al.(1991)Conn, Gould, and Toint]{auglag1}
Andrew~R. Conn, Nicholas I.~M. Gould, and Philippe Toint.
\newblock A globally convergent augmented lagrangian algorithm for optimization
  with general constraints and simple bounds.
\newblock \emph{SIAM Journal on Numerical Analysis}, 28\penalty0 (2):\penalty0
  545--572, 1991.
\newblock \doi{10.1137/0728030}.
\newblock URL \url{https://doi.org/10.1137/0728030}.

\bibitem[Cuturi(2013)]{cuturi2013sinkhorn}
Marco Cuturi.
\newblock Sinkhorn distances: Lightspeed computation of optimal transport.
\newblock In C.J. Burges, L.~Bottou, M.~Welling, Z.~Ghahramani, and K.Q.
  Weinberger, editors, \emph{Advances in Neural Information Processing
  Systems}, volume~26. Curran Associates, Inc., 2013.
\newblock URL
  \url{https://proceedings.neurips.cc/paper/2013/file/af21d0c97db2e27e13572cbf59eb343d-Paper.pdf}.

\bibitem[Delon and Desolneux(2020)]{delon:hal-02178204}
Julie Delon and Agn{\`e}s Desolneux.
\newblock {A {W}asserstein-type distance in the space of Gaussian Mixture
  Models}.
\newblock \emph{{SIAM Journal on Imaging Sciences}}, 13\penalty0 (2):\penalty0
  936--970, 2020.
\newblock URL \url{https://hal.archives-ouvertes.fr/hal-02178204}.

\bibitem[Du and Wu(2021)]{https://doi.org/10.48550/arxiv.2105.11570}
Wei Du and Xintao Wu.
\newblock Robust fairness-aware learning under sample selection bias, 2021.
\newblock URL \url{https://arxiv.org/abs/2105.11570}.

\bibitem[Dwork and Lei(2009)]{Dwork09differentialprivacy}
Cynthia Dwork and Jing Lei.
\newblock Differential privacy and robust statistics.
\newblock \emph{STOC '09: Proceedings of the forty-first annual ACM symposium
  on Theory of computing}, pages 371--380, 2009.
\newblock URL \url{https://dl.acm.org/doi/10.1145/1536414.1536466}.

\bibitem[Escobar and West(1995)]{EscobarWest}
Michael~D. Escobar and Mike West.
\newblock {B}ayesian density estimation and inference using mixtures.
\newblock \emph{Journal of the American Statistical Association}, 90\penalty0
  (430):\penalty0 577--588, 1995.
\newblock \doi{10.1080/01621459.1995.10476550}.
\newblock URL
  \url{https://www.tandfonline.com/doi/abs/10.1080/01621459.1995.10476550}.

\bibitem[Ferguson(1973)]{Ferguson}
Thomas~S. Ferguson.
\newblock {A {B}ayesian Analysis of Some Nonparametric Problems}.
\newblock \emph{The Annals of Statistics}, 1\penalty0 (2):\penalty0 209 -- 230,
  1973.
\newblock \doi{10.1214/aos/1176342360}.
\newblock URL \url{https://doi.org/10.1214/aos/1176342360}.

\bibitem[Fiksel et~al.(2021)Fiksel, Datta, Amouzou, and
  Zeger]{doi:10.1080/01621459.2021.1909599}
Jacob Fiksel, Abhirup Datta, Agbessi Amouzou, and Scott Zeger.
\newblock Generalized {B}ayes quantification learning under dataset shift.
\newblock \emph{Journal of the American Statistical Association}, 0\penalty0
  (0):\penalty0 1--19, 2021.
\newblock \doi{10.1080/01621459.2021.1909599}.
\newblock URL \url{https://doi.org/10.1080/01621459.2021.1909599}.

\bibitem[Fitzsimons et~al.(2019)Fitzsimons, Al~Ali, Osborne, and
  Roberts]{e21080741}
Jack Fitzsimons, AbdulRahman Al~Ali, Michael Osborne, and Stephen Roberts.
\newblock A general framework for fair regression.
\newblock \emph{Entropy}, 21\penalty0 (8), 2019.
\newblock ISSN 1099-4300.
\newblock \doi{10.3390/e21080741}.
\newblock URL \url{https://www.mdpi.com/1099-4300/21/8/741}.

\bibitem[Gajane and Pechenizkiy(2018)]{gajane2018formalizing}
Pratik Gajane and Mykola Pechenizkiy.
\newblock On formalizing fairness in prediction with machine learning, 2018.
\newblock URL
  \url{https://www.fatml.org/media/documents/formalizing_fairness_in_prediction_with_ml.pdf}.

\bibitem[Grant and Boyd(2008)]{gb08}
Michael Grant and Stephen Boyd.
\newblock Graph implementations for nonsmooth convex programs.
\newblock In V.~Blondel, S.~Boyd, and H.~Kimura, editors, \emph{Recent Advances
  in Learning and Control}, Lecture Notes in Control and Information Sciences,
  pages 95--110. Springer-Verlag Limited, 2008.
\newblock \url{http://stanford.edu/~boyd/graph_dcp.html}.

\bibitem[Gratch et~al.(2014)Gratch, Artstein, Lucas, Stratou, Scherer,
  Nazarian, Wood, Boberg, DeVault, Marsella, Traum, Rizzo, and
  Morency]{gratch-etal-2014-distress}
Jonathan Gratch, Ron Artstein, Gale Lucas, Giota Stratou, Stefan Scherer,
  Angela Nazarian, Rachel Wood, Jill Boberg, David DeVault, Stacy Marsella,
  David Traum, Skip Rizzo, and Louis-Philippe Morency.
\newblock The distress analysis interview corpus of human and computer
  interviews.
\newblock In \emph{Proceedings of the Ninth International Conference on
  Language Resources and Evaluation ({LREC}'14)}, pages 3123--3128, Reykjavik,
  Iceland, May 2014. European Language Resources Association (ELRA).
\newblock URL
  \url{http://www.lrec-conf.org/proceedings/lrec2014/pdf/508_Paper.pdf}.

\bibitem[Grünwald and van Ommen(2017)]{safebayes}
Peter Grünwald and Thijs van Ommen.
\newblock Inconsistency of {B}ayesian inference for misspecified linear models,
  and a proposal for repairing it.
\newblock \emph{Bayesian analysis}, 2017.
\newblock URL \url{https://pure.uva.nl/ws/files/22184651/1510974325.pdf}.

\bibitem[Han et~al.(2018)Han, Yao, Yu, Niu, Xu, Hu, Tsang, and
  Sugiyama]{han2018coteaching}
Bo~Han, Quanming Yao, Xingrui Yu, Gang Niu, Miao Xu, Weihua Hu, Ivor Tsang, and
  Masashi Sugiyama.
\newblock Co-teaching: Robust training of deep neural networks with extremely
  noisy labels.
\newblock In \emph{NeurIPS}, pages 8535--8545, 2018.

\bibitem[Haynes(2013)]{Haynes2013}
Winston Haynes.
\newblock \emph{Maximum Likelihood Estimation}, pages 1190--1191.
\newblock Springer New York, New York, NY, 2013.
\newblock ISBN 978-1-4419-9863-7.
\newblock \doi{10.1007/978-1-4419-9863-7_1235}.
\newblock URL \url{https://doi.org/10.1007/978-1-4419-9863-7_1235}.

\bibitem[Hoff and Wakefield(2012)]{sandwitch}
Peter Hoff and Jon Wakefield.
\newblock {B}ayesian sandwich posteriors for pseudo-true parameters, 2012.
\newblock URL \url{https://arxiv.org/abs/1211.0087}.

\bibitem[Holmes and Walker(2017)]{10.1093/biomet/asx010}
C.~C. Holmes and S.~G. Walker.
\newblock {Assigning a value to a power likelihood in a general {B}ayesian
  model}.
\newblock \emph{Biometrika}, 104\penalty0 (2):\penalty0 497--503, 03 2017.
\newblock ISSN 0006-3444.
\newblock \doi{10.1093/biomet/asx010}.
\newblock URL \url{https://doi.org/10.1093/biomet/asx010}.

\bibitem[Holzmann et~al.(2006)Holzmann, Munk, and Gneiting]{10.2307/4616956}
Hajo Holzmann, Axel Munk, and Tilmann Gneiting.
\newblock Identifiability of finite mixtures of elliptical distributions.
\newblock \emph{Scandinavian Journal of Statistics}, 33\penalty0 (4):\penalty0
  753--763, 2006.
\newblock ISSN 03036898, 14679469.
\newblock URL \url{http://www.jstor.org/stable/4616956}.

\bibitem[Hooker and Vidyashankar(2011)]{hooker2012bayesian}
Giles Hooker and Anand Vidyashankar.
\newblock {B}ayesian model robustness via disparities, 2011.
\newblock URL \url{https://arxiv.org/abs/1112.4213}.

\bibitem[Huber(2011)]{huber2011robust}
Peter~J Huber.
\newblock Robust statistics.
\newblock In \emph{International encyclopedia of statistical science}, pages
  1248--1251. Springer, 2011.

\bibitem[Ishwaran and Zarepour(2000)]{IZ}
H~Ishwaran and M~Zarepour.
\newblock {Markov chain Monte Carlo in approximate Dirichlet and beta
  two-parameter process hierarchical models}.
\newblock \emph{Biometrika}, 87\penalty0 (2):\penalty0 371--390, 06 2000.
\newblock ISSN 0006-3444.
\newblock \doi{10.1093/biomet/87.2.371}.
\newblock URL \url{https://doi.org/10.1093/biomet/87.2.371}.

\bibitem[Ishwaran and Zarepour(2002{\natexlab{a}})]{ishwaran2002dirichlet}
Hemant Ishwaran and Mahmoud Zarepour.
\newblock Dirichlet prior sieves in finite normal mixtures.
\newblock \emph{Statistica Sinica}, pages 941--963, 2002{\natexlab{a}}.

\bibitem[Ishwaran and Zarepour(2002{\natexlab{b}})]{ishwaran2002exact}
Hemant Ishwaran and Mahmoud Zarepour.
\newblock Exact and approximate sum representations for the dirichlet process.
\newblock \emph{Canadian Journal of Statistics}, 30\penalty0 (2):\penalty0
  269--283, 2002{\natexlab{b}}.

\bibitem[Jebara et~al.(2004)Jebara, Kondor, and
  Howard]{Jebara2004ProbabilityPK}
Tony Jebara, Risi Kondor, and Andrew~G. Howard.
\newblock Probability product kernels.
\newblock \emph{J. Mach. Learn. Res.}, 5:\penalty0 819--844, 2004.

\bibitem[Jiang et~al.(2020)Jiang, Pacchiano, Stepleton, Jiang, and
  Chiappa]{pmlr-v115-jiang20a}
Ray Jiang, Aldo Pacchiano, Tom Stepleton, Heinrich Jiang, and Silvia Chiappa.
\newblock {W}asserstein fair classification.
\newblock In Ryan~P. Adams and Vibhav Gogate, editors, \emph{Proceedings of The
  35th Uncertainty in Artificial Intelligence Conference}, volume 115 of
  \emph{Proceedings of Machine Learning Research}, pages 862--872. PMLR, 22--25
  Jul 2020.
\newblock URL \url{https://proceedings.mlr.press/v115/jiang20a.html}.

\bibitem[Jiang and Tanner(2008)]{2008}
Wenxin Jiang and Martin~A. Tanner.
\newblock Gibbs posterior for variable selection in high-dimensional
  classification and data mining.
\newblock \emph{The Annals of Statistics}, 36\penalty0 (5), Oct 2008.
\newblock ISSN 0090-5364.
\newblock \doi{10.1214/07-aos547}.
\newblock URL \url{http://dx.doi.org/10.1214/07-AOS547}.

\bibitem[Johnson(2022)]{nlopt}
Steven~G. Johnson.
\newblock The nlopt nonlinear-optimization package.
\newblock \emph{The Comprehensive R Archive Network}, 2022.

\bibitem[Kleijn and van~der Vaart(2006)]{10.1214/009053606000000029}
B.~J.~K. Kleijn and A.~W. van~der Vaart.
\newblock {Misspecification in infinite-dimensional Bayesian statistics}.
\newblock \emph{The Annals of Statistics}, 34\penalty0 (2):\penalty0 837 --
  877, 2006.
\newblock \doi{10.1214/009053606000000029}.
\newblock URL \url{https://doi.org/10.1214/009053606000000029}.

\bibitem[Kleijn and van~der Vaart(2012)]{kleijn2012bernstein}
Bas~JK Kleijn and Aad~W van~der Vaart.
\newblock The bernstein-von-mises theorem under misspecification.
\newblock \emph{Electronic Journal of Statistics}, 6:\penalty0 354--381, 2012.

\bibitem[Lavine(1994)]{10.1214/aos/1176325623}
Michael Lavine.
\newblock {More Aspects of Polya Tree Distributions for Statistical Modelling}.
\newblock \emph{The Annals of Statistics}, 22\penalty0 (3):\penalty0 1161 --
  1176, 1994.
\newblock \doi{10.1214/aos/1176325623}.
\newblock URL \url{https://doi.org/10.1214/aos/1176325623}.

\bibitem[Lazar(2003)]{10.2307/30042042}
Nicole~A. Lazar.
\newblock {B}ayesian empirical likelihood.
\newblock \emph{Biometrika}, 90\penalty0 (2):\penalty0 319--326, 2003.
\newblock ISSN 00063444.
\newblock URL \url{http://www.jstor.org/stable/30042042}.

\bibitem[Le et~al.(2019)Le, Yamada, Fukumizu, and Cuturi]{le2019treesliced}
Tam Le, Makoto Yamada, Kenji Fukumizu, and Marco Cuturi.
\newblock Tree-sliced variants of {W}asserstein distances.
\newblock In \emph{Advances in Neural Information Processing Systems}. Curran
  Associates, Inc., 2019.
\newblock URL
  \url{https://proceedings.neurips.cc/paper/2019/file/2d36b5821f8affc6868b59dfc9af6c9f-Paper.pdf}.

\bibitem[Levin et~al.(2006)Levin, Peres, and Wilmer]{LevinPeresWilmer2006}
David~A. Levin, Yuval Peres, and Elizabeth~L. Wilmer.
\newblock \emph{{Markov chains and mixing times}}.
\newblock American Mathematical Society, 2006.
\newblock URL
  \url{http://scholar.google.com/scholar.bib?q=info:3wf9IU94tyMJ:scholar.google.com/&output=citation&hl=en&as_sdt=2000&ct=citation&cd=0}.

\bibitem[Liu et~al.(2021)Liu, Kong, and
  Oh]{https://doi.org/10.48550/arxiv.2111.06578}
Xiyang Liu, Weihao Kong, and Sewoong Oh.
\newblock Differential privacy and robust statistics in high dimensions, 2021.
\newblock URL \url{https://arxiv.org/abs/2111.06578}.

\bibitem[Maria(1965)]{10.2307/2315957}
A.~J. Maria.
\newblock A remark on stirling's formula.
\newblock \emph{The American Mathematical Monthly}, 72\penalty0 (10):\penalty0
  1096--1098, 1965.
\newblock ISSN 00029890, 19300972.
\newblock URL \url{http://www.jstor.org/stable/2315957}.

\bibitem[McAuliffe et~al.(2006)McAuliffe, Blei, and Jordan]{McAuliffe}
McAuliffe, Blei, and Jordan.
\newblock {Nonparametric empirical {B}ayes for the Dirichlet process mixture
  model}.
\newblock \emph{Stat Comput}, 1\penalty0 (2):\penalty0 5--14, 2006.

\bibitem[Miller and Dunson(2019)]{doi:10.1080/01621459.2018.1469995}
Jeffrey~W. Miller and David~B. Dunson.
\newblock Robust bayesian inference via coarsening.
\newblock \emph{Journal of the American Statistical Association}, 114\penalty0
  (527):\penalty0 1113--1125, 2019.
\newblock \doi{10.1080/01621459.2018.1469995}.
\newblock URL \url{https://doi.org/10.1080/01621459.2018.1469995}.
\newblock PMID: 31942084.

\bibitem[Miller and Harrison(2018)]{doi:10.1080/01621459.2016.1255636}
Jeffrey~W. Miller and Matthew~T. Harrison.
\newblock Mixture models with a prior on the number of components.
\newblock \emph{Journal of the American Statistical Association}, 113\penalty0
  (521):\penalty0 340--356, 2018.
\newblock \doi{10.1080/01621459.2016.1255636}.
\newblock URL \url{https://doi.org/10.1080/01621459.2016.1255636}.
\newblock PMID: 29983475.

\bibitem[Minsker et~al.(2017)Minsker, Srivastava, Lin, and
  Dunson]{JMLR:v18:16-655}
Stanislav Minsker, Sanvesh Srivastava, Lizhen Lin, and David~B. Dunson.
\newblock Robust and scalable {B}ayes via a median of subset posterior
  measures.
\newblock \emph{Journal of Machine Learning Research}, 18\penalty0
  (124):\penalty0 1--40, 2017.
\newblock URL \url{http://jmlr.org/papers/v18/16-655.html}.

\bibitem[Muirhead(2005)]{muirhead2005aspects}
Robb~J. Muirhead.
\newblock \emph{Aspects of Multivariate Statistical Theory}.
\newblock Wiley-Interscience, 2005.

\bibitem[M{\"u}ller and Quintana(2004)]{muller2004nonparametric}
Peter M{\"u}ller and Fernando~A Quintana.
\newblock Nonparametric bayesian data analysis.
\newblock \emph{Statistical science}, 19\penalty0 (1):\penalty0 95--110, 2004.

\bibitem[M{\"u}ller et~al.(2015)M{\"u}ller, Quintana, Jara, and
  Hanson]{muller2015bayesian}
Peter M{\"u}ller, Fernando~Andr{\'e}s Quintana, Alejandro Jara, and Tim Hanson.
\newblock \emph{Bayesian nonparametric data analysis}, volume~1.
\newblock Springer, 2015.

\bibitem[Muzellec and Cuturi(2018)]{muzellec2019generalizing}
Boris Muzellec and Marco Cuturi.
\newblock Generalizing point embeddings using the {W}asserstein space of
  elliptical distributions.
\newblock In \emph{Advances in Neural Information Processing Systems},
  volume~31. Curran Associates, Inc., 2018.
\newblock URL
  \url{https://proceedings.neurips.cc/paper/2018/file/b613e70fd9f59310cf0a8d33de3f2800-Paper.pdf}.

\bibitem[Nabi and Shpitser(2018)]{Nabi2018}
Razieh Nabi and Ilya Shpitser.
\newblock Fair inference on outcomes.
\newblock \emph{Proceedings of the AAAI Conference on Artificial Intelligence},
  32\penalty0 (1), Apr. 2018.
\newblock \doi{10.1609/aaai.v32i1.11553}.
\newblock URL \url{https://ojs.aaai.org/index.php/AAAI/article/view/11553}.

\bibitem[Owen(2001)]{owen2001empirical}
Art~B Owen.
\newblock \emph{Empirical likelihood}.
\newblock Chapman and Hall/CRC, 2001.

\bibitem[Panaretos and Zemel(2019)]{2019}
Victor~M. Panaretos and Yoav Zemel.
\newblock Statistical aspects of {W}asserstein distances.
\newblock \emph{Annual Review of Statistics and Its Application}, 6\penalty0
  (1):\penalty0 405–431, Mar 2019.
\newblock ISSN 2326-831X.
\newblock \doi{10.1146/annurev-statistics-030718-104938}.
\newblock URL \url{http://dx.doi.org/10.1146/annurev-statistics-030718-104938}.

\bibitem[Pinsker(1964)]{alma991023405949705251}
M.~S. Pinsker.
\newblock \emph{Information and information stability of random variables and
  processes / by M.S. Pinsker. Translated and edited by Amiel Feinstein.}
\newblock Holden-Day series in time series analysis. Holden-Day, San Francisco,
  1964.

\bibitem[{R Core Team}(2022)]{RCore}
{R Core Team}.
\newblock \emph{R: A Language and Environment for Statistical Computing}.
\newblock R Foundation for Statistical Computing, Vienna, Austria, 2022.
\newblock URL \url{https://www.R-project.org/}.

\bibitem[Resnick(2013)]{Resnick}
Sydney Resnick.
\newblock A probability path.
\newblock \emph{Birkhäuser Boston}, 2013.

\bibitem[Santambrogio(2015)]{noauthororeditor}
Filippo Santambrogio.
\newblock Optimal transport for applied mathematicians. calculus of variations,
  pdes and modeling, 2015.
\newblock URL \url{https://www.math.u-psud.fr/~filippo/OTAM-cvgmt.pdf}.

\bibitem[Schennach(2005)]{10.1093/biomet/92.1.31}
Susanne~M. Schennach.
\newblock {Bayesian exponentially tilted empirical likelihood}.
\newblock \emph{Biometrika}, 92\penalty0 (1):\penalty0 31--46, 03 2005.
\newblock ISSN 0006-3444.
\newblock \doi{10.1093/biomet/92.1.31}.
\newblock URL \url{https://doi.org/10.1093/biomet/92.1.31}.

\bibitem[Shafahi et~al.(2020)Shafahi, Saadatpanah, Zhu, Ghiasi, Studer, Jacobs,
  and Goldstein]{Shafahi2020Adversarially}
Ali Shafahi, Parsa Saadatpanah, Chen Zhu, Amin Ghiasi, Christoph Studer, David
  Jacobs, and Tom Goldstein.
\newblock Adversarially robust transfer learning.
\newblock In \emph{International Conference on Learning Representations}, 2020.
\newblock URL \url{https://openreview.net/forum?id=ryebG04YvB}.

\bibitem[Teh(2010)]{teh2010dirichlet}
Yee~Whye Teh.
\newblock Dirichlet process.
\newblock \emph{Encyclopedia of machine learning}, 1063:\penalty0 280--287,
  2010.

\bibitem[Vehtari et~al.(2016)Vehtari, Mononen, Tolvanen, Sivula, and
  Winther]{JMLR:v17:14-540}
Aki Vehtari, Tommi Mononen, Ville Tolvanen, Tuomas Sivula, and Ole Winther.
\newblock {B}ayesian leave-one-out cross-validation approximations for gaussian
  latent variable models.
\newblock \emph{Journal of Machine Learning Research}, 17\penalty0
  (103):\penalty0 1--38, 2016.
\newblock URL \url{http://jmlr.org/papers/v17/14-540.html}.

\bibitem[Verdinelli and Wasserman(1998)]{10.1214/aos/1024691240}
Isabella Verdinelli and Larry Wasserman.
\newblock {{B}ayesian goodness-of-fit testing using infinite-dimensional
  exponential families}.
\newblock \emph{The Annals of Statistics}, 26\penalty0 (4):\penalty0 1215 --
  1241, 1998.
\newblock \doi{10.1214/aos/1024691240}.
\newblock URL \url{https://doi.org/10.1214/aos/1024691240}.

\bibitem[Villani(2003)]{Villani2003TopicsIO}
C{\'e}dric Villani.
\newblock Topics in optimal transportation.
\newblock \emph{American Mathematical Society}, 2003.
\newblock URL \url{https://www.math.ucla.edu/~wgangbo/Cedric-Villani.pdf}.

\bibitem[Vilnis and McCallum(2014)]{https://doi.org/10.48550/arxiv.1412.6623}
Luke Vilnis and Andrew McCallum.
\newblock Word representations via gaussian embedding, 2014.
\newblock URL \url{https://arxiv.org/abs/1412.6623}.

\bibitem[Wang et~al.(2020{\natexlab{a}})Wang, Guo, Narasimhan, Cotter, Gupta,
  and Jordan]{NEURIPS2020_37d097ca}
Serena Wang, Wenshuo Guo, Harikrishna Narasimhan, Andrew Cotter, Maya Gupta,
  and Michael Jordan.
\newblock Robust optimization for fairness with noisy protected groups.
\newblock In H.~Larochelle, M.~Ranzato, R.~Hadsell, M.F. Balcan, and H.~Lin,
  editors, \emph{Advances in Neural Information Processing Systems}, volume~33,
  pages 5190--5203. Curran Associates, Inc., 2020{\natexlab{a}}.
\newblock URL
  \url{https://proceedings.neurips.cc/paper/2020/file/37d097caf1299d9aa79c2c2b843d2d78-Paper.pdf}.

\bibitem[Wang et~al.(2020{\natexlab{b}})Wang, Hu, and Hu]{9156647}
Zhen Wang, Guosheng Hu, and Qinghua Hu.
\newblock Training noise-robust deep neural networks via meta-learning.
\newblock In \emph{2020 IEEE/CVF Conference on Computer Vision and Pattern
  Recognition (CVPR)}, pages 4523--4532, 2020{\natexlab{b}}.
\newblock \doi{10.1109/CVPR42600.2020.00458}.

\bibitem[Xu and Mannor(2010)]{NIPS2010_19f3cd30}
Huan Xu and Shie Mannor.
\newblock Distributionally robust markov decision processes.
\newblock In J.~Lafferty, C.~Williams, J.~Shawe-Taylor, R.~Zemel, and
  A.~Culotta, editors, \emph{Advances in Neural Information Processing
  Systems}, volume~23. Curran Associates, Inc., 2010.
\newblock URL
  \url{https://proceedings.neurips.cc/paper/2010/file/19f3cd308f1455b3fa09a282e0d496f4-Paper.pdf}.

\bibitem[Yang et~al.(2019)Yang, Lafferty, and Pollard]{yang2019fair}
Dana Yang, John Lafferty, and David Pollard.
\newblock Fair quantile regression, 2019.
\newblock URL \url{https://arxiv.org/abs/1907.08646}.

\end{thebibliography}
\bibliographystyle{plainnat}
\setcitestyle{notesep={; },round,aysep={},yysep={;}}

\clearpage

\beginsupplement
\begin{center}
\section*{Supplementary material to \\
``Robust probabilistic inference via a constrained transport metric" }\label{SM}
 Abhisek Chakraborty, Anirban Bhattacharya, Debdeep Pati\\
 Department of Statistics, Texas A\&M University, College Station, TX, USA
\end{center}
Section \ref{ssec:th1th2_proof} contains proofs of Theorems \ref{th2:lemma2} and \ref{th:npbayes_equivalence} in Section \ref{ssec:npBayes} in the main document, relating to the non-parametric Bayes interpretation of our semi-parametric methodology D-BETEL. Section \ref{ssec:th3} contains the proof of Theorem \ref{th3} in Section \ref{ssec:andrew} in the main document, concerning the tailor-made transport metric ANDREW. Sections \ref{aux:th1th2} and \ref{aux:th3} record auxiliary results for Theorems \ref{th2:lemma2}-\ref{th:npbayes_equivalence} and Theorem \ref{th3}, respectively. Section \ref{sup:glm} presents additional numerical results on the generalised linear regression example discussed in Section \ref{ssec:gen_reg} in the main document. All bibliographical references can be found in the main document. 
\vspace{0.3in}



\section{Proofs of Theorems \ref{th2:lemma2} and \ref{th:npbayes_equivalence} in the main document}\label{ssec:th1th2_proof}
\subsection{Preliminaries}
Let us first set up some notation to facilitate the proofs of Theorems \ref{th2:lemma2} and \ref{th:npbayes_equivalence}. 
First, observe that the prior on $(b_1,\ldots, b_k)$ in equation block \eqref{eqn:npbayes1} is $\mbox{Multinomial}(N;1/k, \ldots, 1/k)$ with probability mass function $N!/(k^N \prod_{h=1}^k b_h!)$.  Next, we re-parameterize $(b_1,\ldots,b_k)\to (w_1,\ldots,w_k)$ and define 
\begin{align}\label{eqn_W_j}
    &\mb{W}_j = \bigg\{w: w_{h} = \frac{b_{h}}{N},\ h=1,\ldots,j;\  b_h \in \mb{Z}^{+},\ \sum_{h=1}^j b_h = N\ \bigg\},\quad j\in\mb{Z}^{+},
\end{align}
where $\mb{Z}^{+}$ is set of all positive integers. 
Also, define
\begin{align}\label{eqn_W_j_Tilde}
   &\wt{\mb{W}}_j=\bigg\{w\in\mb{W}_j:\big|w_h - w^{\star}_h\big| < \frac{1}{N^{1-\delta}}, \forall\ h=1,\ldots,j \bigg\}, \quad j\in\mb{Z}^{+},
\end{align}
where $w^{\star}_h, \ h=1,\ldots, j$ is the solution to the $\mbox{D}$-BETEL optimization without the parametric constraint, and we assume that $1/2<\delta<1$. The proof of Theorem \ref{th2:lemma2} involves studying a log ratio of the form
\begin{align}\label{eqn_R_j}
R_j = &\ \log\ \Bigg\{\frac{\prod_{h=1}^j {w^{\star}_{h}}^{-w_{h}^{\star} N + \frac{1}{2}}}{\prod_{h=1}^j w_{h}^{-w_{h} N + \frac{1}{2}}}\Bigg\}
= N\big[ H_{N}(w^{\star}) - H_{N}(w)\big] + \frac{1}{2} \sum_{h=1}^j \log\bigg(\frac{w^{\star}_{h}}{w_h}\bigg), \quad j\in\mb{Z}^{+},
\end{align}
over $\mb{W}_j\setminus\wt{\mb{W}}_j$, where $H_N(w) = -\sum_{h=1}^j w_h\log w_h$.

Further, 
we follow the definition in \eqref{eqn:andrewdef} to introduce 
\begin{align}\label{eqn_C}
 &C_{\theta,\varepsilon} =
\bigg \{w\in\mb{W}_n:\ \mbox{D}[F_{\theta}, \nu(w,x)]  \leq \varepsilon \bigg\},
\end{align}
In the constrained $\mbox{D}$-BETEL formulation, for every $\theta\in\Theta$: 
\begin{align}\label{eqn_DBETEL}
  w^{\star}_{1:n}(\theta,\varepsilon) = \argmax_{w\in C_{\theta,\varepsilon}} H_N(w),
\end{align}
and $\varepsilon>0$. With slight abuse of notations, we shall use $w^{\star}_h$ in place of  $w^{\star}_h(\theta,\varepsilon)$.  In places, it is useful to resort to the dual form of the maximization problem above, introduced in equation \eqref{eqn:dbetel:alt} in the main document:
\begin{align}\label{eqn_DBETEL_dual}
    w^{\star}_{1:n}(\theta,\varepsilon) = \argmax_{w\in \mb{W}_n} g(w)\quad \text{where}\quad g(w) = H_N(w) - \lambda_{\star}\ \mbox{D}^2(F_{\theta} ,  \nu(w, x)),
\end{align}
where there exists a $\lambda_{\star}>0$ for every choice of $\varepsilon>0$. The proof of Theorem \ref{th:npbayes_equivalence} hinges on studying the log ratio of the form $R_n = N\big[ H_{N}(w^{\star}) - H_N(w)\big] + \frac{1}{2} \sum_{h=1}^n \log(w^{\star}_{h}/w_h)$, defined in equation \eqref{eqn_R_j}.
In order to study the behaviour of $R_n$ in a carefully constructed neighbourhood of $w^{\star}$ and outside it, we introduce the set 
\begin{align}\label{eqn_C_Tilde}
\wt{C}_{\theta, \varepsilon}=
\bigg\{w\in C_{\theta, \varepsilon}:\  
&\mbox{D}^2[F_{\theta} , \nu(w^{\star}, x) ] - \frac{1}{N^{1-\delta}}\leq \mbox{D}^2[F_{\theta} , \nu(w, x) ],\notag\\
&0\leq H_{N}(w^{\star}) - H_{N}(w)\leq \frac{1}{N^{1-\delta}} \bigg\}
\end{align}
for some $0<\delta<1$. In this article, we focus on the case  $\mbox{D}\equiv W_{\rm AR}$ (refer to equations \eqref{eqn_C}, \eqref{eqn_DBETEL_dual}, \eqref{eqn_C_Tilde}) for the sake the proof of Theorem \ref{th:npbayes_equivalence}. More precisely, we utilize $\mbox{D}\equiv W_{\rm AR}$ in the proofs of Lemmas \ref{th3:lemma1}-\ref{th2_lemma2_asymtotic} that leads to the proof of Theorem \ref{th:npbayes_equivalence}. However, we envision that the proof steps hold true for a more general class of distance metric $\mbox{D}$, of which ANDREW in Section \ref{ssec:andrew} is a special case.

Before we proceed further, we briefly discuss  the intuitions behind the choice of  $\wt{C}_{\theta, \varepsilon}$ in the proof of Theorem \ref{th:npbayes_equivalence}, especially in connection to a similar approach taken in the proof of the main theorem  regarding non-parametric Bayes interpretation of Bayesian exponentially tilted empirical likelihood for moment conditional model in \citet{10.1093/biomet/92.1.31}. Given a random sample $x = (x_1,\ldots, x_n)^\T$ from an unknown data generating distribution $P$ on $\mathbb{R}^d$, the exponentially tilted empirical likelihood for moment conditional model takes the form:
\begin{equation}
  L_{\rm MCM}(\theta) = \bigg\{\prod_{i=1}^n w_{i} :  \argmax_{w}\prod_{i=1}^n w_{i}^{-w_{i}},\ w_i>0,\ \sum_{i=1}^n w_i = 1, \ \sum_{i=1}^n w_i\ u(x_i,\theta)=0  \bigg\},\quad \theta\in\Theta, 
\end{equation} 
where $u:\mathbb{R}^d\times\Theta\to\mathbb{R}$. For the above formulation, a choice of neighbourhood of the form
\begin{align}\label{C_MCM}
&{C}_{\theta,  \rm MCM}=
\bigg\{w\in \mb{W}_n: \sum_{i=1}^n w_i\ u(x_i, \theta) = 0\bigg\},\notag\\
&\wt{C}_{\theta,  \rm MCM}=
\bigg\{w\in {C}_{\theta,  \rm MCM}: |w_h - w_{h}^{\star}|<\frac{1}{N^{1-\delta}},\ h = 1,\ldots,n \bigg\} ,
\end{align}
where $w^{\star}_{1:n} = \argmax_{w\in C_{\theta, \rm MCM}} H_N(w)$ and $0<\delta<1$, will be apt since 
\begin{align}
     \bigg|\sum_{i=1}^n w_{i}^\star\ u(x_i,\theta) - \sum_{i=1}^n w_i\ u(x_i,\theta)\bigg|\leq \sum_{i=1}^n |w_{i}^\star -  w_{i}||u(x_i,\theta)| \leq \frac{1}{N^{1-\delta}}\sum_{i=1}^n |u(x_i,\theta)| \downarrow 0,
\end{align}
and $H_{N}(w^{\star}) - H_{N}(w)\to 0$ as $ N\to\infty$. Further, $\wt{C}_{\theta,  \rm MCM}$ satisfies two conditions: (i) for every $w\in \wt{C}_{\theta,  \rm MCM}$, $|\prod_{i=1}^n w_i - \prod_{i=1}^n w_{i}^{\star}|\to 0$ as $N\to\infty$, and  (ii) for every $w\in \mbox{W}_j\setminus\wt{C}_{\theta, \rm MCM}$, $N[H_N(w^*) - H_N(w)]> N^{\delta}$.

However, in the proof of our Theorem \ref{th:npbayes_equivalence} for exponentially tilted empirical likelihood with distance based constraints, a simple choice of neighbourhood of $w^{\star}$ like in equation \eqref{C_MCM} poses significant algebraic challenges to analyse the behaviour of $R_n$ in  and outside the neighbourhood. This motivates the choice of the neighborhood of the form
\begin{align}
\wt{C}_{\theta, \varepsilon}=
\bigg\{w\in C_{\theta, \varepsilon}:
&\mbox{D}^2[F_{\theta} , \nu(w, x) ]- \frac{1}{N^{1-\delta}}\leq \mbox{D}^2[F_{\theta} , \nu(w, x) ],\notag\\
&0\leq H_{N}(w^{\star}) - H_{N}(w)\leq \frac{1}{N^{1-\delta}} \bigg\}.
\end{align}
We focus on $\mbox{D}\equiv W_{\rm AR}$ and demonstrate that $\wt{C}_{\theta, \varepsilon}$ satisfies the two conditions, i.e (i) for every $w\in \wt{C}_{\theta, \varepsilon}$, $|\prod_{i=1}^n w_i - \prod_{i=1}^n w_{i}^{\star}|\to 0$ as $N\to\infty$, and  (ii) for every $w\in {C}_{\theta, \varepsilon}\setminus \wt{C}_{\theta, \varepsilon}$, $N[H_N(w^*) - H_N(w)]> N^{\delta}$. 
The arguments presented in the proofs of Theorem \ref{th:npbayes_equivalence} and corresponding auxiliary results can potentially be extended for other transport metrics.

Now we are in the position to present the proofs of Theorems \ref{th2:lemma2} and \ref{th:npbayes_equivalence}.

\subsection{Proof of Theorem \ref{th2:lemma2}}
\begin{proof}
For the sake of clarity, throughout the proof we shall denote $k$ as $k_N$, in order to recognise it as sequence of random variables indexed by $N$. We shall proceed to prove:
\begin{align}
    P(k_N \neq n \quad \text{infinitely often} \mid x_{1:n}) = 0,
\end{align}
which will follow from the fact that,
\begin{align}\label{1BC_condition}
  \sum_{N=N^{\star}}^{\infty} P(k_N \neq n  \mid x_{1:n}) < \infty,
\end{align}
for a finite $N^{\star} = 1 + \bigg[\big\{\min_{s\neq t} |x_s -x_t|_{\infty}\big\}^{-1/(\alpha-\beta)}\bigg]$, followed by the application of the first Borel--Cantelli Lemma. 

Our general strategy will be to begin with the joint distribution of 
\begin{align*}
  (\xi^{\star}, x_{1:n})=  (k, \ b_1,\ldots,b_k,\ \mu_{1},\ldots, \mu_{k},\ x_{1:n})  
\end{align*}
in equation blocks \eqref{eqn:npbayes3}--\eqref{eqn:npbayes2} and compute the marginal posterior of $P(k_N\mid x_{1:n})$. We shall prove equation \eqref{1BC_condition} in two parts:
$\sum_{N=N^{\star}}^{\infty} P(k_N > n  \mid x_{1:n}) < \infty,\ \text{and}\ \sum_{N=N^{\star}}^{\infty} P(k_N < n  \mid x_{1:n}) < \infty.$  

\subsection*{Part 1}
For the proof of the fact that $\sum_{N=N^{\star}}^{\infty} P(k_N >n  \mid x_{1:n}) < \infty$, we consider
\begin{align}\label{eqn:part1}
    P(k_N > n  \mid x_{1:n}) 
    &= \sum_{j=n+1}^{\infty} P(k_N =j  \mid x_{1:n})\notag\\ 
    &= \frac{ \sum_{j=n+1}^{\infty}\ p(j)\ m(x_{1:n}
    \mid k_N = j)}{\sum_{j=1}^{\infty}\ p(j)\ m(x_{1:n}
    \mid k_N = j)}
    \leq \frac{ \sum_{j=n+1}^{\infty}\ p(j)\ m(x_{1:n}
    \mid k_N = j)}{p(n)\ m(x_{1:n}
    \mid k_N = n)}\notag\\ 
   & \leq \bigg\{\frac{\sup_{j>n} m(x_{1:n}
    \mid k_N = j) }{m(x_{1:n}
    \mid k_N = n)}\bigg\} \frac{ \sum_{j=n+1}^{\infty}\ p(j)\ }{p(n)}\notag\\
    &\leq M_{j,n, N}\ \frac{(1-p_N)^{(n-1+1)}}{p_N(1-p_N)^{(n-1)} Q^{-1}_{N,n}}
    \leq  M_{j,n, N}\ \frac{(1-p_N)}{p_N M_{N}^{-N}} \notag\\
    &=   \frac{M_{j,n, N}}{ (AM_{N}^{N+1}-1)M_{N}^{-N}}
    \leq    \frac{M_{j,n, N}}{ AM_{N}-A}
    \leq   \frac{M_{j,n, N}}{AN^{\alpha d}}.
\end{align}
where 
\begin{align}\label{eqn:bounded_ratio_marginallikelihhod}
     M_{j,n, N} = \frac{\sup_{j>n} m(x_{1:n}
    \mid k_N = j) }{m(x_{1:n}
    \mid k_N = n)}.
\end{align}
The second step follows from  the fact that the marginal likelihood of data is a weighted average of marginal likelihood of data given the number of mixture components $k_N$, i.e $m(x_{1:n}) =   \int_{\xi^{\star}} \pi_{\infty,N}(\xi^{\star})P^{(N)}(x_{1:n}\mid \xi^{\star})d\xi^{\star} = \sum_{j=1}^{\infty}\ p(j)\ m(x_{1:n} \mid k_N = j) $. The third and fourth steps are trivial.  The fifth steps follows from the fact that the prior probability attached to any particular SRSWOR of size $n$ from $\mbox{H}^{(N)}$ is $p_N(1-p_N)^{(n-1)} Q^{-1}_{N,n}$, and the prior probability of the event  $k>n$ is $(1-p_N)^{(n-1+1)}$. The sixth step follows from $Q_{N,n}<M_{N}^N$ and $p_N = 1- 1/AM_{N}^{N+1}$. The seventh and the eighth step are trivial, and the last step follows from $M_N = (N^{\alpha} + 1)^d > N^{\alpha d} +1$. Finally, if we can show $M_{j,n, N} \leq K_1$ in equation \eqref{eqn:bounded_ratio_marginallikelihhod} for large enough $N$, first part of the proof would follow from equation \eqref{eqn:part1} since $\sum_{N=1}^{\infty} 1/N^{\alpha d}<\infty$ for $\alpha d > 1$. 

Next, we show that  $M_{j,n, N} \leq K_1$ for some $K_1>0$. To that end, under the hierarchical model in equation \eqref{eqn:npbayes3}--\eqref{eqn:npbayes2} in the main document, the unconstrained posterior of $\xi^{\star}\mid x_{1:n}$ is  
\begin{align}\label{unconstrained_post}
&\pi^{(free)}_{N}(\xi^{\star}\mid x_{1:n})  
 \ \propto \ \pi_{\infty, N}(\xi^{\star})\
P^{(N)}(x_{1:n}\mid \xi^{\star}) \notag\\ 
& = \bigg\{p(k)\prod_{h=1}^k \mbox{H}^{(N)}(\mu_h) \bigg\} \
\bigg\{\frac{1}{k^N}\ \frac{N!}{\prod_{h=1}^k b_h !} \bigg\}\ 
\bigg\{\prod_{i=1}^n \bigg[\sum_{h=1}^k \frac{b_h}{N} \mbox{U}_d(x_i\mid\eta_{h},\tau^{-1}I) \bigg] \bigg\},
\end{align}
where we use $ \mbox{U}_d(x_i\mid\eta_h,\tau^{-1}I)$ to denote $\prod_{j=1}^d \mbox{Uniform}(\eta_{h, j} -  \tau^{-1},\ \eta_{h, j} +  \tau^{-1})$. From the equation \eqref{unconstrained_post}, we have
\begin{align}\label{unconstrained_marglik}
    &m(x_{1:N}\mid k_N = j)\notag\\ 
    = &\sum_{w_{1:j}}\sum_{\mu_{1:j}}\prod_{h=1}^j \mbox{H}^{(N)}(\mu_h)\ \bigg\{\frac{N!/\prod_{h=1}^j b_h!}{j^N}\bigg\}\prod_{i=1}^n \bigg\{\sum_{h=1}^j w_h \tau^{d}_{N}\ \delta(|x_i - \eta_h|_{\infty} <2\rho_N)\bigg\}\notag\\
    = & \frac{\tau^{d}_{N}}{M_{N}^{j}}\sum_{w_{1:j}}\ \bigg\{\frac{N!/\prod_{h=1}^j b_h!}{j^N}\bigg\}\prod_{i=1}^n  w_i \notag\\
   = & \frac{\tau^{dn}_{N}}{M_{N}^{j}}\bigg[\sum_{w\in \wt{W}_j}\ \bigg\{\frac{N!/\prod_{h=1}^j b_h!}{j^N}\bigg\}\prod_{i=1}^n  w_i + \sum_{w\in W_j\setminus\wt{W}_j}\ \bigg\{\frac{N!/\prod_{h=1}^j b_h!}{j^N}\bigg\}\prod_{i=1}^n  w_i \bigg],
\end{align}
where $\mb{W}_j$ and $\wt{\mb{W}}_j$ are as in equations \eqref{eqn_W_j} and \eqref{eqn_W_j_Tilde}.  The first step holds true since exactly $1$ indicator in the sum  $\sum_{h=1}^j w_h \tau^{d}_{N}\ \delta(|x_i - \eta_h|_{\infty} <2\rho_N)$ corresponding to $x_i, i =1,\ldots,n$ survives as long as  $N\geq N^{\star}$ where it is sufficient to ensure
that
\begin{align}
    \min_{s\neq t} |x_s -x_t|_{\infty} > 2\rho_{N^{\star}} = \frac{1}{(N^{\star})^{\alpha-\beta}}\iff N^{\star} > \big\{\min_{s\neq t} |x_s -x_t|_{\infty}\big\}^{-1/(\alpha-\beta)}.
\end{align} 
Without loss of generality, we assume that only the $i$-th indicator in the sum  $\sum_{h=1}^j w_h \tau^{d}_{N}\ \delta(|x_i - \eta_h|_{\infty} <2\rho_N)$ corresponding to $x_i, i =1,\ldots,n$ survives. Next, we  simplify the equation \eqref{unconstrained_marglik} further, and to that end we note that:

$(i)$ The solution to the $\mbox{D}$-BETEL optimization $w^{\star}_h, \ h=1,\ldots, j$ without the parametric constraint  is $(1/j,\ldots, 1/j)^\T$, and yields
\begin{align}\label{DBETEL_max}
   \bigg\{\frac{N!/\prod_{h=1}^j b^{\star}_h!}{j^N}\bigg\}\prod_{i=1}^n  w_{i}^{\star} = \frac{1}{\sqrt{2\pi}^{j-1}}\frac{j^{-(n + j/2)}}{N^{j-1/2}}.
\end{align}
Further,  for any $w\in\mb{W}_j\setminus\wt{\mb{W}_j}$, we consider the log ratio
\begin{align}\label{eqn:ratio}
R_j = &\ \log\ \Bigg\{\frac{\prod_{h=1}^j {w^{\star}_{h}}^{-w_{h}^{\star} N + \frac{1}{2}}}{\prod_{h=1}^j w_{h}^{-w_{h} N + \frac{1}{2}}}\Bigg\}\notag\\
= &\ N \sum_{h=1}^j (w_h - w^{\star}_h)\log w^{\star}_h + N\sum_{h=1}^j w_h \log\bigg(\frac{w_h}{w^{\star}_{h}}\bigg) + \frac{1}{2} \sum_{h=1}^j \log\bigg(\frac{w^{\star}_{h}}{w_h}\bigg)\notag\\
\geq&\ N\sum_{h=1}^j w_h \log\bigg(\frac{w_h}{w^{\star}_{h}}\bigg)
=\ N\ \mbox{KL}(w\mid\mid w^{\star})
\geq\ N\times 2 \bigg\{|w-w^{\star}|_{\rm TV}\bigg\}^2\notag\\
\geq&\ N\ \bigg\{|w-w^{\star}|_{\rm TV}\bigg\}^2 
= N\ \bigg\{\sup_{h=1,\ldots j} |w_h-w^{\star}_h|\bigg\}^2  \geq N\ \bigg\{\inf_{h=1,\ldots j}
|w_h-w^{\star}_h|\bigg\}^2 \notag\\
\geq& N\ \times \bigg\{\frac{1}{N^{1-\delta}}\bigg\}^2 = N^{\delta^{\star}}.
\end{align}
where $\delta^{\star} = 2\delta - 1 >0$ since $\delta >1/2$   by assumption. The first inequality holds since  $N \sum_{h=1}^j (w_h - w^{\star}_h)\log w^{\star}_h = -N\log j \sum_{h=1}^j(w_h - w^{\star}_h) = 0$, and $\sum_{h=1}^j  \log\big(w^{\star}_{h}/w_h\big)$ is non-negative. The second inequality is due to the Lemma \ref{Pinsker}. Rest is simple algebra.
Finally, the set $\mb{W}_j\setminus\wt{\mb{W}}_j$ can contain at max $N^n$ elements. From the equations \eqref{DBETEL_max}-\eqref{eqn:ratio},  the contributions corresponding to the elements of the set $W_j\setminus\wt{W}_j$ can be bounded as follows:
\begin{align}\label{term1}
     &0\leq\sum_{w\in W_j\setminus\wt{W}_j}\ \bigg\{\frac{N!/\prod_{h=1}^j b_h!}{j^N}\bigg\}\prod_{i=1}^n  w_i \leq \bigg(\frac{N^{n}}{e^{N^{\delta^{\star}}}}\bigg)\times \frac{j^{-(n + j/2)}}{N^{j-1/2}}  = \mbox{UB}_{j, \rm neglect}.
\end{align}

$(ii)$ For the contributions corresponding to the elements of the set $\wt{W}_j$:
\begin{align}\label{term2}
    &h_{j,N}\times \frac{1}{\sqrt{2\pi}^{j-1}} \frac{1}{N^{j-1/2}}\mbox{LB}_{j,N}\leq\sum_{w\in \wt{W}_j}\ \bigg\{\frac{N!/\prod_{h=1}^j b_h!}{j^N}\bigg\}\prod_{i=1}^n  w_i \leq h_{j,N}\times \frac{1}{\sqrt{2\pi}^{j-1}} \frac{1}{N^{j-1/2}} \mbox{UB}_{j},
\end{align}
where
\begin{align*}
   \mbox{LB}_{j,N} = 
   & \bigg[j^{-N}\bigg\{\bigg(\frac{1}{j}-\frac{1}{N^{1-\delta}}\bigg)\bigg(\frac{1}{j}+\frac{1}{N^{1-\delta}}\bigg)\bigg\}^{(-N/2) + (j/2N^{1-\delta})+(j/4)}\bigg]\notag\\
   &\bigg[\bigg\{\bigg(\frac{1}{j}-\frac{1}{N^{1-\delta}}\bigg)^{\min(j/2,n)}\bigg(\frac{1}{j}+\frac{1}{N^{1-\delta}}\bigg)^{n-\min(j/2,n)}\bigg\}\bigg],\notag\\
   \mbox{UB}_{j} =& [j^{- j/2}]\ [j^{-n}],
\end{align*}
and  $h_{j,N}$ is the number of elements in $\wt{W}_j$ that trivially satisfy $(2N^{\delta} +1)^{j-1}\leq h_{j,N} \leq (2N^{\delta} +1)^j$. The step above assumes that $j$ is even, but we can easily do similar calculation for the case where $j$ is odd to obtain expressions that are similar in spirit. The above inequality holds since we simply replace each term in the sum by the minimum and maximum among all the terms in the sum respectively, to obtain the lower and upper bound.

So, we combine $(i)-(ii)$ to obtain the ratio of marginal likelihoods as they appear in the equation \eqref{eqn:bounded_ratio_marginallikelihhod}:
\begin{align}
   \mbox{LB}^{\star}_{j, n, N}
   &=\frac{h_{j,N}\ \sqrt{2\pi}^{-(j-n)}\ \mbox{LB}_{j,N}}{h_{n,N}\ (NM_{N})^{j-n}\mbox{UB}_{n} + M_{N}^{j-n} \sqrt{2\pi}^{n-1} N^{j-1/2} \mbox{UB}_{n, \rm neglect}}\notag\\
   &\leq\frac{m(x_{1:N}\mid k_N = j)}{m(x_{1:N}\mid k_N = n)}\notag\\
   &\leq \frac{h_{j,N}\ \sqrt{2\pi}^{-(j-n)}\ \mbox{UB}_{j}\ + (\sqrt{2\pi})^{n-1} N^{j-1/2} \mbox{UB}_{j, \rm neglect} }{h_{n,N}\ (NM_{N})^{j-n}\ \mbox{LB}_{N, n}}\
   = \mbox{UB}^{\star}_{j, n, N}.
\end{align}
It is now enough to focus on $\mbox{UB}^{\star}_{j, n, N}$, and obtain two-sided bounds as follows:
\begin{align}\label{eqn:ration_combined}
  &\bigg[\frac{(2N^{\delta} +1)^{j-n-1}}{{N^{(d\alpha +1)}}^{(j-n)}}\sqrt{2\pi}^{-(j-n)}\bigg]\frac{\mbox{UB}_{j}}{\mbox{LB}_{N, n}} + \bigg(\frac{\sqrt{2\pi})^{n-1} N^{n-1/2}}{M_{N}^{j-n}\ h_{n,N}}\bigg)\frac{\mbox{UB}_{j, \rm neglect}}{\mbox{LB}_{N, n}}\notag\\
  &\leq \mbox{UB}^{\star}_{j, n, N}\notag\\
  &\leq\bigg[\frac{(2N^{\delta} +1)^{j-n}}{{N^{(d\alpha +1)}}^{(j-n)}}\sqrt{2\pi}^{-(j-n)}\bigg]\frac{\mbox{UB}_{j}}{\mbox{LB}_{N, n}} +   \bigg(\frac{\sqrt{2\pi})^{n-1} N^{n-1/2}}{M_{N}^{j-n}\ h_{n,N}}\bigg)\frac{\mbox{UB}_{j, \rm neglect}}{\mbox{LB}_{N, n}}
\end{align}
where the terms within the third braces are decreasing functions of $j(\geq n)$. For the sake of simplicity, we obtain a simpler upper bound to the upper-bound of $\mbox{UB}^{\star}_{j, n, N}$ in equation \eqref{eqn:ration_combined}. To that end, first we note that 
\begin{align}\label{eqn:ration_big}
  \frac{\mbox{UB}_{j}}{\mbox{LB}_{N, n}} 
  =T^{(1, \star)}_{N, n}\times T^{(2, \star)}_{N, n} \leq T^{(1)}_{N, n}\times T^{(2)}_{N, n} = T_{N, n}.
\end{align}
where
\begin{align}
  T^{(1, \star)}_{N, n}
  =&\bigg\{\frac{j^{-n}}{(\frac{1}{n}-\frac{1}{N^{1-\delta}})^{n/2}(\frac{1}{n}+\frac{1}{N^{1-\delta}})^{n/2}}\bigg\} 
  \leq  \bigg\{\frac{j^{-n}}{(\frac{1}{n}-\frac{1}{N^{1-\delta}})^{n}}\bigg\}\notag\\
  \leq & \bigg\{\frac{j^{-n}}{n^{-n}}\bigg(1-\frac{n}{N^{1-\delta}}\bigg)^{-n}\bigg\},\quad \text{since}\ j\geq n, \notag\\ 
  \leq&  \bigg\{\bigg(1-\frac{n}{N^{1-\delta}}\bigg)^{-n}\bigg\}  
  = T^{(1)}_{N, n},  
\end{align}
and 
\begin{align}
  T^{(2,\star)}_{N, n}
  =&  \bigg\{\frac{j^{-j/2}}{n^{-N}\{(\frac{1}{n}-\frac{1}{N^{1-\delta}})(\frac{1}{n}+\frac{1}{N^{1-\delta}})\}^{(-N/2) + (n/2N^{1-\delta})+(n/4)}}\bigg\}\notag\\
  = &  \bigg\{\frac{j^{-j/2}}{n^{-N}\{(\frac{1}{n^2}-\frac{1}{N^{2-2\delta}})\}^{(-N/2) + (n/2N^{1-\delta})+(n/4)}}\bigg\}\notag\\
  = &  \bigg\{\frac{j^{-j/2}}{\{(1-\frac{n^2}{N^{2-2\delta}})^{-N/2}\}\ \{(\frac{1}{n^2}-\frac{1}{N^{2-2\delta}})^{n/2N^{1-\delta}}\}\ \{n^{-n/2}(1-\frac{n^2}{N^{2-2\delta}})^{n/4} \}}\bigg\}\notag\\
  \leq &  \bigg\{\bigg(1-\frac{n^2}{N^{2-2\delta}}\bigg)^{N/2}\ \bigg(\frac{1}{n^2}-\frac{1}{N^{2-2\delta}}\bigg)^{-n/2N^{1-\delta}}\ \bigg(1-\frac{n^2}{N^{2-2\delta}}\bigg)^{-n/4}\bigg\},\quad \text{since}\ j\geq n, \notag\\
  \leq &  \bigg\{ \bigg(\frac{1}{n^2}-\frac{1}{N^{2-2\delta}}\bigg)^{-n/2N^{1-\delta}}\ \bigg(1-\frac{n^2}{N^{2-2\delta}}\bigg)^{-n/4}\bigg\}\quad\text{since}\ (1-\frac{n^2}{N^{2-2\delta}})\leq 1, \notag\\
  =& T^{(2)}_{N, n}.    
\end{align}
The final expression of $T_{N, n}$ in equation \eqref{eqn:ration_big} is devoid of $j(>n)$ and both $T^{(1)}_{N, n}, T^{(2)}_{N, n}$ are decreasing functions of $N$. So,  we are only left to tackle
\begin{align}\label{eqn:ration_small}
    \bigg(\frac{\sqrt{2\pi})^{n-1} N^{n-1/2}}{M_{N}^{j-n}\ h_{n,N}}\bigg)\frac{\mbox{UB}_{j, \rm neglect}}{\mbox{LB}_{N, n}}
    &\leq \bigg(\frac{\sqrt{2\pi})^{n-1} N^{n-1/2}}{M_{N}^{j-n}\ h_{n,N}}\bigg)\frac{\mbox{UB}_{j, \rm neglect}\times T_{N, n}}{\mbox{UB}_{j}}\notag \\
    &= \frac{(\sqrt{2\pi})^{n-1}}{\ h_{n,N}}\frac{T_{N,n}}{(NM_{N})^{j-n}}\bigg(\frac{N^n}{e^{N^{\delta^{\star}}}}\bigg)\notag\\
    &\leq \frac{(\sqrt{\pi/2})^{n-1}}{\ N^{\delta(n-1)}}\frac{T_{N,n}}{(NM_{N})^{j-n}}\bigg(\frac{N^n}{e^{N^{\delta^{\star}}}}\bigg)\notag\\
    &\leq (\sqrt{\pi/2})^{n-1}\ T_{N,n}\ \bigg(\frac{N^{(1-\delta)n +\delta}}{e^{N^{\delta^{\star}}}}\bigg).
\end{align}
The first step uses the equation \eqref{eqn:ration_big}. The second step is simple algebra. The third step uses $h_{n,N}\geq (2N^{\delta} +1)^{n-1}$. The last step uses $j\geq n$ ,  is devoid of $j$, and is a decreasing function in $N$.  Finally, the equations \eqref{eqn:ration_combined}-\eqref{eqn:ration_small} together complete the proof of the fact that $\sum_{N=N^{\star}}^{\infty} P(k_N > n  \mid x_{1:n}) < \infty$.

\subsection*{Part 2}
Next, we move to the proof of the fact that $\sum_{N=N^{\star}}^{\infty} P(k_N < n  \mid x_{1:n}) < \infty$. Note that, $\min_{h,h^{\prime}} |\mu_h-\mu_{h^{\prime}}|_{\infty}\geq 2\rho_N$, and $\rho_N\to 0$ as $N\to\infty$. If $k_N<n$, then the likelihood  at at least one sample point can be made 0 by choosing $N$ large enough.  For $j<n$, from equation \eqref{eqn:full_post} we have 
\begin{align}
    &m(x_{1:N}\mid k_N = j)\notag\\ 
    =&\sum_{w_{1:j}}\sum_{\mu_{1:j}}\prod_{h=1}^j H(\mu_h)\ \bigg\{\frac{N!/\prod_{h=1}^j b_h!}{j^N}\bigg\}\prod_{i=1}^n \bigg\{\sum_{h=1}^j w_h \tau^{d}_{N}\ \delta(|x_i - \eta_h|_{\infty} <2\rho_N)\bigg\}\notag\\
    = &\ 0,
\end{align}
where $\delta(\cdot)$ is the Dirac delta measure, and $|\cdot|_{\infty}$ denotes the $\mb{L}_{\infty}$ norm. The final step holds true since $n-j$ indicators are exactly 0 for  $N\geq N^{\star}$ where it is sufficient to ensure
that
\begin{align}
    \min_{s\neq t} |x_s -x_t|_{\infty} > 2\rho_{N^{\star}} = \frac{1}{(N^{\star})^{\alpha-\beta}}\iff N^{\star} > \big\{\min_{s\neq t} |x_s -x_t|_{\infty}\big\}^{-1/(\alpha-\beta)}.
\end{align}
Then
\begin{align}\label{eqn:part2}
    P(k_N < n  \mid x_{1:n})
    = \sum_{j=1}^{n-1} P(k_N =j  \mid x_{1:n})
    = \frac{ \sum_{j=1}^{n-1}\ p(j)\ m(x_{1:n}
    \mid k_N = j)}{\sum_{j=1}^{\infty}\ p(j)\ m(x_{1:n}
    \mid k_N = j)} = 0,
\end{align}
for all $N\geq N^{\star}$ and $\sum_{N=N^{\star}}^{\infty} P(k_N < n  \mid x_{1:n}) = 0$. This completes the proof of the theorem.
\end{proof}

\subsection{Proof of Theorem \ref{th:npbayes_equivalence}}

Before we move to our detailed proof the theorem, we shall first look at a brief sketch of our arguments. Our general strategy is to start with the joint distribution of 
$(\xi^{\star},\ \theta,\ x_{1:n}) = (k, b_{1:k}, \mu_{1:k}, \theta, x_{1:n})$ 
and marginalise out  $\xi^{\star} = (k, b_1,\ldots, b_k,\ \mu_{1},\ldots,\mu_k)$ to compute the posterior of $\theta\mid x_{1:n}$. The proof mainly hinges on three key ideas:
\begin{itemize}
    \item In Theorem \ref{th2:lemma2}, we show that $P(k=n\mid x_{1:n}) \to 1$ \emph{almost surely} as $N\to\infty$. Consequently, we can focus on  
    $\pi_{\varepsilon, N}(\theta, b_1,\ldots, b_k, \mu_{1},\ldots,\mu_{k} \mid x_{1:n}, \ k = n)$,
    instead of $\pi_{\varepsilon, N}(\theta, b_1,\ldots, b_k, \mu_{1},\ldots,\mu_{k} \mid x_{1:n})$. 
    \item In the assumptions in Section \ref{ssec:npBayes}, we construct a $d$-dimensional uniform grid that expands to entire $\mb{R}^d$ and the length of the sides of the grid cells go to $0$ as $N\to\infty$. This construct makes sure that only one term in the sum over $(\mu_1,\ldots, \mu_k)$ survives for $N$ large enough.
    \item Our proof critically exploits Sterling's approximation of Multinomial probabilities of the form $1/k^N (N!/\prod_{h=1}^k b_h !)$. With the help of Lemma \ref{th2:lemma1}, for $N$ large enough, we can show that  when we consider the sum over $(b_1,\ldots,b_n)$ or equivalently $(w_1,\ldots, w_n)$, only the term(s) corresponding to the weights $(w_1,\ldots, w_n)$
    ``close" to the $\mbox{D}$-BETEL weights survive. We put this idea in rigorous form in equation $\eqref{eqn:continuity}$.
\end{itemize}

\begin{proof}
The proof proceeds in two steps. In step 1, we carry out algebraic simplification of the posterior without resorting to any asymptotic arguments, and in step 2, we invoke the special asymptotic regime to prove the final result.
\subsection*{Step 1 (simplification of the posterior)}
For notational homogeneity, we use $ \mbox{U}_d(x_i\mid\eta_h,\tau^{-1}I)$ to denote $\prod_{j=1}^d \mbox{Uniform}(\eta_{h, j} -  \tau^{-1},\ \eta_{h, j} +  \tau^{-1})$.
Under the hierarchical model in equation \eqref{eqn:npbayes3}--\eqref{eqn:npbayes2} in the main document, the joint posterior of $\theta,\xi^{\star}\mid x_{1:n}$ is  
\begin{align}\label{eqn:full_post}
\ \pi_{\varepsilon, N}(\theta,\xi^{\star}\mid x_{1:n})  
& \ \propto \ \pi(\theta) \ \pi_{\varepsilon, N}(\xi^{\star}\mid\theta)\
P^{(N)}(x_{1:n}\mid \xi^{\star},\theta) \notag\\ 
& = \pi(\theta) \ \Bigg\{\frac{1_{A_{\varepsilon, N}(\theta)}(\xi^{\star}) \ \pi_{\infty, N}(\xi^{\star})}{\int_{\xi^{\star}}1_{A_{\varepsilon, N}(\theta)}(\xi^{\star})\ \pi_{\infty,N}(\xi^{\star})d\xi^{\star}}\Bigg\}\ P^{(N)}(x_{1:n}\mid \xi^{\star},\theta) \notag\\
& = \pi(\theta)\ \Bigg\{\frac{1_{A_{\varepsilon, N}(\theta)}(\xi^{\star}) \ }{\int_{\xi^{\star}}1_{A_{\varepsilon, N}(\theta)}(\xi^{\star})\ \pi_{\infty, N}(\xi^{\star})d\xi^{\star}}\Bigg\}\ \bigg\{p(k)\ \prod_{h=1}^k \mbox{H}^{(N)}(\mu_h) \bigg\} \notag\\
&\ \propto \
 \pi(\theta)\ \Bigg\{\frac{1_{A_{\varepsilon, N}(\theta)}(\xi^{\star}) \ }{\int_{\xi^{\star}}1_{A_{\varepsilon, N}(\theta)}(\xi^{\star})\ \pi_{\infty,N}(\xi^{\star})d\xi^{\star}}\Bigg\}\ \bigg\{p(k)\prod_{h=1}^k \mbox{H}^{(N)}(\mu_h) \bigg\} \notag\\
&\quad\quad\quad\quad \bigg\{\frac{1}{k^N}\ \frac{N!}{\prod_{h=1}^k b_h !} \bigg\}\ 
\bigg\{\prod_{i=1}^n \bigg[\sum_{h=1}^k \frac{b_h}{N} \mbox{U}_d(x_i\mid\eta_{h},\tau^{-1}I) \bigg] \bigg\}. 
\end{align}


Next, use of the first part of Lemma \ref{th2:lemma1_b}, together with  a reparametrization of the posterior by writing $b_h/N=w_h,\ h=1,\ldots,n$ yield,
\begin{align} 
\pi_{\varepsilon, N}(\theta,\xi^{\star}\mid x_{1:n}) \ \propto \
& \pi(\theta)\  \frac{ 1_{A_{\varepsilon, N}(\theta)}(\xi^{\star})\ }{\int_{\xi^{\star}}1_{A_{\varepsilon, N}(\theta)}(\xi^{\star})\ \pi_{\infty,N}(\xi^{\star})d\xi^{\star}}\ \bigg\{p(k)\ \prod_{h=1}^k \mbox{H}^{(N)}(\mu_h) \bigg\}\notag\\
&\quad\quad \bigg\{\frac{U_{k, N, \mathbf{b}}}{V_{k, N}}\prod_{h=1}^k w_{h}^{-w_h N + \frac{1}{2}} \bigg\}\ 
\bigg\{\prod_{i=1}^n \bigg[\sum_{h=1}^k w_h \mbox{U}_d(x_i\mid\eta_{h},\tau^{-1}I) \bigg] \bigg\} 
\end{align}
where $V_{k, N} = k^N N^{k-\frac{1}{2}}\ (\sqrt{2\pi})^{k-1}$ and $e^{\frac{1}{12N+1}-\sum_{h=1}^k\frac{1}{12b_h }}<U_{k, N, \mathbf{b}}<e^{\frac{1}{12N}-\sum_{h=1}^k\frac{1}{12b_h +1}}$. 
Before proceeding to step 2, we introduce some notations that will aid the calculations going forward. Given, $X\mid\xi^{\star}  \sim P^{(N)} :=\sum_{h=1}^n w_h \mbox{U}_d(x\mid\eta_{h},\tau^{-1} I)$, and  
for a fixed $\theta\in\Theta$,  $F_{\theta}(\cdot) = \sum_{k=1}^{K_{0}} s_{0k}\ \mbox{ED}_{h}(\cdot\mid m_{0k}, \Sigma_{0k})$,  we can define $C_{\theta,\varepsilon}$ as in equation \eqref{eqn_C}.
Consequently, we can write
\begin{align*}
    \int_{\xi^{\star}}1_{A_{\varepsilon, N}(\theta)}(\xi^{\star})\ \pi_{\infty,N}(\xi^{\star})d\xi^{\star} &= \sum_{k} p(k)\sum_{w \in C_{\theta,\varepsilon}}\bigg[\sum_{\mu_1,\ldots,\mu_k} \prod_{h=1}^k\mbox{H}^{(N)}(\mu_h)\  \bigg\{\frac{U_{k, N, \mathbf{b}}}{V_{k, N}}\prod_{h=1}^k w_h^{-w_h N + \frac{1}{2}} \bigg\}\  \bigg].
\end{align*}

Next, in view of Theorem \ref{th2:lemma2}, we work with a partition parameter space 
\begin{align*}
 \{k=n\}\cup \{k\neq n\}, 
\end{align*}
and focus on $\pi_{\varepsilon, N}(\theta,\ b_1,\ldots, b_k, \mu_{1},\ldots,\mu_{k} \mid x_{1:n}, \ k = n) 
\propto$
\begin{align} 
& \frac{ \pi(\theta) 1\{ C_{\theta,\varepsilon} \}}{\int_{\xi^{\star}}1_{A_{\varepsilon, N}(\theta)}(\xi^{\star})\ \pi_{\infty, N}(\xi^{\star})d\xi^{\star}}\ \bigg\{ \prod_{h=1}^n \mbox{H}^{(N)}(\mu_h)\bigg\}\ \bigg\{\prod_{h=1}^n w_h^{-w_h N + \frac{1}{2}} \bigg\}\ \bigg\{\prod_{h=1}^n  w_h   \bigg\} .
\end{align}
Marginalization of the above with respect to $w_{1:n}$ and $\mu_{1:n}$ leads to 
\begin{align} 
\pi_{\varepsilon,N}(\theta, \mid x_{1:n},\ k = n) \ \propto \ 
& \pi(\theta)\ \frac{N_x^{\star}(\theta)}
{D_x^{\star}(\theta) }
\end{align}
where
\begin{align*} 
&N_x^{\star}(\theta) = \sum_{w \in C_{\theta,\varepsilon}}\bigg[\bigg\{\prod_{h=1}^n \mbox{H}^{(N)}(x_h)\bigg\} \ \bigg\{\prod_{h=1}^n w_{h}^{-w_h N + \frac{1}{2}} \bigg\}\ \bigg\{\prod_{h=1}^n  w_h  \bigg\} \bigg],\\
& D_x^{\star}(\theta) = \sum_{w \in C_{\theta,\varepsilon}}\bigg[\bigg\{\prod_{h=1}^n \mbox{H}^{(N)}(x_h)\bigg\} \  \bigg\{\prod_{h=1}^n w_{h}^{-w_h N + \frac{1}{2}} \bigg\}\bigg]. 
\end{align*}
Only one term in the sum over $\mu_{1},\ldots, \mu_k$ survives since $\mbox{H}^{(N)}(\mu_h)$ is uniform over the mid-points of the grid $\mathcal{X}^N$ with grid length $2\rho_N$ and the components have range $2\rho_{N}$. Next, by the uniformity assumption on $\mbox{H}^{(N)}$, it reduces to
\begin{align}\label{eqn:reduced_post} 
&\pi_{\varepsilon,N}(\theta\mid x_{1:n}, \ k = n) \ \propto \ \pi(\theta)\ \frac{
\sum_{w \in C_{\theta,\varepsilon}} \bigg\{\prod_{h=1}^n w_{h}^{-w_h N + \frac{1}{2}} \bigg\}\ \bigg\{\prod_{h=1}^n  w_h  \bigg\}}
{\sum_{w \in C_{\theta,\varepsilon}}\ \bigg\{\prod_{h=1}^n w_{h}^{-w_h N + \frac{1}{2}} \bigg\} }.
\end{align}

\subsection*{Step 2 (invoking the asymptotic regime)}
We recall that, in the $\mbox{D}$-BETEL formulation, for every $\theta\in\Theta$: 
\begin{align*}
  w^{\star}_{1:n}(\theta,\varepsilon) = \argmax_{w\in C_{\theta,\varepsilon}} H_{N}(w),
\end{align*}
where $H_{N}(w) = -\sum_{h=1}^n w_h\log w_h$, and $\varepsilon>0$. For notational simplicity we shall use $w^{\star}_h$ in place of  $w^{\star}_h(\theta,\varepsilon)$. Note that, $w^{\star}_{1:n}$ is either the global maximizer $(1/n,\ldots,1/n)^{\T}$ of $\prod_{h=1}^n w_{h}^{-w_h}$, or lies  at the boundary of feasible set $C_{\theta,\varepsilon}$ defined in equation \eqref{eqn_C}. In view of that, we shall analyse equation \eqref{eqn:reduced_post} in two cases. In case 1, we can carry out the rest of the analysis as in Theorem \ref{th2:lemma2}. But this case is not of much statistical worth. 
In the case 2, our goal is to  study the behaviour of the log ratio $R_n$ defined in equation \eqref{eqn_R_j}. 
In particular, we study the behaviour of $R_n$ in a carefully constructed neighbourhood  $\wt{C}_{\theta, \varepsilon}$ (refer to the equation \eqref{eqn_C_Tilde}) of $w^{\star}$ and outside it i.e in $C_{\theta, \varepsilon}\subset\wt{C}_{\theta, \varepsilon}$.
For any $w\in C_{\theta,\varepsilon}\setminus\wt{C}_{\theta,\varepsilon}$, we have  $R_n >c N^{\delta} - n\log N$ where $c>0$, from the Lemma \ref{th3:lemma1}. Also, note that the set $C_{\theta,\varepsilon}\setminus\wt{C}_{\theta,\varepsilon}$ can contain at max $N^n$ elements.  
 
 Now, we are in a position to return to the posterior in equation \eqref{eqn:reduced_post} for further simplification:
 \begin{align*}
  \pi_{\varepsilon,N}(\theta\mid x_{1:n},\ k = n) \ \propto \pi(\theta)\ \frac{\mbox{N}(\wt{C}_{\theta,\varepsilon}) + \mbox{N}(C_{\theta,\varepsilon}\setminus\wt{C}_{\theta,\varepsilon})}{\mbox{D}(\wt{C}_{\theta,\varepsilon}) + \mbox{D}(C_{\theta,\varepsilon}\setminus\wt{C}_{\theta,\varepsilon})}   
 \end{align*}
 where 
 \begin{align}
  &\mbox{N}(S)  = \sum_{w\in S }  \bigg\{\prod_{h=1}^n w_{h}^{-w_h N + \frac{1}{2}} \bigg\}\ \bigg\{\prod_{h=1}^n  w_h  \bigg\},\ \mbox{D}(S)  = \sum_{w\in S }  \bigg\{\prod_{h=1}^n w_{h}^{-w_h N + \frac{1}{2}} \bigg\},
 \end{align}
with $S=\wt{C}_{\theta,\varepsilon}$ or $C_{\theta,\varepsilon}\setminus\wt{C}_{\theta,\varepsilon}$.
In the remaining of the proof, we shall rigorously demonstrate that both the numerator and the denominator in the above expression are dominated by the terms corresponding to $w\in C_{\theta,\varepsilon}$, and the terms corresponding to $w\in C_{\theta,\varepsilon}\setminus\wt{C}_{\theta,\varepsilon}$ are negligible. For brevity of presentation, we introduce:
\begin{align}
  \mbox{Diff}(S) = \sum_{w\in S}  \bigg\{\prod_{h=1}^n w_{h}^{-w_h N + \frac{1}{2}} \bigg\} \bigg\{\prod_{h=1}^n  w_h - \prod_{h=1}^n  w^{\star}_h  \bigg\}
\end{align}
with $S=\wt{C}_{\theta,\varepsilon}$ or $C_{\theta,\varepsilon}\setminus\wt{C}_{\theta,\varepsilon}$. Next, we focus on the final piece of the proof:
\begin{align}\label{eqn:continuity}
&\Bigg|\frac{\mbox{N}(\wt{C}_{\theta,\varepsilon}) + \mbox{N}(C_{\theta,\varepsilon}\setminus\wt{C}_{\theta,\varepsilon})}{\mbox{D}(\wt{C}_{\theta,\varepsilon}) + \mbox{D}(C_{\theta,\varepsilon}\setminus\wt{C}_{\theta,\varepsilon})}  - \prod_{h=1}^n  w^{\star}_h\Bigg|\ 
=\Bigg|\frac{\mbox{Diff}(\wt{C}_{\theta,\varepsilon}) + \mbox{Diff}(C_{\theta,\varepsilon}\setminus\wt{C}_{\theta,\varepsilon})}
{\mbox{D}(\wt{C}_{\theta,\varepsilon}) + \mbox{D}(C_{\theta,\varepsilon}\setminus\wt{C}_{\theta,\varepsilon}) }\Bigg|\ \notag\\ 
\leq& \Bigg|\frac{\mbox{Diff}(\wt{C}_{\theta,\varepsilon})}{\mbox{D}(\wt{C}_{\theta,\varepsilon}) + \mbox{D}(C_{\theta,\varepsilon}\setminus\wt{C}_{\theta,\varepsilon}) }\Bigg|\ + \Bigg|\frac{\mbox{Diff}(C_{\theta,\varepsilon}\setminus\wt{C}_{\theta,\varepsilon})}{\mbox{D}(\wt{C}_{\theta,\varepsilon}) + \mbox{D}(C_{\theta,\varepsilon}\setminus\wt{C}_{\theta,\varepsilon}) }\Bigg|,\ \quad \text{(triangle inequality)},\notag\\ 
\leq& \Bigg|\frac{\mbox{Diff}(\wt{C}_{\theta,\varepsilon})}{\mbox{D}(\wt{C}_{\theta,\varepsilon}) }\Bigg|\ + \Bigg|\frac{\mbox{Diff}(C_{\theta,\varepsilon}\setminus\wt{C}_{\theta,\varepsilon})}{\mbox{D}(\wt{C}_{\theta,\varepsilon}) }\Bigg|
\leq \Bigg|\frac{\mbox{Diff}(\wt{C}_{\theta,\varepsilon})}{\mbox{D}(\wt{C}_{\theta,\varepsilon}) }\Bigg|\ + \Bigg|\frac{\mbox{Diff}(C_{\theta,\varepsilon}\setminus\wt{C}_{\theta,\varepsilon})}{\prod_{h=1}^n {w_{h}^{\star}}^{-w_{h}^{\star} N + \frac{1}{2}} }\Bigg|.
\end{align}
To argue that the first term $\to 0$ as $N\to\infty$,  it is enough to demonstrate that $\sup_{w\in C_{\theta, \epsilon}}\bigg|\prod_{h=1}^n  w_h - \prod_{h=1}^n  w^{\star}_h  \bigg|\to 0$ as $N\to\infty$ (Lemma \ref{th2_lemma2_asymtotic}). Next we argue that the second term $\to 0$ as $N\to\infty$. The facts that $\bigg|\prod_{h=1}^n  w_h - \prod_{h=1}^n  w^{\star}_h  \bigg|\leq 1/n^n$ trivially, and for any $w\in C_{\theta,\varepsilon}\setminus\wt{C}_{\theta,\varepsilon}$, $R_n >c N^{\delta} - n\log N, c>0$ yield
\begin{align*}
    \Bigg|\frac{\mbox{Diff}(C_{\theta,\varepsilon}\setminus\wt{C}_{\theta,\varepsilon})}{\prod_{h=1}^n {w_{h}^{\star}}^{-w_{h}^{\star} N + \frac{1}{2}} }\Bigg|\leq \frac{1}{n^n}\frac{N^{2n}}{e^{cN^\delta}}
\end{align*}
which is a decreasing sequence in $N$ that $\downarrow 0$ as $N\to \infty$.

Finally, on application of part two of Lemma \ref{th2:lemma1_b}, Theorem \ref{th2:lemma2}, observation in equation \eqref{eqn:continuity} we have 
\begin{align*} 
&\pi_{\varepsilon,N}(\theta\mid x_{1:n}) 
\ \propto\ w(\theta)\ \prod_{i=1}^n w_{i}^{\star}(C_{\theta,\varepsilon}).
\end{align*}
The right hand side is precisely the  $\mbox{D}$-BETEL posterior.
\end{proof}

\section{Proof of Theorem \ref{th3} in the main document}\label{ssec:th3}

Here we present the the proof of the Theorem \ref{th3}  in Section \ref{ssec:andrew}\ in the main document and a cascade of auxiliary results.

\begin{proof}
Exploiting the independence between the two blocks of both $X_0^{\star}$ and $X_1^{\star}$ , we have
\begin{align}\label{derive_WAR} 
W_{\rm AR}^2(p_0, p_1)
=& \inf_{\gamma\in  R^{\alpha}(p^{\star}_0,p^{\star}_1)} \mb{E}_{\gamma}\big|\big|X_0^{\star} - X_1^{\star}
\big|\big|^2\notag\\
=& \inf_{\gamma_1\in \pi(p_0,p_1)\cap \EMM^{\alpha}_{2d}(K_0 K_1)} \mb{E}_{\gamma_1}\big|\big|X_0 - X_1\big|\big|^2 +  \inf_{\gamma_2\in \pi(\widetilde{p}_0,\widetilde{p}_1)} \mb{E}_{\gamma_2}\big|\big|\widetilde{X}_0 - \widetilde{X}_1\big|\big|^2\notag\\
=&\inf_{\gamma_1\in \pi(p_0,p_1)\cap \EMM^{\alpha}_{2d}(K_0 K_1)} \mb{E}_{\gamma_1}\big|\big|X_0 - X_1\big|\big|^2 + \sum_{j=1}^d \inf_{\gamma_{2j}\in \pi(p_{0j},p_{1j})} \mb{E}_{\gamma_{2j}}\big(X_{0j} - X_{1j}\big)^2\notag\\
=&\ \inf_{\gamma_1\in \pi(p_0,p_1)\cap \EMM^{\alpha}_{2d}(K_0 K_1)} \mb{E}_{\gamma_1}\big|\big|X_0 - X_1\big|\big|^2 + \sum_{j=1}^d W_2^2(p_{0j},p_{1j})\notag\\
=& \ \inf_{\gamma_1\in \pi(p_0,p_1)\cap \EMM^{\alpha}_{2d}(K_0 K_1)} \mb{E}_{\gamma_1}\big|\big|X_0 - X_1\big|\big|^2 + \sum_{k=1}^d \int_{0}^1 (F_{0k}^{-1}(z) - F_{1k}^{-1}(z))^2 dz 
\end{align}
where $F_{jk}^{-1}$ is quantile function of the univariate random variable $X_{jk},\ j= 0,1;\ k = 1,2,\ldots,d$. Now if we assume
$X_0\sim p_{0}\equiv \sum_{k=1}^{K_0} s_{0k} \mbox{ED}(m_{0k},\Sigma_{0k}),\
X_1\sim p_{1}\equiv \sum_{k=1}^{K_1} s_{1k} \mbox{ED}(m_{1k},\Sigma_{1k})
$,
by Lemma \ref{th1:lemma4}, the expression above reduces to 
\begin{align} 
W_{\rm AR}^2(p_0, p_1) =& \inf_{pi\in\pi^{\alpha}(s_0,s_1)}\bigg[\sum_{k,l} \pi_{kl}\ W_{2}^{2}(\mbox{ED}(m_{0k},\Sigma_{0k}),\mbox{ED}(m_{1l},\Sigma_{1l}))\bigg]\notag\\
&+ \sum_{k=1}^d \int_{0}^1 \big(F_{0k}^{-1}(z) - F_{1k}^{-1}(z)\big)^2 dz
\end{align}
Further if we assume $ p_1 \equiv {\sum_{k=1}^{K_1} s_{1k} \delta_{m_{1k}}} \ $, by Lemma \ref{th1:lemma5}, the expression above becomes 
\begin{align} 
W_{\rm AR}^2(p_0,  p_1) =& \inf_{\pi\in\pi^{\alpha}(s_0,s_1)}\bigg[\sum_{k,l} \pi_{kl}\ ||m_{0k}-m_{1l}||^2\bigg]+ \nu_{h}\sum_{k=1}^{K_0} s_{0k} \mbox{tr}(\Sigma_{0k})\notag\\
&+ \sum_{k=1}^d \int_{0}^1 \big(F_{0k}^{-1}(z) - F_{1k}^{-1}(z)\big)^2 dz
\end{align}
which on application of Lemma \ref{th1:lemma6}, completes the proof. 
\end{proof}

\section{Auxiliary results for the proof of Theorems \ref{th2:lemma2} and \ref{th:npbayes_equivalence}}\label{aux:th1th2}
\begin{definition}[Total variation metric, \citet{LevinPeresWilmer2006}]
Consider a measurable space $(\Omega ,{\mathcal {F}})$and probability measures $p$ and $q$ defined on $(\Omega, \mathcal{F})$. The total variation distance between $p$ and $q$ is defined as
$|p-q|_{\rm TV}=\sup _{A\in {\mathcal {F}}}\left|P(A)-Q(A)\right|$.
\end{definition}
When $\Omega$  is countable, the total variation distance is related to the $\mb{L}_1$ norm by the identity 
\begin{align*}
   |p-q|_{\rm TV} = \frac{1}{2}||p-q||_1 = \frac{1}{2}\sum_{\omega\in \Omega}|p(\omega) - q(\omega)|.
\end{align*}

\begin{lemma}[Pinsker's inequality, \citet{alma991023405949705251}]\label{Pinsker}
For any two probability distributions $p,q$ on $(\Omega ,{\mathcal {F}})$,
\begin{align*}
|p-q|_{\rm TV} \leq \sqrt{\frac{1}{2}\ \rm{KL}(p\ ||\ q)}.    
\end{align*}
\end{lemma}

\begin{lemma}\label{th2:lemma1_b}
Under the hierarchical specification in \eqref{eqn:npbayes1}--\eqref{eqn:npbayes3}, the prior on 
$(b_1,\ldots, b_k)$
is given by $\rm{Multinomial}(N;1/k, \ldots, 1/k)$ with probability mass function $N!/(k^N \prod_{h=1}^k b_h!)$. Then, if $b_h > 0$ for all $h=1,\ldots,k$,
\begin{align*}
   \frac{ e^{\frac{1}{12N+1}-\sum_{h=1}^k\frac{1}{12b_h }}}{V_{k,N}}\bigg\{\prod_{h=1}^k  w_{h}^{-w_h N+\frac{1}{2}}\bigg\} \leq \frac{N!}{k^N \prod_{h=1}^k b_h!} \leq \frac{ e^{\frac{1}{12N}-\sum_{h=1}^k\frac{1}{12b_h + 1}}}{V_{k,N}}\bigg\{\prod_{h=1}^k  w_{h}^{-w_h N+\frac{1}{2}}\bigg\}.
\end{align*}
where $w_h = b_h/N,\ h=1,\ldots, k$ and $V_{k,N} = k^N N^{k-\frac{1}{2}}\ (\sqrt{2\pi})^{k-1}$. Also, as $b_h\to \infty,\ h= 1,\ldots, k$,
\begin{align*}
    e^{\frac{1}{12N+1}-\sum_{h=1}^k\frac{1}{12b_h }} \to 1, \quad \text{and}\quad e^{\frac{1}{12N}-\sum_{h=1}^k\frac{1}{12b_h + 1}}\to 1.
\end{align*}
\begin{proof}
Let us recall the Sterling's approximation for factorials, \citep{10.2307/2315957}:
for all $a\geq 1$,
\begin{equation}\label{th2:lemma1}
\sqrt{2\pi} \ a^{a+\frac{1}{2}}\ e^{-a + \frac{1}{12a +1}}  <   a! < \sqrt{2\pi} \ a^{a+\frac{1}{2}}\ e^{-a + \frac{1}{12a}}.
\end{equation}
On repeated use of equation \eqref{th2:lemma1},
\begin{align*}
\frac{\sqrt{2\pi} \ N^{N+\frac{1}{2}}\ e^{-N + \frac{1}{12N+1}}}{k^N \prod_{h=1}^k \sqrt{2\pi} \ b_{h}^{b_h+\frac{1}{2}}\ e^{-b_h + \frac{1}{12b_h}}} \leq \frac{N!}{k^N \prod_{h=1}^k b_h!} \leq \frac{\sqrt{2\pi} \ N^{N+\frac{1}{2}}\ e^{-N + \frac{1}{12N}}}{k^N  \ \prod_{h=1}^k \sqrt{2\pi} \ b_{h}^{b_h+\frac{1}{2}}\ e^{-b_h + \frac{1}{12b_h +1}}}.
\end{align*}
Since $\sum_{h=1}^k b_h = N$, the expression above simplifies to:

\begin{align*}
    \frac{N^{N+\frac{1}{2}}\ e^{ \frac{1}{12N+1}}}{k^N (2\pi)^{k-1} \prod_{h=1}^k   b_{h}^{b_h+\frac{1}{2}}\ e^{\frac{1}{12b_h}}} \leq \frac{N!}{k^N \prod_{h=1}^k b_h!} \leq \frac{N^{N+\frac{1}{2}}\ e^{\frac{1}{12N}}}{k^N (\sqrt{2\pi})^{k-1} \ \prod_{h=1}^k  b_{h}^{b_h+\frac{1}{2}}\ e^{\frac{1}{12b_h +1}}}.
\end{align*}
Under the re-parametrization $w_h = b_h/N, \ h = 1,\ldots, k$:
\begin{align*}
    &\frac{N^{\frac{1}{2}-k}\ e^{ \frac{1}{12N+1}}}{k^N (2\pi)^{k-1} \prod_{h=1}^k   w_{h}^{w_h N+\frac{1}{2}}\ e^{\frac{1}{12b_h}}} \leq \frac{N!}{k^N \prod_{h=1}^k b_h!} \leq \frac{N^{\frac{1}{2}-k}\ e^{\frac{1}{12N}}}{k^N (\sqrt{2\pi})^{k-1} \ \prod_{h=1}^k  w_{h}^{w_h N+\frac{1}{2}}\ e^{\frac{1}{12b_h +1}}}\\
\iff&\frac{ e^{\frac{1}{12N+1}-\sum_{h=1}^k\frac{1}{12b_h }}}{k^N N^{k-\frac{1}{2}}\ (\sqrt{2\pi})^{k-1}}\bigg\{\prod_{h=1}^k  w_{h}^{-w_h N+\frac{1}{2}}\bigg\} \leq \frac{N!}{k^N \prod_{h=1}^k b_h!} \leq \frac{ e^{\frac{1}{12N}-\sum_{h=1}^k\frac{1}{12b_h + 1}}}{k^N N^{k-\frac{1}{2}}\ (\sqrt{2\pi})^{k-1}}\bigg\{\prod_{h=1}^k  w_{h}^{-w_h N+\frac{1}{2}}\bigg\}.
\end{align*}
Note that, the calculations above do not utilise our  assumptions in Section \ref{ssec:npBayes}.
\end{proof}
\end{lemma}

\begin{lemma}\label{th3:lemma1}
For every $w\in C_{\theta,\varepsilon}\setminus\wt{C}_{\theta, \varepsilon}$, $R_n > cN^{\delta} - n\log N$ where $c>0$. 
\begin{proof}
Fix $w\in C_{\theta,\varepsilon}\setminus\wt{C}_{\theta, \varepsilon}$, by definition of $\wt{C}_{\theta, \varepsilon}$ (refer to equation \eqref{eqn_C_Tilde}) either of the two following displayed equations hold true:
\begin{align*}
 &H_{N}(w^{\star}) - H_{N}(w)> 1/N^{1-\delta},\\  
 &  W^{2}_{\rm AR}[F_{\theta} , \nu(w^{\star}, x) ] - 1/N^{1-\delta}> W^{2}_{\rm AR}[F_{\theta} , \nu(w, x) ].
\end{align*}
If the first equation holds true, then 
\begin{align*}
    R_n =  N\big[ H_{N}(w^{\star}) - H_{N}(w)\big] + \frac{1}{2} \sum_{h=1}^n \log\bigg(\frac{w^{\star}_{h}}{w_h}\bigg) >  N^{\delta} - n\log N,
\end{align*}
since $1/N\leq w_{h}^{\star}, w_{h}\leq (N-1)/N$ trivially yields $(1/2) \sum_{h=1}^n \log(w^{\star}_{h}/w_h)\geq - (n/2) \log (N-1)$.
Next, suppose the second equation holds true. By definition of $w^{\star}$ in equation \eqref{eqn_DBETEL_dual}, we have 
{\footnotesize
\begin{align*}
   &g(w^{\star}) - g(w) > 0 , \notag\\
   =&  H_{N}(w^{\star}) - H_{N}(w) >  \lambda_{\star} \bigg(W_{\rm AR}^2(F_{\theta},  \nu(w^{\star}, x))- W_{\rm AR}^2(F_{\theta} ,  \nu(w, x))\bigg)  >\frac{\lambda_{\star}}{N^{1-\delta}}.
\end{align*}
} 
That yields $R_n > \lambda_{\star} N^{\delta} - n\log N$. Hence, we have the proof.
\end{proof}
\end{lemma}

\begin{lemma}\label{W2_lowerbound}
Fix $\eta>0$. If $W_{\rm AR}^2[\nu(w^{(1)}, x), \nu(w^{(2)}, x)]\leq\eta$, there exists a $\eta^{\prime}$ depended on $\eta$ such that $|w^{(1)}_i - w^{(2)}_i|\leq\eta^{\prime},\ i = 1,\ldots,n$.
\begin{proof}
By definition \ref{derive_WAR}, we have 
\begin{align}
 W_{\rm AR}^2[\nu(w^{(1)}, x), \nu(w^{(2)}, x)]\geq \sum_{j=1}^d \inf_{\gamma_{2j}\in \pi(p_{0j},p_{1j})} \mb{E}_{\gamma_{2j}}\big(X_{0j} - X_{1j}\big)^2.  
\end{align}
Without loss of generality, we assume that $\argmin_{j=1,\ldots,d } \inf_{\gamma_{2j}\in \pi(p_{0j}, p_{1j})} \mb{E}_{\gamma_{2j}}\big(X_{0j} - X_{1j}\big)^2 =1$. That yields $W_{\rm AR}^2[\nu(w^{(1)}, x), \nu(w^{(2)}, x)]\geq d \ [\inf_{\gamma_{21}\in \pi(p_{01}, p_{11})} \mb{E}_{\gamma_{21}}\big(X_{01} - X_{11}\big)^2] $. Next, we focus on
\begin{align}
 &\inf_{\gamma_{21}\in \pi(p_{01}, p_{11})}\mb{E}_{\gamma_{21}}\big(X_{01} - X_{11}\big)^2
 = \inf_{\gamma\in \pi(w^{(1)}, w^{(2)})}\sum_{i,j=1}^n \gamma_{(ij)} (x_{i1} -x_{j1})^2\notag\\
 = &\inf_{\gamma\in \pi(w^{(1)}, w^{(2)})}\bigg\{\sum_{i\neq j} \gamma_{(ij)} (x_{i1} -x_{j1})^2\bigg\}
 \geq\ [\min_{i\neq j}(x_{i1} -x_{j1})^2]\inf_{\gamma\in \pi(w^{(1)}, w^{(2)})}\bigg\{\sum_{i\neq j} \gamma_{(ij)}\bigg\}\notag\\
 =&\ [\min_{i\neq j}(x_{i1} -x_{j1})^2] \inf_{\gamma\in \pi(w^{(1)}, w^{(2)})}\bigg\{ 1-\sum_{i=1}^n \gamma_{(ii)}\bigg\}
 =[\min_{i\neq j}(x_{i1} -x_{j1})^2] \bigg\{ 1-\sum_{i=1}^n \min(w^{(1)}_i, w^{(2)}_i) \bigg\}\notag\\
 \geq&\ \frac{1}{2} [\min_{i\neq j}(x_{i1} -x_{j1})^2] \sum_{i=1}^n |w^{(1)}_i - w^{(2)}_i|, 
\end{align}
where the last step holds since, for any two $a, b\in\mb{R}, |a-b|/2 = (a+b)/2 - \min(a,b) $, and the rest is trivial.  So, we have $(d/2) [\min_{i\neq j}(x_{i1} -x_{j1})^2] \sum_{i=1}^n |w^{(1)}_i - w^{(2)}_i|\ ]\leq \eta $ which yields $ \sum_{i=1}^n |w^{(1)}_i - w^{(2)}_i|\ \leq 2\eta/ [d\times\min_{i\neq j}(x_{i1} -x_{j1})^2]=\eta^{\prime}$. Hence we have the proof.
\end{proof}
\end{lemma}
The arguments in Lemma \ref{W2_lowerbound} can easily be extended for Wasserstein metric with a general cost function.

\begin{lemma}\label{th2_lemma2_asymtotic}
$\sup_{w\in \wt{C}_{\theta, \epsilon}}\big|\prod_{h=1}^n  w_h - \prod_{h=1}^n  w^{\star}_h  \big|\to 0$ as $N\to\infty$. 
\begin{proof}
Fix $\theta\in\Theta$ and $\varepsilon>0$. Since for every $w\in \wt{C}_{\theta, \epsilon}$ (refer to equation \eqref{eqn_C_Tilde}), we have $0<W_{\rm AR}^2[F_{\theta}, \nu(w, x)]\leq W_{\rm AR}^2[F_{\theta}, \nu(w^{\star}, x)]\leq \varepsilon$, use of  triangle inequality yields $W_{\rm AR}^2[\nu(w^*, x), \nu(w, x)]\leq 2\varepsilon$. Consequently, by application of Lemma \ref{W2_lowerbound}, we can construct a small rectangle $R$ around $w^\star$ such that $R=\{w: |w_h - w_{h}^{\star}|\le \varepsilon^{\prime},\ h=1,\ldots, n\}$ for an appropriately chosen $\varepsilon^{\prime}>0$.  Further, we assume that no point inside $R$ other than $w^\star$, satisfies $H_{N}(w) = H_{N}(w^\star)$.
Next, we  define a sequence of sets $\mathcal{C}_N = \big\{ w \in R : 0 \le H_{N}(w^*) - H_{N}(w) \le 1/N^{1-\delta}\}$. Clearly, we have $C_{\theta, \epsilon}\subset \mathcal{C}_N$. Denoting $a_{(N)} = \sup_{w \in \mathcal{C}_N} | \prod w_h - \prod w_h^*|$, it is now enough to show that $a_{(N)} \to 0$ as $ N \to \infty$. 

First, we argue that $\mathcal{C}_N\subset R$ is a compact set. To that end, we note that $\mathcal{C}_N\subset R$ is bounded by construction, and it is closed being the intersection of two closed sets -- $R$ and the inverse image of a closed set under the continuous function $H$. So, $\mathcal{C}_N$ is a compact set. Since $\mathcal{C}_N$ is compact and the  function $w\to\prod_{h=1}^n w_h$ is continuous, $a_{(N)}$ is attained at some point $w^{(N)} \in \mathcal{C}_N$. 

Next, we argue that $\cap_{N=1}^{\infty} \mathcal{C}_N = \{w^{\star}\} $. Here $\mathcal{C}_1\supset \mathcal{C}_2\supset\ldots$ is a nested sequence of non-empty compact sets, and  since $w^\star\in \mathcal{C}_N$ for all $N$ we have $w^{\star}\in \cap_{N=1}^{\infty} \mathcal{C}_N$ . Now, if there is another point $\tilde{w}\in \cap_{N=1}^{\infty} \mathcal{C}_N$, we have $H(\tilde{w}) = H(w^\star)$, which leads to a contradiction. So, $\mathcal{C}_N$ is a nested (decreasing), non-empty sequence of compact sets in $\mathbb{R}^n$, whose intersection is a singleton. 

Next, we argue that $\mbox{diam}(\mathcal{C}_N):= \sup\{ \|w - w^{\prime}\| : w, w^{\prime} \in \mathcal{C}_N)\} \to 0$ as $N\to\infty$. To that end, suppose $\mbox{diam}(\mathcal{C}_N)> r >0$ for all $N$. Then there exists $b^{(N)}, c^{(N)}\in \mathcal{C}_N$ such that $|b^{(N)}- c^{(N)}|> r$ and $b^{(N)}\to b,\ c^{(N)}\to c$ as $N\to\infty$. Consequently, $|b-c|>r$ , and since $\cap_{N=1}^{\infty} \mathcal{C}_N$ is compact $b, c \in \cap_{N=1}^{\infty} \mathcal{C}_N$, this leads to a contradiction. So, $\mbox{diam}(\mathcal{C}_N) \to 0$ as $N\to\infty$.

Finally, since $w^{\star}, w^{(N)}\in \mathcal{C}_N$, we have $|w^{\star}- w^{(N)}|\leq \mbox{diam}(\mathcal{C}_N)$ trivially. Application of Sandwich theorem yields $w^{(N)}\to w^{\star}$ as $N\to\infty$. Hence we have $a_N\to 0$ as $N\to\infty$. 
\end{proof}
\end{lemma}



\section{Auxiliary results for the proof of Theorem \ref{th3}}\label{aux:th3}
Given two EMMs $p_0$ and $p_1$ on $\mb{R}^d$, the optimal transport plans are usually not an EMM. To avoid this, we can choose to restrict the set of admissible transport plans to the family of  EMMs, and tentatively define a modified 2-Wasserstein metric by
\begin{align}\label{eqn:mwalpha}
\mbox{\rm MW}_{2,\alpha}^{2}(p_0, p_1)
= \inf_{\nu\in \pi(p_0,p_1) \cap \EMM^{\alpha}_{2d}(\infty) }\ \int_{\mb{R}^d\times\mb{R}^d}\
\left\lVert y_0 - y_1 \right\rVert^2\ d\nu(y_0,y_1).
\end{align}

It is worth pointing out that, we briefly discussed about a slight variation of the definition above in Section \ref{ssec:andrew} of the main document given  by 
\begin{align} 
\mbox{\rm MW}_{2}^{2}(p_0, p_1)
= \inf_{\nu\in \pi(p_0,p_1) \cap \EMM_{2d}(\infty) }\ \int_{\mb{R}^d\times\mb{R}^d}\
\left\lVert y_0 - y_1 \right\rVert^2\ d\nu(y_0,y_1).
\end{align}

In this section, it's beneficial to work with \eqref{eqn:mwalpha} as an important building block to introduce our novel Wasserstein metric $W_{\rm AR}^{2}$, which critically involves an additional augmentation scheme  described in Section \ref{ssec:andrew} of the main document. To that end, we first list out some definitions following \citet{10.2307/4616956}.

\begin{definition}\label{def:indentifiable_emm}
Finite mixtures from the location scatter family $\{f_{\zeta, d} : \zeta = (\theta,\mu,\Sigma)\in \mb{A}^d\}$ are called identifiable if a relation of the form
\begin{align*} 
\sum_{j=1}^m \lambda_j f_{\zeta_j, d}(x) = \sum_{j=1}^m \lambda^{\prime}_j f_{\zeta^{\prime}_j, d}(x),\ x\in\mb{R}^d
\end{align*}
where $m$ is a positive integer, $\sum_{j=1}^m \lambda_j = \sum_{j=1}^m \lambda^{\prime}_j = 1$, $\lambda_j, \lambda^{\prime}_j > 0$ for $j = 1, \ldots,m$ implies there exists a permutation $\sigma$ such that $(\lambda_j, \zeta_j) = (\lambda^{\prime}_{\sigma(j)}, \zeta^{\prime}_{\sigma(j)})$ for all $j$.
\end{definition}

\begin{definition}
A function $f_d(\cdot,\theta)$ is called a density generator if it is a non-negative function on $[0,\infty) $ or $(0,\infty) $ such that the spherically symmetric function $f_d(x^\T x,\theta), x\in\mb{R}^d$ integrates to 1.
\end{definition}

\begin{definition}
A function $\phi(u),\ u\geq 0$ is called a characteristic generator in dimension $d\geq 1$ ,if $\phi(t^\T t)$  is the characteristic function of a probability distribution on $\mb{R}^d$. 
\end{definition}
Next, we record a series of sufficient conditions for identifiability of EMMs in terms of density generators and characteristic generators provided in  \citet{10.2307/4616956}. 

\begin{lemma}\label{th1:lemma1}[A sufficient condition for identifiability of EMMs via characteristic generators ]{
Suppose that a parametric family of characteristic generators gives rise to families of elliptical densities $\{f_{\zeta, d} : \zeta = (\theta,\mu,\Sigma)\in \mb{A}^d\}$ in dimension $1\leq d < q$. Suppose there exists a total ordering $\preceq$ on the set $\mb{B}^*$ such that $\beta_1\prec\beta_2$ implies
\begin{align*}
\lim_{u\to \infty}\frac{\phi_{\beta_1}(u)}{\phi_{\beta_2}(u)} = 0.
\end{align*}
}
Then finite mixtures from the class $\{f_{\zeta, d} : \zeta = (\theta,\mu,\Sigma)\in \mb{A}^d\}$ of elliptical distributions in $\mb{R}^d$ are identifiable for each $1\leq d < q$.
\end{lemma}
Family of multivariate t-distribution, symmetric stable law, band-limited densities satisfy the sufficient condition.

\begin{lemma}\label{th1:lemma2}[A sufficient condition for identifiability of EMMs via density generators]{

Let $f_d(\cdot,\theta),\theta\in\Theta$  be a parametric family of density generators for spherically symmetric distributions in $\mb{R}^d$. Let $\mb{C}=\Theta\times(0,\infty)\times\mb{R}$ and let $\gamma_j = (\theta_j,a_j,b_j)\in\mb{C},\ j= 1,2$. Suppose there exists a total ordering $\preceq$ on the set $\mb{C}$ such that $\gamma_1\prec\gamma_2$ implies
\begin{align*} 
\lim_{u\to \infty}\frac{f_d(a_2 u^2 + b_2 u + c_2) ,\theta)}{f_d(a_1 u^2 + b_1 u + c_1) ,\theta)} = 0,\ c_1,c_2\in\mb{R}.
\end{align*}
}
Then finite mixtures from the class $\{f_{\zeta, d} : \zeta = (\theta,\mu,\Sigma)\in \mb{A}^d\}$ of elliptical distributions in $\mb{R}^d$ are identifiable for each $1\leq d < q$.
\end{lemma}
Exponential power distribution, the original Kotz distribution and the multivariate normal law satisfy the above condition. \citet{10.2307/4616956} also provides sufficient conditions for identifiability of location-scatter mixtures for density generators that lacks smoothness at the origin, and Normal scale mixtures.  With these, we now have all the necessary machinery to present the results that leads to the proof of Theorem \ref{th3}.
\begin{lemma}\label{th1:lemma3}{
The optimal transport between two elliptical distributions $ \mbox{\rm ED}_h(a, A)$ and $ \mbox{\rm ED}_{h}(b, B)$ in $\mb{R}^d$ is in the same elliptical family in $\mb{R}^{2d}$. 
}
\end{lemma}
\begin{proof}
Since $X_0\sim p_0\equiv \mbox{\rm ED}_h(a, A)$ and $X_1\sim p_1\equiv \mbox{\rm ED}_h(b, B)$, we know $W_{2}^2(p_0, p_1) = ||a-b||^2 + \nu_h\ \mb{B}(A,B)$ \citep{muzellec2019generalizing}, where $\mb{B}(A,B) = \mbox{tr}[A + B - 2(A^{1/2}BA^{1/2})^{1/2}]$ is the Bures metric \citep{https://doi.org/10.48550/arxiv.1712.01504} between matrices. So, it is enough to show that there exists a coupling belonging to the same  elliptical family in $\mb{R}^{2d}$ with marginals $p_0$ and $p_1$ that incurs the optimal cost $ ||a-b||^2 + \nu_h\ \mb{B}(A,B)$. To that end, consider the construction
\begin{align*} 
& X_0\sim \mbox{ED}_h(a, A)\\
& X_1 = b + A^{-1/2}(A^{1/2}BA^{1/2})^{1/2}A^{1/2}(X_0-a),
\end{align*}
which ensures that $X_1 \sim  \mbox{ED}_{h}(b, B)$. Further, it ensures that $(X_1, X_2)$ is  a coupling belonging to the same  elliptical family in $\mb{R}^{2d}$ which is easy to see from the form characteristic function in equation \eqref{eqn:ch_EMM}.
Denoting $H=A^{-1/2}(A^{1/2}BA^{1/2})^{1/2}A^{1/2}$ and  $D=X_0-X_1=(I - H)X_0 - b + Ha$, we have
\begin{align*} 
\mu_D &= \mb{E}(D) = a-b,\\
\Sigma_D&=\mbox{Var}(D) = \nu_h(I-H)A(I-H)^T\\
&=\nu_h[A+B-A^{-1/2}(A^{1/2}BA^{1/2})^{1/2}A^{1/2}-A^{1/2}(A^{1/2}BA^{1/2})^{1/2}A^{-1/2}].
\end{align*}
Consequently, we have $ 
\mb{E}(||D||^2) = \mb{E}(D^\T D)
= \mu_{D}^\T\mu_{D} +\mbox{tr}(\Sigma_D)
=||a-b||^2 + \nu_h\mb{B}(A,B),
$
which completes the proof.
\end{proof}


\begin{lemma}\label{th1:lemma4}[Discrete formulation]{ Suppose we have two EMMs, that satisfy conditions in Lemma \ref{th1:lemma1} and/or Lemma \ref{th1:lemma2}, $p_0 = \sum_{k=1}^{K_0} s_{0k}  \mu_{0k}$ with $\mu_{0k}\sim\mbox{\rm ED}_h(m_{0k}, \Sigma_{0k})$, and $p_1 = \sum_{k=1}^{K_1} s_{1k} \mu_{1k}$ with $\mu_{1k}\sim\mbox{\rm ED}_h(m_{1k}, \Sigma_{1k})$. Then, we have
\begin{align*} 
\mbox{\rm MW}_{2,\alpha}^{2}(p_0, p_1) = \inf_{\pi\in\pi^{\alpha}(s_0,s_1)}\sum_{k,l} \pi_{kl}\ W_{2}^{2}(\mu_{0k}, \mu_{1l})
\end{align*}
}
where $\Pi = ((\pi_{kl}))$ and 
$
  \pi^{\alpha}(s_0,s_1) = \big\{ \Pi: \Pi 1_{K_1} = s_0,\ \Pi^{\T} 1_{K_0} = s_1,\ D_{\rm KL}(\Pi\ ||\ s_0 s_{1}^{\T}) \leq\alpha\big\}$.
\end{lemma}
\begin{proof}
The proof of the Lemma is adapted from Proposition 4 in \citet{delon:hal-02178204}, that proves the result for Gaussian mixture models.  We extend the result for our augmented and restricted class of Elliptical mixture models.

First, we assume $\pi^{\star}$ be a solution to the linear program
\begin{align}\label{th1:lemmma3_lp}
\inf_{\pi\in\pi^{\alpha}(s_0,s_1)}\sum_{k,l} \pi_{kl}\ W_{2}^{2}(\mu_{0k}, \mu_{1l}).    
\end{align}
Next, for every $(k,l)$ we denote the optimal coupling as
\begin{align*} 
\gamma_{kl}= \argmin_{\gamma\in\pi(\mu_{0k},\mu_{1l})} \int_{\mb{R}^d\times\mb{R}^d}\
\left\lVert y_0 - y_1 \right\rVert^2\ d\gamma(y_0,y_1). 
\end{align*}
By Lemma \ref{th1:lemma3}, $\gamma_{kl}$ belongs to the same elliptical family of distributions. Next, we construct $\gamma^{\star}=\sum_{k,l} \pi^{\star}_{kl}\ \gamma_{kl}$ . By construction,   $\gamma^{\star}\in\pi(p_0,p_1) \cap \EMM^{\alpha}_{2d}(K_0 K_1)$ and that trivially yields
\begin{align}\label{th1:lemma3_oneside} 
\sum_{k,l} \pi^{\star}_{kl}\ W_{2}^{2}(\mu_{0k}, \mu_{1l})
&=\int_{\mb{R}^d\times\mb{R}^d}\
\left\lVert y_0 - y_1 \right\rVert^2\ d\gamma^{\star}(y_0,y_1)\notag\\
&\geq \inf_{\gamma\in \pi(p_0,p_1) \cap \EMM^{\alpha}_{2d}(K_0 K_1) }\ \int_{\mb{R}^d\times\mb{R}^d}\
\left\lVert y_0 - y_1 \right\rVert^2\ d\gamma(y_0,y_1).
\end{align}

Next,  for every $\gamma=(\gamma_{(0)},\gamma_{(1)})\in\pi(p_0,p_1) \cap \EMM^{\alpha}_{2d}(K_0 K_1)$, there  exists a $K$ such that $\gamma = \sum_{j=1}^{K} w_j\gamma_j$ where $\gamma_j$ is from the same elliptical family of distribution in $\mb{R}^{2d}$. Further, since $\gamma$ has marginal distributions $p_0$ and $p_1$, we have
\begin{align*}
    &\int_{\mb{R}^d} \gamma d\gamma_0 = \sum_{j=1}^{K} w_j\int_{\mb{R}^d}\gamma_j d\gamma_{(0)} =  \sum_{k=1}^{K_0} s_{0k}  \mu_{0k},\quad \\
    &\int_{\mb{R}^d} \gamma d\gamma_1 = \sum_{j=1}^{K} w_j\int_{\mb{R}^d}\gamma_j d\gamma_{(0)} =  \sum_{k=1}^{K_1} s_{1k}  \mu_{1k}.
\end{align*}
By the definition of identifiability of elliptical mixture models in definition \ref{def:indentifiable_emm}, we know that  mixtures must have the same components. Consequently, there exists a $(k,l), 1\leq k\leq K_0, 1\leq l\leq K_1$ such that $\gamma_j\in\pi(\mu_{0l},\mu_{1k})$, and we can express $\gamma = \sum_{k,l}\pi_{kl} \gamma_{kl}$  for some $\pi\in\pi^{\alpha}(s_0,s_1)$. So, we have
\begin{align}\label{th1:lemma3_otherside}
\int_{\mb{R}^d\times\mb{R}^d}\
\left\lVert y_0 - y_1 \right\rVert^2\ d\gamma(y_0,y_1)
\geq \sum_{k,l} \pi_{kl} W_2^{2}(\mu_{0k},\mu_{1l}) 
\geq \sum_{k,l} \pi^{\star}_{kl} W_2^{2}(\mu_{0k},\mu_{1l}),
\end{align}
where the final inequality holds by equation \eqref{th1:lemmma3_lp}.
Equations \eqref{th1:lemma3_oneside}-\eqref{th1:lemma3_otherside} together complete the proof.
\end{proof}

\begin{lemma}\label{th1:lemma5}{
Suppose the conditions in Lemma \ref{th1:lemma1} and/or \ref{th1:lemma2} hold. Suppose we  $p_0 = \sum\limits_{k=1}^{K_0} s_{0k}  \mu_{0k}$ with $\mu_{0k}\sim\mbox{\rm ED}_h(m_{0k}, \Sigma_{0k})$, and $p_1 = \sum\limits_{k=1}^{K_1} s_{1k} \delta_{m_{1k}}$. Let $ p_0^{\prime} = \sum\limits_{k=1}^{K_0} s_{0k}\ \delta_{m_{0k}}$. Then
\begin{align*} 
\mbox{\rm MW}_{2,\alpha}^{2}(p_0, p_1) = W_{2}^{2}( p_0^{\prime}, p_1) + \nu_{h}\sum\limits_{k=1}^{K_0} s_{0k}\ \mbox{\rm tr}(\Sigma_{0k}).
\end{align*}
}
\end{lemma}

\begin{proof}
By Lemma \ref{th1:lemma4},
\begin{align*} 
&\mbox{\rm MW}_{2,\alpha}^{2}(p_0, p_1) = \inf_{\pi\in\pi^{\alpha}(s_0,s_1)}\sum_{k,l} \pi_{kl}\ W_{2}^{2}(\mu_{0k}, \delta_{m_{1l}}) \notag\\
&= \inf_{\pi\in\pi^{\alpha}(s_0,s_1)}\sum_{k,l} \pi_{kl}\big\{ ||m_{0k}-m_{1l}||^2 + \nu_h\ \mbox{tr}(\Sigma_{0k})\big\}
= \bigg[ \inf_{\pi\in\pi^{\alpha}(s_0,s_1)} \big<\Pi, M\big>\bigg] + \nu_{h}\sum\limits_{k=1}^{K_0} s_{0k}\ \mbox{tr}(\Sigma_{0k}),
\end{align*}
and we have the proof.
\end{proof}

\begin{lemma}\label{th1:lemma6}[Entropy regularization of discrete optimal transport]
For every $\alpha\in\mb{R}^{+}$, $\exists\ \lambda_{\alpha}\ > 0$ such that
\begin{align*} 
\inf_{\pi\in\pi^{\alpha}(s_0,s_1)} \big<\Pi, M\big> = \inf_{\pi\in\pi(s_0,s_1)} \bigg[\big<\Pi, M\big> - \frac{1}{\lambda_{\alpha}}H(\Pi)\bigg]
\end{align*}
where $H(\Pi) = -\sum_{k,l} \pi_{kl}\log \pi_{kl}$.
\end{lemma}
\begin{proof}
Proof of this result is recorded in \citet{cuturi2013sinkhorn}.
\end{proof}


\section{Additional simulation results for generalised linear regression in \ref{ssec:gen_reg}}\label{sup:glm}
In Table \ref{table:glm_1} in Section \ref{ssec:gen_reg} in the main document, we expand on the performance of $\mbox{D}$-BETEL for varying  extent of perturbations in the data generating mechanism with sample size $n=100$. Here we present additional simulation results for $n=250, 500$. In particular, we compare  $\mbox{D}$-BETEL against standard posterior based approach, as well as Bayesian ETEL \citep{Chib2018} with the estimating equations set to  $\mbox{E}[\partial \log l(\beta\mid X, Y)/ \partial \beta]=0$ to infer about the parameter $\beta$. 
\begin{table}[!htb]
  \caption{\emph{\textbf{Generalised linear regression (Poisson regression).} Here the \textbf{sample size $n$ is $250$}. We compare standard posterior yielded from the fully parametric model, moment conditional model (MCM) based on the maximum likelihood equations,
  and $\mbox{D}$-BETEL based parameter estimates over 50 replicated simulations with proportion of outlier $p=0.10, 0.12, 0.15$. 
  We report the $L_1$ error of posterior means, length of the HPD sets and associated coverage probabilities (within braces). $\mbox{D}$-BETEL is more resistant towards presence of outliers all values of $p$ considered, however it provides slightly wider $95\%$ credible sets while maintaining the high coverage probability. }}\label{table:glm_2}
  \centering
  \begin{tabular}{llllllll}
    \toprule
    \cmidrule(r){1-2}
     &     &       \multicolumn{2}{c}{$\mbox{D}$-BETEL} & \multicolumn{2}{c}{Standard posterior} & \multicolumn{2}{c}{MCM} \\
    \midrule
    p & $\theta$     & $||\theta - \hat{\theta}||_1$     & HPD  & $||\theta - \hat{\theta}||_1$     & HPD & $||\theta - \hat{\theta}||_1$     & HPD   \\
    \midrule
    0.10 & $\beta_0$ & 0.03 & 0.19 (1.00)  & 0.14   & 0.12  (0.24)  &0.31 & 0.14 (0.12)\\
        & $\beta_1$ & 0.01 & 0.03 (1.00)  & 0.06 &0.02 (0.00) 
        & 0.05& 0.02 (0.22)\\
    \midrule
    0.12 & $\beta_0$ &0.05 &0.20 (0.95) &0.24& 0.12 (0.16)
    &0.36 & 0.20 (0.14)\\
        & $\beta_1$ & 0.01 &0.04 (0.95) &0.06& 0.02 (0.11)
        &0.05 & 0.04 (0.18)\\
    \midrule
    0.15 & $\beta_0$ & 0.05& 0.30 (1.00)  &0.30  & 0.11 (0.20) & 0.35 & 0.13 (0.06)\\
        & $\beta_1$ & 0.01 &0.05 (1.00)  &0.05& 0.01 (0.20) & 0.05 & 0.02 (0.08)\\
    
    \bottomrule
  \end{tabular}
\end{table}

\begin{table}[!h]
  \caption{\emph{\textbf{Generalised linear regression (Poisson regression.)} Here the \textbf{sample size $n$ is $500$}. We compare standard posterior yielded from the fully parametric model, moment conditional model (MCM) based on the maximum likelihood equations,
  and $\mbox{D}$-BETEL based parameter estimates over 50 replicated simulations with proportion of outlier $p=0.10, 0.12, 0.15$. 
  We report the $L_1$ error of posterior means, length of the HPD sets and associated coverage probabilities (within braces). $\mbox{D}$-BETEL is more resistant towards presence of outliers all values of $p$ considered, however it provides slightly wider $95\%$ credible sets while maintaining the high coverage probability}.}\label{table:glm_3}
  \centering
  \begin{tabular}{lllllllll}
    \toprule
    \cmidrule(r){1-2}
     &     &       \multicolumn{2}{c}{$\mbox{D}$-BETEL} & \multicolumn{2}{c}{Standard posterior} & \multicolumn{2}{c}{MCM} \\
    \midrule
    p & $\theta$     & $||\theta - \hat{\theta}||_1$     & HPD  & $||\theta - \hat{\theta}||_1$     & HPD & $||\theta - \hat{\theta}||_1$     & HPD  \\
    \midrule
    0.10 & $\beta_0$ & 0.03 & 0.15 (0.91)  & 0.15   & 0.12  (0.47) & 0.24 & 0.15 (0.22)\\
        & $\beta_1$ & 0.01 & 0.02 (0.93) & 0.03 &0.02 (0.19) 
        & 0.03 & 0.03 (0.32)\\
    \midrule
    0.12 & $\beta_0$ &0.06 &0.19 (0.83) &0.19& 0.12 (0.20) & 0.24 & 0.23 (0.26)   \\
        & $\beta_1$ & 0.01 &0.04 (0.85) &0.03& 0.02 (0.35) & 0.03 & 0.04 (0.38) \\
    \midrule
    0.15 & $\beta_0$ & 0.08& 0.30 (0.92)  &0.21  & 0.12 (0.21) & 0.31 &0.16 (0.18)\\
        & $\beta_1$ & 0.02 &0.06 (0.92)  &0.03& 0.02 (0.25) & 0.04 & 0.03 (0.22)\\
    
    \bottomrule
  \end{tabular}
\end{table}

\end{document}